\documentclass[10pt,reqno]{amsart}

\usepackage[T1]{fontenc}
\usepackage[margin=1in]{geometry}

\usepackage{amsmath,amsthm,amssymb,amsfonts}
\usepackage{mathrsfs}
\usepackage{mathtools}
\usepackage{stmaryrd}
\usepackage{thmtools}

\usepackage{appendix}
\usepackage{booktabs}
\usepackage{enumitem}
\usepackage{fancyhdr}
\usepackage{float}

\usepackage{xcolor}
\usepackage{hyperref}
\definecolor{RRed}{HTML}{CD071E}
\definecolor{BBlue}{HTML}{0047AB}
\hypersetup{colorlinks=true,linkcolor=RRed,citecolor=BBlue}
\usepackage{pgfplots}
\usepackage{subfigure}
\usepackage{tikz}
\usepgfplotslibrary{fillbetween}

\declaretheorem[style=definition,qed=$\blacklozenge$,numberwithin=section,name=Definition]{definition}

\declaretheorem[style=plain,sibling=definition,name=Theorem]{theorem}
\declaretheorem[style=plain,sibling=definition,name=Lemma]{lemma}
\declaretheorem[style=plain,sibling=definition,name=Corollary]{corollary}

\makeatletter
\declaretheoremstyle[preheadhook=\renewcommand\@upn{},headfont=\itshape,qed=$\blacktriangle$]{example}
\declaretheorem[style=example,sibling=definition,name=Example]{example}

\makeatother

\newcommand{\ie}{\textit{i}.\textit{e}.}
\newcommand{\eg}{\textit{e}.\textit{g}.}

\newcommand{\ub}{\underline{b}}

\newcommand{\ul}{\underline{l}}

\newcommand{\uii}{\underline{i}}
\newcommand{\ujj}{\underline{j}}

\newcommand{\ux}{\underline{x}}
\newcommand{\uy}{\underline{y}}

\newcommand{\uL}{\underline{L}}
\newcommand{\uX}{\underline{X}}
\newcommand{\uY}{\underline{Y}}

\newcommand{\uone}{\underline{1}}

\newcommand{\ttd}{\mathtt{d}}

\newcommand{\tth}{\mathtt{h}}

\newcommand{\ttl}{\mathtt{l}}

\newcommand{\ttr}{\mathtt{r}}

\newcommand{\tG}{\mathtt{G}}

\newcommand{\tP}{\mathtt{P}}
\newcommand{\tR}{\mathtt{R}}

\newcommand{\EE}{\mathbb{E}}

\newcommand{\II}{\mathbb{I}}
\newcommand{\NN}{\mathbb{N}}
\newcommand{\PP}{\mathbb{P}}
\newcommand{\RR}{\mathbb{R}}
\newcommand{\TT}{\mathbb{T}}
\newcommand{\ZZ}{\mathbb{Z}}

\newcommand{\cC}{\mathcal{C}}
\newcommand{\cD}{\mathcal{D}}

\newcommand{\cI}{\mathcal{I}}

\newcommand{\cL}{\mathcal{L}}

\newcommand{\cO}{\mathcal{O}}
\newcommand{\cP}{\mathcal{P}}

\newcommand{\cX}{\mathcal{X}}
\newcommand{\cY}{\mathcal{Y}}

\newcommand{\rb}{\mathrm{b}}

\newcommand{\diff}{\mathrm{d}}
\newcommand{\ee}{\mathrm{e}}
\newcommand{\rB}{\mathrm{B}}

\newcommand{\rP}{\mathrm{P}}

\newcommand{\sH}{\mathsf{H}}

\newcommand{\bb}{\mathsf{b}}
\newcommand{\cc}{\mathsf{c}}

\newcommand{\xx}{\mathsf{x}}

\newcommand{\rrD}{\mathscr{D}}
\newcommand{\rrF}{\mathscr{F}}

\newcommand{\dCC}{\mathfrak{C}}
\newcommand{\error}{\mathfrak{E}}
\newcommand{\bhatta}{\mathfrak{B}}
\newcommand{\entropy}{\mathrm{H}}

\newcommand{\LDPC}{\mathrm{LDPC}}
\newcommand{\BP}{\mathrm{BP}}
\newcommand{\MAP}{\mathrm{MAP}}
\newcommand{\ML}{\mathrm{ML}}

\newcommand{\argmax}{\mathrm{argmax}}

\newcommand{\stab}{\mathrm{stab}}
\newcommand{\avg}{\mathrm{avg}}

\newcommand{\bec}{\mathrm{BEC}(\epsilon)}
\newcommand{\bscp}{\mathrm{BSC}(p)}
\newcommand{\bawgnc}{\mathrm{BAWGNC}(\sigma)}

\begin{document}


\title[Stability of Low-Density Parity-Check Codes]{The Stability of Low-Density Parity-Check Codes \\ and Some of Its Consequences}

\author[W.~Liu]{\footnotesize Wei Liu}
\address{\'{E}cole polytechnique f\'{e}d\'{e}rale de Lausanne, Switzerland}
\email{wei.liu@alumni.epfl.ch}

\author[R.~Urbanke]{R\"{u}diger Urbanke}
\address{\'{E}cole polytechnique f\'{e}d\'{e}rale de Lausanne, Switzerland}
\email{rudiger.urbanke@epfl.ch}


\begin{abstract}
We study the stability of low-density parity-check codes under blockwise or bitwise maximum \textit{a posteriori} decoding, where transmission takes place over a binary-input memoryless output-symmetric channel. Our study stems from the consideration of constructing universal capacity-achieving codes under low-complexity decoding algorithms, where universality refers to the fact that we are considering a family of channels with equal capacity. Consider, \eg, the right-regular sequence by Shokrollahi and the heavy-tail Poisson sequence by Luby \textit{et al}. Both sequences are provably capacity-achieving under belief propagation decoding when transmission takes place over the binary erasure channel. 

In this paper we show that many existing capacity-achieving sequences of low-density parity-check codes under belief propagation decoding are not universal under belief propagation decoding. We reveal that the key to showing this non-universality result is determined by the stability of the underlying codes. More concretely, for an ordered and complete channel family and a sequence of low-density parity-check code ensembles, we determine a stability threshold associated with them, which further gives rise to a sufficient condition such that the sequence is not universal under belief propagation decoding. Furthermore, we show that the same stability threshold applies to blockwise or bitwise maximum \textit{a posteriori} decoding as well. We then present how stability can determine an upper bound on the corresponding blockwise or bitwise maximum \textit{a posteriori} threshold, revealing the operational significance of the stability threshold.
\end{abstract}

\maketitle


\section{Introduction}
\label{sec:stab+intro}

\subsection{Motivation} 

Gallager introduced low-density parity-check (LDPC) codes in the 1960s~\cite{Gal60,Gal63}. However, due to the limitation of computational resources during that era, they were regrettably forgotten until the invention of Turbo codes~\cite{BGT93} and the rediscovery works by Mackay and Neal~\cite{MN96} as well as Luby \textit{et al}.~\cite{LMSSS97}. Since then, systematic developments of LDPC codes have been carried out by Luby \textit{et al}.~\cite{LMSS01a,LMSS01b} as well as Richardson and Urbanke~\cite{RU01a,RU01b}. Their capacity-approaching performances and low-complexity encoding and decoding algorithms have made LDPC codes of great practical interest and have become standardized in many communication protocols.

In their seminal paper~\cite{LMSSS97}, Luby \textit{et al}.~constructed the heavy-tail Poisson sequence (or its commercialized version, the Tornado sequence), and this sequence is well-known to achieve capacity for the binary erasure channel (BEC) under belief propagation (BP) decoding. Around the same time, Shokrollahi constructed the right-regular sequence~\cite{Sho99} and showed that this sequence is also capacity-achieving for the BEC under BP decoding. While it is relatively simple to show that these sequences are capacity-achieving for the BEC, it is not clear whether the same sequence can achieve the identical capacity for other important communication channels, \eg, the binary symmetric channel (BSC) and the binary additive white-Gaussian noise channel (BAWGNC), under BP decoding or even maximum \textit{a posteriori} (MAP) decoding. Perhaps the only codes known to date under the framework of LDPC codes that are provably capacity-achieving under BP decoding over the whole family of binary-input memoryless output-symmetric (BMS) channels with equal capacity are the sequence of spatially-coupled regular LDPC codes~\cite{KRU13}. This is possible because the sequence of uncoupled regular LDPC codes achieves capacity under MAP decoding. Basically, by increasing the ``density'' of the underlying parity-check matrices, the weight distributions of these regular LDPC codes behave more and more like completely random parity-check codes which are known to universally achieve capacity under MAP decoding. Consequently by the recent spatial coupling technique~\cite{KRU11}, one can universally achieve capacity under BP decoding. Other important contributions towards the design of such universal sequences of LDPC codes under BP decoding have been made in~\cite{SS11}, where the authors derived a lower bound on the universally achievable fraction of capacity under BP decoding, and linear programming techniques were employed in order to numerically find ``good'' degree distribution pairs. Moreover, it was shown in this paper that suitably chosen regular LDPC codes can achieve capacity under maximum likelihood (ML) decoding. 


\subsection{Main Results and Outline} 

The main results and outline of this paper are summarized as follows.

\begin{enumerate}[nosep,leftmargin=*]
\item[--] \textit{Universality}. We investigate the universality of LDPC codes under BP decoding. We show that many existing capacity-achieving sequences of LDPC codes designed for the BEC under BP decoding cannot achieve the same capacity for a broad family of BMS channels. The study of universality is carried out in Sections~\ref{sec:ca} and~\ref{sec:univ+bp}. More precisely, in Section~\ref{subsec:ca+exa} we investigate the heavy-tail Poisson sequence and the right-regular sequence and show a common characteristic shared between the two sequences. Such a characteristic is related to the fraction of degree-two variable nodes in the Tanner graph representation of the underlying code. Then in Section~\ref{subsec:ca+deg2} we determine bounds on the fraction of degree-two variable nodes for both sequences. Subsequently, in Section~\ref{subsec:univ} we determine a sufficient condition under which a capacity-achieving sequence for the BEC under BP decoding is not universal. In Section~\ref{subsec:poi+univ} we show that the heavy-tail Poisson sequence satisfies this condition, and thus is not universally capacity-achieving under BP decoding. Section~\ref{subsec:rr+univ} discusses the right-regular sequence and shows non-rigorously that the same conclusion also holds. Since these two sequences and many other capacity-achieving sequences are designed through the same principle, one can conclude that using such sequences cannot achieve the same capacity for other BMS channels under BP decoding.
\item[--] \textit{Stability}. We determine the stability threshold of LDPC codes under either blockwise or bitwise MAP decoding, when transmission takes place over a BMS channel. We present how stability can determine an upper bound on the corresponding blockwise or bitwise MAP threshold. The study of stability is carried out in Section~\ref{sec:stab+cond+map}. More precisely, in Section~\ref{subsec:exploration} we present an exploration process that reveals a particular type of neighboring structure around a generic variable node of degree two belonging to the underlying Tanner graph. Subsequently Section~\ref{subsec:cycle+deg2} shows  that, when the channel entropy is above the stability threshold, how this exploration process allows us to determine a global structure of the Tanner graph such that the channel realizations along this graphical structure satisfy a particular condition. As a consequence, Section~\ref{subsec:lbd+block+error+prob} shows that the stability threshold is an upper bound on the corresponding blockwise MAP threshold. Furthermore, in Section~\ref{subsec:lbd+bit+error+prob} we show how stability determines an upper bound on the corresponding bitwise MAP threshold. 
\end{enumerate} 


\subsection{Preliminaries}
\label{subsec:prelim}

In this section we give the preliminaries on the code and channel models and the reader is referred to~\cite{RU08} for an exposition of more details. We start by giving some useful notation and conventions which will be used throughout this paper. 

We denote by $\RR$ and $\ZZ$ the set of real numbers and integers, respectively. The set of natural numbers is denoted by $\NN \coloneqq \{i \in \ZZ | i \geqslant 1\}$. We use the shorthand notation $[n] \coloneqq \{i \in \NN \, | \, i \leqslant n\} = \{1,2,\dots,n\}$ for $n \in \NN$. For a scalar $x \in [0,1]$, very often we write $\bar{x}$ as a shorthand for $1-x$. In general we distinguish scalars from vectors; \eg, we write $x$ as a scalar and $\ux$ as a vector. For $\ux \in \RR^{n}$, we write $\ux_i$ as the $i$-th element of $\ux$. Let $\ee \coloneqq \lim_{n \to \infty} (1+1/n)^n$ be Euler's constant. The natural exponential function is defined by $x \mapsto \exp(x) \coloneqq \ee^x$. Its inverse function is the natural logarithmic function $x \mapsto \ln(x)$. We also adopt the Landau notation $\cO(f(x))$, $\varTheta(f(x))$ and $\varOmega(f(x))$ for a real-valued function $f$.


\subsubsection{Degree Distributions}

We first review the degree distributions associated with an ensemble of LDPC block codes. The number of \textit{variable} nodes is denoted by $n \in \NN$, among which the number of variable nodes of \textit{degree} $i$ is denoted by $\varLambda_{i,n}$, where $2 \leqslant i \leqslant \ttl$ for some $\ttl \in \NN$. We refer to $\ttl$ as the \textit{maximum} degree of variable nodes. The number of \textit{check} nodes is denoted by $m \in \NN$, among which the number of check nodes of \textit{degree} $i$ is denoted by $P_{i,m}$, where $2 \leqslant i \leqslant \ttr$ for some $\ttr \in \NN$. We refer to $\ttr$ as the \textit{maximum} degree of check nodes. When $\sum_{i=2}^{\ttl} i \varLambda_{i,n} = \sum_{i=2}^{\ttr} i P_{i,m}$, we call $n$ the \textit{blocklength}, write $P_{i,m}$ as $P_{i,n}$, and call $r_n = 1 - m/n$ the \textit{rate}. 

Let $L_{i,n} = \varLambda_{i,n} / n$ be the fraction of variable nodes of degree $i$. Then the variable \textit{degree distribution from a node perspective}, $L_n$, is defined via the generating function $L_n(x) = \sum_{i=2}^{\ttl} L_{i,n} x^i$. Let $R_{i,n} = P_{i,n} / (n (1 - r_n))$ be the fraction of check nodes of degree $i$. Then the check \textit{degree distribution from a node perspective}, $R_n$, is defined via the generating function $R_n(x) = \sum_{i=2}^{\ttr} R_{i,n} x^i$. The rate $r_n$ can be expressed as $1 - L_n ^{\prime}(1) / R_n ^{\prime} (1)$.

Given degree distributions $L_n$ and $R_n$, the standard approach to construct the underlying Tanner graph $\tG$ is the following. We regard each variable or check node of degree $i$ as consisting of $i$ half-edges and each half-edge has one socket. Label the half-edges belonging to variable nodes with the set of integers $\{1,\dots,nL^{\prime} _{n} (1)\}$ in an arbitrary way and do the same for half-edges belonging to check nodes. Let $\Pi$ be a uniform permutation on this set of integers, and connect the $i$-th half-edge on the variable node side to the $\Pi(i)$-th half-edge on the check node side. This is referred to as the \textit{configuration model} in the random graphs literature; see, \eg,~\cite{Bol01} for an exposition and Figure~\ref{fig:stab+ldpc} for an illustration. Very often in the sequel we assume that the maximum degrees $\ttl$ and $\ttr$ are fixed and do not depend on $n$. We write $L_n \longrightarrow_{\rrD} L$ to indicate that for every $i$, $L_{i,n}$ converges to a limit $L_i$ as $n$ tends to $\infty$. Let $L(x) = \sum_{i=2}^{\ttl} L_i x^i$. We refer to $L$ as the variable \textit{design} degree distribution from a node perspective, and call $\ttl_{\avg} = L^{\prime}(1)$ the \textit{average} variable-node degree. Similarly, we write $R_n \longrightarrow_{\rrD} R$ to indicate that for every $i$, $R_{i,n}$ converges to a limit $R_i$ as $n$ tends to $\infty$. Let $R(x) = \sum_{i=2}^{\ttr} R_i x^i$. We refer to $R$ as the check \textit{design} degree distribution from a node perspective, and call $\ttr_{\avg} = R^{\prime}(1)$ the \textit{average} check-node degree. The \textit{design} rate $r$ is defined as $r = 1 - \ttl_{\avg} / \ttr_{\avg}$. We also refer to $(L,R)$ as a design degree distribution \textit{pair} from a node perspective.

\begin{figure}[H]
\centering
\includegraphics{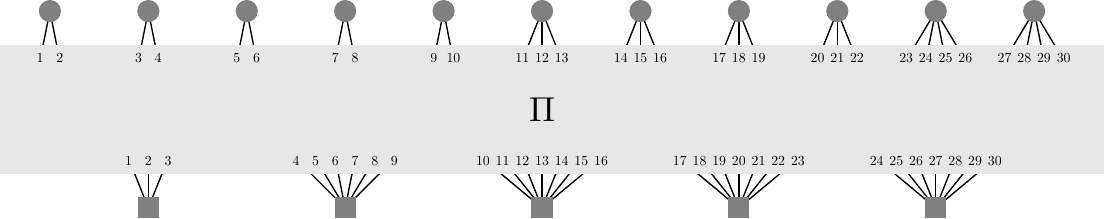}
\caption{Illustration of the configuration of half-edges via $\Pi$.}
\label{fig:stab+ldpc}
\end{figure}

Given a degree distribution $L_n$, we can construct the degree distribution $\lambda_n$ \textit{from an edge perspective} as follows. For every $i$, let $\lambda_{i,n} = i L_{i,n} / L_n ^{\prime}(1)$. Then the variable degree distribution from an edge perspective is defined via the generating function $\lambda_n(x) = \sum_{i=2}^{\ttl} \lambda_{i,n} x^{i-1}$. Similarly, for a degree distribution $R_n$, we can construct the degree distribution $\rho_n$ \textit{from an edge perspective} as follows. For every $i$, we define $\rho_{i,n} = i R_{i,n} / R_n ^{\prime}(1)$. Then the check degree distribution from an edge perspective is defined via the generating function $\rho_n(x) = \sum_{i=2}^{\ttr} \rho_{i,n} x^{i-1}$. Moreover, we denote by $\LDPC(\lambda_n,\rho_n)$ or $\LDPC(L_n, R_n)$ the ensemble of all Tanner graphs that are constructed according to the configuration model, and we use $\tG$ to indicate either a generic element from the ensemble or a random element endowed with a uniform distribution over the ensemble.

For fixed $\ttl$ and $\ttr$, $\lambda_n \longrightarrow_{\rrD} \lambda$ if and only if $L_n \longrightarrow_{\rrD} L$, where $\lambda_i = i L_i / L^{\prime}(1)$. We refer to $\lambda$ as the variable \textit{design} degree distribution from an edge perspective, and write $\lambda(x) = \sum_{i=2}^{\ttl} \lambda_i x^{i-1}$. Similarly, $\rho_n \longrightarrow_{\rrD} \rho$ if and only if $R_n \longrightarrow_{\rrD} R$, where $\rho_i = i R_i / R^{\prime}(1)$. We refer to $\rho$ as the check \textit{design} degree distribution from an edge perspective. We also refer to $(\lambda,\rho)$ as a design degree distribution \textit{pair} from an edge perspective. Furthermore, we use the subscript ``$_{(N)}$'' to denote a \textit{sequence} of design degree distribution pairs $\{(\lambda_{(N)}, \rho_{(N)})\}$ indexed by $N \in \NN$.


\subsubsection{BMS Channel Family}

We consider transmission over BMS channels, and denote by $p(y | x)$ the channel transition probability characterizing the BMS channel, where $x \in \cX = \{{-1}, 1\}$ and $y \in \cY \subseteq \RR$. An alternative description is through the $L$-\textit{distribution} which is the distribution of the log-likelihood ratio
\begin{equation}
\label{eqn:defn+llr}
L = l(Y) = \ln\big(p(Y|1) / p(Y|{-1})\big),
\end{equation}
conditioned on $X = 1$. The associated $L$-\textit{density} is denoted by $\cc$.

For an $L$-density $\cc$, the \textit{entropy}, \textit{Bhattacharyya}, and \textit{error probability} functionals are denoted by $\entropy(\cc)$, $\bhatta(\cc)$, and $\error(\cc)$, respectively, and they are given by
\begin{equation*}
\entropy(\cc) = \int_{-\infty}^{\infty} \cc(y) \log_2(1 + \ee^{-y}) \diff y; \text{   } \bhatta(\cc) = \int_{-\infty}^{\infty} \cc(y) \ee^{-y/2} \diff y; \text{   } \error(\cc) = \frac{1}{2} \int_{-\infty}^{\infty} \cc(y) \ee^{-(|y/2| + y/2)} \diff y.
\end{equation*}

\begin{example}
\label{exa:three+bms}
Let $\epsilon \in [0,1]$ and let $p \in [0,1/2]$. The $L$-densities for $\bec$ and $\bscp$ are denoted by 
\begin{equation*}
\mathsf{c}_{\epsilon}(y) = \epsilon \Delta_0(y) + \bar{\epsilon} \Delta_{\infty}(y) \text{ and } \mathsf{c}_p(y) = p \Delta_{\ln(p/\bar{p})}(y) + \bar{p} \Delta_{\ln(\bar{p}/p)}(y), 
\end{equation*}
respectively, where $\bar{\epsilon} = 1-\epsilon$ and $\bar{p} = 1-p$, and $\Delta_z$ is the distributional derivative of the Heaviside step function $\sH_z$ defined by $\sH_z(x) = \II_{{[z,\infty]}}(x)$. Furthermore, $\epsilon$ is referred to as the erasure parameter and $p$ is referred to as the flipping parameter. The $L$-density for $\bawgnc$ is
\begin{equation*}
\cc_{\sigma}(y) = \sqrt{\sigma^2/(8\pi)} \ee^{- (y - 2/\sigma^2)^2 \sigma^2/ 8}.
\end{equation*}
This is a Gaussian density with mean $2/\sigma^2$ and variance $4/\sigma^2$.
\end{example}

\begin{example}
\label{exa:three+functional}
The entropy, Bhattacharyya and error probability functionals at $\cc_{\epsilon}$ for the $\bec$ are given by $\entropy(\cc_{\epsilon}) = \epsilon$, $\bhatta(\cc_{\epsilon}) = \epsilon$ and $\error(\cc_{\epsilon}) = \epsilon/2$, respectively. These functionals evaluated at $\cc_p$ for the $\bscp$ are respectively given by $\entropy(\cc_p) = h_2(p)$, $\bhatta(\cc_{\epsilon}) = 2 \sqrt{p(1-p)}$ and $\error(\cc_p) = p$.
\end{example}

Very often in the sequel, instead of considering a single BMS channel, we will consider a \textit{family} of BMS channels. Under the $L$-distribution representation, we denote the family by $\{\cc_{\sigma}\}_{\sigma \in \cI}$, where $\sigma$ is a scalar parameter that characterizes the family and $\cI$ is the index set. The channel family $\{\cc_{\sigma}\}_{\sigma \in \cI}$ is said to be \textit{complete} if $\cI = [\underline{\sigma},\overline{\sigma}]$, $\entropy(\cc_{\underline{\sigma}}) = 0$ and $\entropy(\cc_{\overline{\sigma}}) = 1$, and for every $\tth \in [0,1]$, there exists a $\sigma \in \cI$ such that $\entropy(\cc_{\sigma}) = \tth$.

Frequently we need certain relationships between members in a channel family. A particularly important notion is \textit{degradation} and we typically consider channel families that can be partially ordered by degradation. The notation ``$\leftarrowtriangle$'' is used to indicate the relationship imposed by degradation. Specifically, $\cc \leftarrowtriangle \cc^{\prime}$ indicates that the BMS channel characterized by $\cc$ is \textit{degraded} with respect to the BMS channel characterized by $\cc^{\prime}$. Under the $|D|$-distribution representation~\cite{RU08}, the claim that $\cc$ with $|D|$-distribution $|\dCC|$ is degraded with respect to $\cc^{\prime}$ with $|D|$-distribution $|\dCC^{\prime}|$ is equivalent to
\begin{equation*}
\int_{z}^1 |\dCC|(x) \diff x \geqslant \int_{z}^1 |\dCC^{\prime}|(x) \diff x,
\end{equation*}
for all $z \in [0,1]$. A family of BMS channels $\{\cc_{\sigma}\}_{\underline{\sigma}} ^{\overline{\sigma}}$ is said to be \textit{ordered by degradation} if $\sigma \geqslant \sigma^{\prime}$ implies $\cc_{\sigma} \leftarrowtriangle \cc_{\sigma^{\prime}}$.  


\subsubsection{BP Decoding and MAP Decoding}

The generic transmission of information using an ensemble of LDPC block codes is the following. Given a degree distribution pair $(\lambda_n, \rho_n)$, let $\tG$ be the Tanner graph chosen uniformly at random from the $\LDPC(\lambda_n, \rho_n)$ ensemble. Then a codeword $\uX = \uX(\tG)$ is chosen uniformly at random and is transmitted through a BMS channel with transition $p_{Y|X}$. Denote by $\uY$ the channel output and let $\uL$ be the log-likelihood ratio, \ie, $\uL_i = l(\uY_i)$ for all $i \in [n]$. 

Let $\rP_{\rb} ^{\BP}(\tG, \cc, \ell, \ux)$ be the fraction of bits that are decoded incorrectly after $\ell$ iterations of BP decoding, conditioned on $\uX = \ux$. Then it holds that $\rP_{\rb} ^{\BP}(\tG, \cc, \ell, \ux) = \rP_{\rb} ^{\BP}(\tG, \cc, \ell)$ for all $\ux \in \cC(\tG)$. We further denote by $\EE_{\LDPC(\lambda_n, \rho_n)}[\rP_{\rb} ^{\BP}(\tG, \cc, \ell)]$ the average probability of bit error, where the randomness is averaged over the $\LDPC(\lambda_n, \rho_n)$ ensemble. Then for fixed $\ell \geqslant 1$,
\begin{equation}
\label{eqn:de+Pb}
\lim_{n \to \infty} \EE_{\LDPC(\lambda_n, \rho_n)}[\rP_{\rb} ^{\BP}(\tG, \cc, \ell)] = \cc \varoast L^{\varoast}\big( \rho^{\boxast}(\xx_{\ell-1}) \big),
\end{equation}
where $\rho^{\boxast}(\xx) \coloneqq \sum_{i=2}^{\ttr} \rho_i \xx^{\boxast (i-1)}$ and $L^{\varoast}(\xx) \coloneqq \sum_{i=2}^{\ttl} L_i \xx^{\varoast i}$. Here $\varoast$ is referred to as variable-node convolution and $\boxast$ is referred to as check-node convolution. The sequence of $L$-densities $\{\xx_{\ell}\}_{\ell \geqslant 0}$ in (\ref{eqn:de+Pb}) is defined by the recursion 
\begin{equation}
\label{eqn:density+evolution+def}
\xx_{\ell} = \TT(\xx_{\ell-1}) \coloneqq \cc \varoast \lambda^{\varoast} \big( \rho^{\boxast} (\xx_{\ell-1}) \big),
\end{equation} 
where $\xx_0 = \cc$ and $\lambda^{\varoast}(\xx) \coloneqq \sum_{i=2}^{\ttl} \lambda_i \xx^{\varoast (i-1)}$. We refer to (\ref{eqn:density+evolution+def}) as the \textit{density evolution} (DE) recursion associated with $\cc$ and $(\lambda,\rho)$, and we call $\TT$ the corresponding density evolution operator.

For an ordered and complete family of BMS channels $\{\cc_{\tth}\}$ indexed by entropy $\tth$ and a design degree distribution pair $(\lambda,\rho)$, let $\{\xx_{\ell}(\tth)\}$ be the sequence of $L$-densities corresponding to the density evolution operator associated with $\cc_{\tth}$ and $(\lambda, \rho)$. Then the \textit{BP threshold} of $\{\cc_{\tth}\}$ and $(\lambda,\rho)$ is defined as
\begin{equation*}
\tth^{\BP}(\{\cc_{\tth}\}, \lambda, \rho) = \sup\Big\{ \tth \,\Big|\, \lim_{\ell \to \infty} \error\big(\xx_{\ell}(\cc_{\tth})\big) = 0 \Big\}.
\end{equation*}

Finally we consider MAP decoding under the same setting. Let $\rP_{\rB} ^{\ML}(\cc_{\tth}, \lambda_n, \rho_n) = \EE_{\LDPC(\lambda_n, \rho_n)}[\rP^{\ML} _{\rB} (\tG, \tth)]$ be the average probability of block error under blockwise MAP decoding. The \textit{blockwise MAP threshold} of $\{\cc_{\tth}\}$ and $(\lambda, \rho)$ is defined as
\begin{equation}
\label{eqn:defn+blockmap}
\tth^{\ML}(\{\cc_{\tth}\}, \lambda, \rho) = \sup\Big\{ \tth \,\Big|\, \lim_{n \to \infty} \rP_{\rB} ^{\ML}(\cc_{\tth}, \lambda_n, \rho_n) = 0 \Big\}.
\end{equation}
Further let $\rP_{\rb} ^{\MAP}(\cc_{\tth}, \lambda_n, \rho_n) = \EE_{\LDPC(\lambda_n, \rho_n)}[\rP^{\MAP} _{\rb} (\tG, \tth)]$ be the average probability of bit error under bitwise MAP decoding. The \textit{bitwise MAP threshold} of $\{\cc_{\tth}\}$ and $(\lambda, \rho)$ is defined as
\begin{equation}
\label{eqn:defn+bitmap}
\tth^{\MAP}(\{\cc_{\tth}\}, \lambda, \rho) = \sup\Big\{ \tth \,\Big|\, \lim_{n \to \infty} \rP_{\rb} ^{\MAP}(\cc_{\tth}, \lambda_n, \rho_n) = 0 \Big\}.
\end{equation}
This concludes Section~\ref{subsec:prelim} on the preliminaries that we are going to use subsequently.

\section{Capacity-Achieving Sequences}
\label{sec:ca}
Capacity-achieving sequences do exist at least for transmission over the BEC, and up to now there are many explicit constructions in the literature and are provably capacity-achieving for transmission over the BEC. Two well-known such sequences, among others, are the heavy-tail Poisson (or Tornado) sequence~\cite{LMSSS97}, and the right-regular (or check-concentrated) sequence~\cite{Sho99}. Both sequences are provably capacity-achieving for the BEC with vanishing probability of bit error under BP decoding.


\subsection{Two Prototype Examples}
\label{subsec:ca+exa}

\begin{example}
\label{exa:ca+poi}
For $\epsilon \in (0,1)$, consider transmission on $\bec$. The heavy-tail Poisson sequence $\{(\lambda_{(N)},\rho_{(N)})\}$ in \cite{LMSS01a} satisfies for every integer $N \geqslant 2$ that
\begin{equation}
\label{eqn:poi+lambda+rho+def}
\lambda_{(N)}(x) = \frac{1}{H_{N-1}} \sum_{i=1}^{N-1} \frac{x^i}{i} \text{ and } \rho_{(N)}(x) = \ee^{(x-1)H_{N-1}/\epsilon},
\end{equation}
where $H_N$ denotes the $N$-th harmonic number.
\end{example}

We define the function $f_{(N)}: [0,1] \to [0,1]$ for every $N$ by setting $f_{(N)}(x) = \epsilon \lambda_{(N)}( 1 - \rho_{(N)}(1-x) )$. Figure~\ref{fig:univ+poi+f+f1} shows the characteristic behaviors of $\{f_{(N)}\}$ and $\{f_{(N)} ^{\prime}\}$ for the case $\epsilon=0.5$. 

\begin{figure}[H]
\centering
\subfigure{
\includegraphics{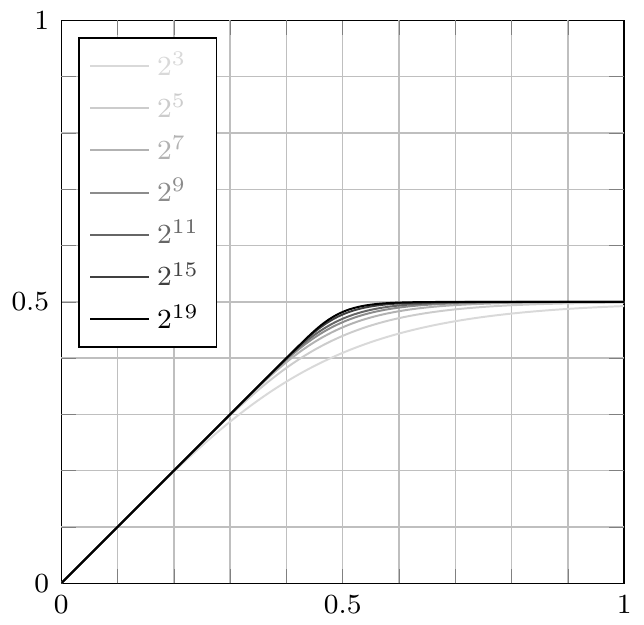}
}
\subfigure{
\includegraphics{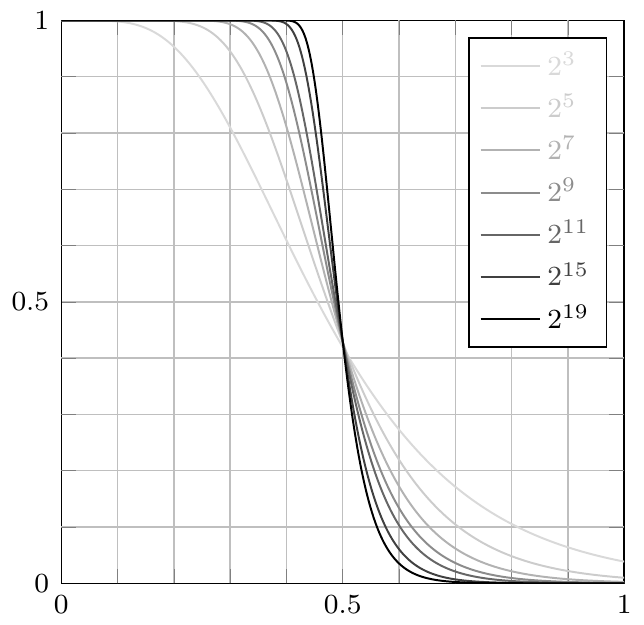}
}
\caption{$\{f_{(N)}\}$ and $\{f_{(N)} ^{\prime}\}$ for the heavy-tail Poisson sequence.}
\label{fig:univ+poi+f+f1}
\end{figure}

As we can observe from the left subfigure of Figure~\ref{fig:univ+poi+f+f1}, the sequence of functions $\{f_{(N)}\}$ approaches the ``limiting'' function $x \mapsto x \II_{\in {[0,\epsilon)}}(x) + \epsilon \II_{\in {[\epsilon,1]}}(x)$. As a consequence, the corresponding sequence of first-order derivatives changes ``quickly'' from $1$ to $0$ at the point $x = \epsilon$, as shown in the right subfigure of Figure~\ref{fig:univ+poi+f+f1}. The following Lemma~\ref{lem:poi+f1+limit} reveals that the sequence of first-order derivatives at $x = \epsilon$ converges to a limit, and the proof of this lemma is given in Appendix~\ref{appdix:inter+ca}.

\begin{lemma}
\label{lem:poi+f1+limit}
Given $\epsilon \in (0,1)$, let the heavy-tail Poisson sequence $\{(\lambda_{(N)}, \rho_{(N)})\}$ be specified according to (\ref{eqn:poi+lambda+rho+def}). Then the sequence of first-order derivatives $\{f_{(N)} ^{\prime}\}$ on $[0,1]$ satisfies
\begin{equation}
\label{eqn:poi+f1+limit}
\lim_{N \to \infty} f_{(N)} ^{\prime} (x) = \II_{[0,\epsilon)}(x) + (1 - \ee^{-\ee^{-\gamma}}) \II_{\{\epsilon\}}(x),
\end{equation}
where $\gamma$ denotes the Euler-Mascheroni constant.
\end{lemma}

Let us notice that, since $\gamma \approx 0.57722$, we can get $1 - \ee^{-\ee^{-\gamma}} \approx 0.42962$. This is also illustrated in the right subfigure of Figure~\ref{fig:univ+poi+f+f1}. As we shall see in later Sections~\ref{sec:univ+bp} and \ref{sec:stab+cond+map}, a critical quantity for our purpose is the first-order derivative evaluated at $x = 0$. In the case of the heavy-tail Poisson sequence, we can obtain for every $N \geqslant 2$ that $\lambda_{(N)} ^{\prime} (0) = 1 / H_{N-1}$ and $\rho_{(N)} ^{\prime}(1) = H_{N-1}/\epsilon$ so that $f_{(N)} ^{\prime}(0) = 1$.

\begin{example}
\label{exa:right+regular}
Consider transmission over $\bec$. The right-regular sequence $\{(\lambda_{(N)},\rho_{(N)})\}$ in \cite{Sho99} satisfies for every integer $N \geqslant 2$ that
\begin{equation}
\label{eqn:rr+lambda+rho+def}
\lambda_{(N)}(x) = \sum_{i=1}^{N-1} \frac{\binom{\alpha}{i}({-1})^{i-1}}{1 - \frac{N}{\alpha}\binom{\alpha}{N}({-1})^{N-1}} x^i \text{ and } \rho_{(N)}(x) = x^{1 / \alpha},
\end{equation}
where with $\bar{\epsilon} = 1 - \epsilon$,
\begin{equation}
\label{eqn:ca+rr+alpha+def}
\alpha = \alpha(\epsilon, N) \coloneqq \frac{\ln(1/\bar{\epsilon})}{\ln(N)} \text{ and } \binom{\alpha}{i} \coloneqq \frac{\prod_{j=0}^{i-1} (\alpha - j)}{i!},
\end{equation}
which is defined for every integer $i \geqslant 1$.
\end{example}

If we specify the function $f_{(N)} = \epsilon \lambda_{(N)}( 1 - \rho_{(N)}(1-x) )$ according to the right-regular sequence in (\ref{eqn:rr+lambda+rho+def}), then one can observe a similar behavior of $\{f_{(N)}\}$ as the heavy-tail Poisson sequence, \ie, the sequence of functions $\{f_{(N)}\}$ approaches $x \mapsto x \II_{\in {[0,\epsilon)}}(x) + \epsilon \II_{\in {[\epsilon,1]}}(x)$. Furthermore, we have the following lemma on the sequence of first-order derivatives at $x=0$, the proof of which is given in Appendix~\ref{appdix:inter+ca}.

\begin{lemma}
\label{lem:ca+rr+mu}
The right-regular sequence specified according to (\ref{eqn:rr+lambda+rho+def}) satisfies
\begin{equation}
\label{eqn:rr+flatness}
\lim_{N \to \infty} f_{(N)} ^{\prime}(0) = \lim_{N \to \infty} \epsilon \lambda_{(N)} ^{\prime}(0) \rho_{(N)} ^{\prime}(1) = 1.
\end{equation}
\end{lemma}


\subsection{Fraction of Degree-Two Variable Nodes}
\label{subsec:ca+deg2}

An important quantity showing up in (\ref{eqn:rr+flatness}) is $\lambda_2 ^{(N)}$, the fraction of degree-two variable nodes from an edge perspective, which is also equal to $\lambda^{\prime} _{(N)}(0)$. 

\begin{table}[H]
\centering
\begin{tabular}{ l c c }
$N$                  & right-regular $\lambda^{\prime} _{(N)} (0)$ & Poisson $\lambda^{\prime} _{(N)} (0)$ \\
\midrule[0.5pt]
$2^7$                & $0.2609$                                    & $0.1843$ \\
$2^9$                & $0.2073$                                    & $0.1467$ \\
$2^{11}$             & $0.1719$                                    & $0.1219$ \\
$2^{15}$             & $0.1280$                                    & $0.0911$ \\
$2^{19}$             & $0.1019$                                    & $0.0727$
\end{tabular}
\caption{The fraction $\lambda_{(N)} ^{\prime}(0)$ for the right-regular and heavy-tail Poisson sequence.}
\label{tab:inter+rr+poi+lambda2}
\end{table}

Table~\ref{tab:inter+rr+poi+lambda2} records this quantity for several design pairs $(\lambda_{(N)}, \rho_{(N)})$ from the right-regular and heavy-tail Poisson sequence, where $\epsilon=0.5$. As we can see from Table~\ref{tab:inter+rr+poi+lambda2}, the fractions of degree-two variable nodes from an edge perspective are decreasing for both sequences. In fact, it is not hard to show that the limit in both cases is $0$ as $N$ tends to $\infty$. On the other hand, the fraction of degree-two variable nodes from a node perspective remains large, as the next lemma illustrates.

\begin{lemma}
\label{lem:ca+frac+deg2}
Consider the heavy-tail Poisson sequence in (\ref{eqn:poi+lambda+rho+def}) and denote by $\{L_2 ^{(N)}\}$ the associated fractions of degree-two variable nodes from a node perspective. Then for all $\epsilon \in (0,1)$, it holds that
\begin{equation}
\label{eqn:ca+poi+l2}
\lim_{N \to \infty} L_2 ^{(N)} = \frac{1}{2}.
\end{equation}
Consider the right-regular sequence in (\ref{eqn:rr+lambda+rho+def}) and denote by $\{L_2 ^{(N)}\}$ the fractions of degree-two variable nodes from a node perspective. Then
\begin{equation}
\label{eqn:ca+rr+l2}
\frac{1}{\ln(9 \ee^{4/3} / 4)} \leqslant \liminf_{N \to \infty} L_2 ^{(N)} \leqslant \limsup_{N \to \infty} L_2 ^{(N)} \leqslant \frac{1}{\ln(9 \ee / 4)},
\end{equation}
which holds for all $\epsilon \in (0,1)$.
\end{lemma}

Notice that $1/\ln(9 \ee^{4/3} / 4) \approx 0.4664$ and $1/\ln(9 \ee / 4) \approx 0.5522$. Thus both sequences contain asymptotically a large fraction of degree-two variable nodes and these bounds hold uniformly for all $\epsilon \in (0,1)$; see Table~\ref{tab:inter+rr+poi+l2} for an illustration of the fractions of degree-two variable nodes from a node perspective for the right-regular sequence. The proof of Lemma~\ref{lem:ca+frac+deg2} is given in Appendix~\ref{appdix:ca+deg2}.

\begin{table}[H]
\centering
\begin{tabular}{ l c c c }
$N$                  & right-regular $L_2 ^{(N)}(0.3)$ & right-regular $L_2 ^{(N)}(0.5)$ & right-regular $L_2 ^{(N)}(0.7)$ \\
\midrule[0.5pt]
$2^7$                & $0.5396$                        & $0.5735$                        & $0.6253$ \\
$2^9$                & $0.5293$                        & $0.5561$                        & $0.5968$ \\
$2^{11}$             & $0.5236$                        & $0.5456$                        & $0.5790$ \\
$2^{15}$             & $0.5172$                        & $0.5333$                        & $0.5579$ \\
$2^{19}$             & $0.5135$                        & $0.5263$                        & $0.5457$
\end{tabular}
\caption{The fractions $L_2 ^{(N)}(\epsilon)$ for the right-regular sequence with $\epsilon \in \{0.3, 0.5, 0.7\}$.}
\label{tab:inter+rr+poi+l2}
\end{table}

\section{Universality under BP Decoding}
\label{sec:univ+bp}
Designing capacity-achieving codes under low-complexity encoding and decoding algorithms for different types of channels is of great practical importance, since in reality we are often faced with different physical constraints where different channel models are needed. Furthermore, from the viewpoint of complexity and performance tradeoff, we need degrees of freedom on the structure of the codes and it is preferable to have a design principle in order to systematically construct such codes. Since designing capacity-achieving sequences for the BEC under BP decoding is relatively tractable due the low-dimensionality of the density evolution recursion for the BEC, perhaps the simplest nontrivial solution of constructing universally capacity-achieving sequences under BP decoding is to directly employ these existing ones systematically designed for the BEC under BP decoding. This approach also represents the starting point of the investigation of this section. 

However, communication channels often exhibit extremal properties in one way or another. An important case for our purpose is the family of BMS channels having the same capacity. Consider, \eg, transmission over the $\bec$ or the $\bscp$ such that $\epsilon = \entropy(\cc_{\epsilon}) = \entropy(\cc_p) = h_2(p)$. We claim that $\bhatta(\cc_{\epsilon}) < \bhatta(\cc_p)$ and this can be seen as follows. Let $h_2 ^{-1} : [0,1/2] \to [0,1]$ denote the inverse of the binary entropy function, and define $b: [0,1] \to [0,1]$ by setting $b(x) = 2 \sqrt{x (1-x)}$. With $\epsilon = h_2(p)$, we can write $\bhatta(\cc_p) = b(h_2 ^{-1}(\epsilon))$. The function $x \mapsto b(h_2 ^{-1}(x))$ is strictly $\cap$-convex on $(0,1)$, and this can be verified by showing the function $x \mapsto h_2(b^{-1}(x))$ is strictly $\cup$-convex on $(0,1)$, where $b^{-1}: [0,1] \to [0,1/2]$ is the inverse of $b$.

\begin{lemma}
\label{lem:fix+H+max+B}
Consider the BEC with $L$-density $\cc_{\epsilon}$ and the BSC with $L$-density $\cc_p$. Let $\epsilon \in (0,1)$ and $p \in (0,1/2)$ be such that $\entropy(\cc_{\epsilon}) = \entropy(\cc_p)$, \ie, $\epsilon = h_2(p)$. Then it holds that $\bhatta(\cc_{\epsilon}) < \bhatta(\cc_p)$.
\end{lemma}

We recall from Section~\ref{subsec:ca+exa} that one common characteristic of the heavy-tail Poisson and the right-regular sequence designed for the $\bec$ is that they both satisfy the condition in (\ref{eqn:rr+flatness}). Consider now transmission over the BSC with capacity $1-\epsilon$. For the BEC having the same capacity, we have $\entropy(\cc_{\epsilon}) = \bhatta(\cc_{\epsilon}) = \epsilon$. By Lemma~\ref{lem:fix+H+max+B} we know that $\bhatta(\cc_p) > \epsilon$. Consequently for any sequence of design degree distribution pairs $\{(\lambda_{(N)},\rho_{(N)})\}_{N \in \NN}$ and for the $\bscp$, the condition in (\ref{eqn:rr+flatness}) implies that 
\begin{equation}
\label{eqn:bsc+fail}
\lim_{N \to \infty} \bhatta(\cc_p) \lambda_{(N)}^{\prime}(0) \rho_{(N)}^{\prime}(1) > 1.  
\end{equation}
As we shall see in Section~\ref{subsec:univ}, under further conditions on $\{(\lambda_{(N)},\rho_{(N)})\}$, the strict inequality in (\ref{eqn:bsc+fail}) is sufficient to make the probability of bit error under BP decoding be bounded away from $0$. We then use the result in Section~\ref{subsec:univ} to show why the heavy-tail Poisson sequence and the right-regular sequence are not universal under BP decoding; these are done in Sections~\ref{subsec:poi+univ} and \ref{subsec:rr+univ}, respectively. We give the proofs of lemmas and theorems of this section in Appendix~\ref{appdix:bp}.


\subsection{Generic Statement}
\label{subsec:univ}

We first recall that for a \textit{single} design degree distribution pair $(\lambda, \rho)$ and an $L$-density $\mathsf{c}$, we have $\Delta_{\infty} = \mathbb{T}(\Delta_{\infty})$, where $\mathbb{T} = \mathbb{T}(\mathsf{c}, \lambda, \rho)$ denotes the associated density evolution operator. We say that $\Delta_{\infty}$ is a fixed point of $\mathbb{T}$. Then a ``linearization'' of $\mathbb{T}$ around the fixed point $\Delta_{\infty}$ determines the stability under BP decoding. More precisely, if the condition $\mathfrak{B}(\mathsf{c}) \lambda^{\prime}(0) \rho^{\prime}(1) > 1$ holds, then there exists a strictly positive constant $\xi = \xi(\mathsf{c}, \lambda, \rho)$ such that $\liminf_{\ell \to \infty} \mathfrak{E}(\mathsf{x}(\ell)) > \xi$ for all $\mathsf{x}(0) \neq \Delta_{\infty}$; see~\cite{RU08} for more details. 

However, when we consider a \textit{sequence} of design degree distribution pairs, it is not clear \textit{a priori} whether or not the stability of the fixed point $\Delta_{\infty}$ is still determined by first-order approximation. Furthermore, the aforementioned lower bound on error probability $\xi$ implicitly depends on the sequence of degree distribution pairs, so that when we consider a sequence of degree distribution pairs, it is not clear \textit{a priori} the sequence of lower bounds is still bounded away from $0$.

The main result of this section is that we determine sufficient conditions on any sequence $\{(\lambda_{(N)},\rho_{(N)})\}$ of degree distribution pairs such that the stability of the fixed point $\Delta_{\infty}$ is still characterized by first-order approximation. To verify this, for a sequence $\{(\lambda_{(N)},\rho_{(N)})\}_{N \in \NN}$ and a BMS channel with $L$-density $\cc$, we let $\TT_{(N)}$ be the density evolution operator associated with the $N$-th pair $(\lambda_{(N)},\rho_{(N)})$ and $\cc$. For $\ell \geqslant 1$ and with $\xx_N (0)$ an $L$-density, we recursively define
\begin{equation*}
\xx_N(\ell) = \TT_{(N)}\big( \xx_N(\ell-1) \big).
\end{equation*}
Then we define a sequence of functions $\{g_{(N)}\}$ such that for each $N$, let $g_{(N)}: [0,1] \to [0,1]$ be defined as 
\begin{equation*}
g_{(N)}(x) = \lambda_{(N)} \big( 1 - \rho_{(N)}(1 - x)  \big).
\end{equation*}
Furthermore, for each $N \in \NN$, we let
\begin{equation*}
\mu_N \coloneqq g_{(N)} ^{\prime} (0) = \lambda_{(N)} ^{\prime} (0) \rho_{(N)} ^{\prime}(1).
\end{equation*}
The main result of this section is the following theorem, the proof of which is given in Appendix~\ref{appdix:generic}.

\begin{theorem}
\label{thm:univ+generic}
Let $\{(\lambda_{(N)}, \rho_{(N)})\}$ be a sequence of design degree distribution pairs such that the limit $\lim_{N \to \infty} \mu_N \eqqcolon \mu_{\infty}$ exists. Moreover, there exist a constant $\kappa \in (0,1)$ and an $M = M(\kappa) \in (0,\infty)$ such that the sequence of second-order derivatives of $\{g_{(N)}\}$ satisfies
\begin{equation}
\label{eqn:generic+univ+g2}
\big| g_{(N)} ^{\prime\prime} (x) \big| \leqslant M,
\end{equation}
for every $x \in [0,\kappa]$ and $N \geqslant 2$. Consider transmission over a BMS channel with $L$-density $\cc$ and let $\{\TT_{(N)}\}$ be the sequence of density evolution operators associated with $\cc$ and $\{(\lambda_{(N)}, \rho_{(N)})\}$. If 
\begin{equation*}
\lim_{N \to \infty} \bhatta(\cc) \mu_N > 1, 
\end{equation*}
then there exist a constant $\xi = \xi(\cc, \mu_{\infty}) > 0$ and an $N_{\infty} = N_{\infty}(\xi) \in \NN$ so that for all $N \geqslant N_{\infty}$, it holds that
\begin{equation*}
\liminf_{\ell \to \infty} \error\big( \xx_N(\ell) \big) \geqslant \xi,
\end{equation*}
where $\xx_N(0)$ is arbitrary and satisfies $\error(\xx_N(0)) \in {(0, \xi]}$.
\end{theorem}


\subsection{Heavy-Tail Poisson Sequence}
\label{subsec:poi+univ}

We start by recalling that, for transmission on $\bec$, the function $f_{(N)}$ on $[0,1]$ characterizing density evolution for $(\lambda_{(N)}, \rho_{(N)})$ is $f_{(N)}(x) = \epsilon \lambda_{(N)} ( 1 - \rho_{(N)}(1 - x) )$. In the case of the heavy-tail Poisson sequence, it simplifies to
\begin{equation}
\label{eqn:poi+f+def}
f_{(N)}(x) = \frac{\epsilon}{H_{N-1}} \sum_{i=1}^{N-1} \frac{1}{i} ( 1 - \ee^{- \alpha x} )^i.
\end{equation}

For $\epsilon \in (0,1)$, we have $\mu_N = 1/\epsilon \in \RR$ for all $N \in \NN$ so that $\mu_{\infty} \in \RR$ as well. The extra and technical condition that Theorem~\ref{thm:univ+generic} imposes is that the sequence of second-order derivatives of $\{f_{(N)}\}$ needs to be uniformly equibounded on some interval $[0,\kappa] \subseteq {[0,\epsilon)}$ with $0 < \kappa < \epsilon$. In the sequel we verify that this condition holds for the heavy-tail Poisson sequence. More precisely, the following Lemma~\ref{lem:poi+f2} characterizes the behavior of the sequence $\{f_{(N)} ^{\prime\prime}\}$. The proof of this lemma is given in Appendix~\ref{appdix:poi+univ}.

\begin{lemma}
\label{lem:poi+f2}
Given $\epsilon \in (0,1)$, let the heavy-tail Poisson sequence $\{(\lambda_{(N)}, \rho_{(N)})\}$ be specified in (\ref{eqn:poi+lambda+rho+def}). Then 
\begin{equation}
\limsup_{N \to \infty} f_{(N)} ^{\prime\prime} (\epsilon) = -\infty.
\end{equation}
On the other hand, let $\kappa \in [0,\epsilon)$ be arbitrary and fixed. Then for every $N \geqslant 3$ and $x \in [0,\kappa]$, it holds that
\begin{equation}
\label{eqn:poi+f2+lemma}
\big| f_{(N)} ^{\prime \prime} (x) \big| \leqslant \frac{2\ee^2}{\epsilon} N \ln(N) \exp\big(-\ee^{-(\gamma+1/2)\kappa/\epsilon}N^{1-\kappa/\epsilon}\big),
\end{equation}
\ie, the sequence of second-order derivatives of $\{f_{(N)}\}$ converges uniformly to $0$ on every closed interval in $[0,\epsilon)$. Therefore for every $\kappa \in [0,\epsilon)$, there exists an $M = M(\epsilon, \kappa) \in (0, \infty)$ such that for every $N \geqslant 2$ and every $x \in [0,\kappa]$, it holds that $|f_{(N)} ^{\prime \prime} (x)| \leqslant M$.
\end{lemma}

\begin{corollary}
Consider the heavy-tail Poisson sequence $\{(\lambda_{(N)}, \rho_{(N)})\}$ which is designed for the $\bec$ according to (\ref{eqn:poi+lambda+rho+def}). Then under BP decoding, it does not achieve capacity for any BMS channel with $L$-density $\cc$ such that $\bhatta(\cc) > \epsilon$.
\end{corollary}

Figure~\ref{fig:univ+poi+f2} shows the sequence of second-order derivatives of $\{f_{(N)}\}$ in (\ref{eqn:poi+f+def}) for the cases $\epsilon \in \{0.3,0.5,0.7\}$. As we can see, such curves are close to $0$ when $x$ is bounded away from $\epsilon$, while when $x$ is close to $\epsilon$, the sequence of such curves takes negative and decreasing values.

\begin{figure}[H]
\centering
\subfigure{
\includegraphics{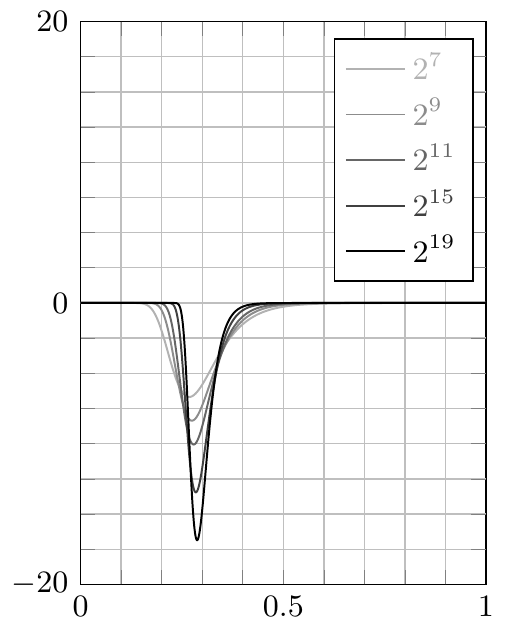}
}
\subfigure{
\includegraphics{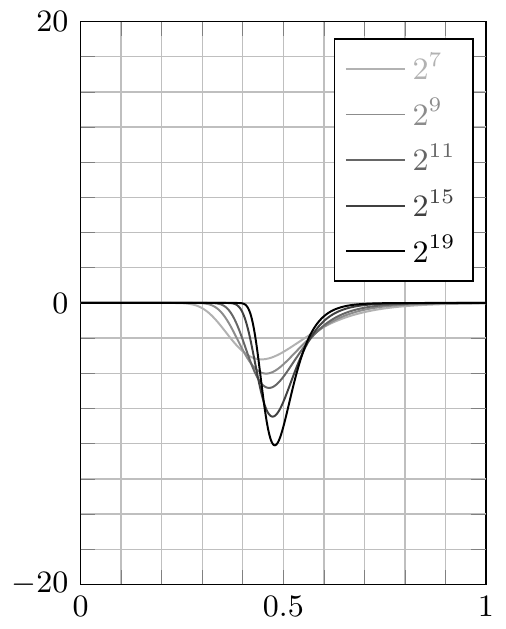}
}
\subfigure{
\includegraphics{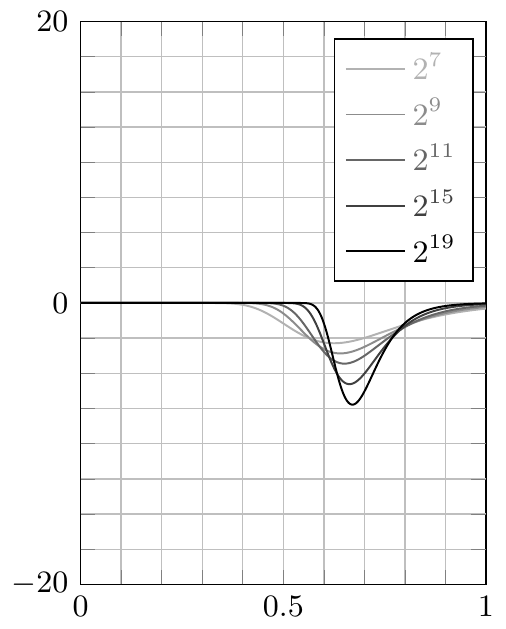}
}
\caption{$\{f_{(N)} ^{\prime\prime}\}$ of the heavy-tail Poisson sequence for $\epsilon \in \{0.3, 0.5, 0.7\}$.}
\label{fig:univ+poi+f2}
\end{figure}


\subsection{Right-Regular Sequence}
\label{subsec:rr+univ}

In this section we discuss universality of the right-regular sequence. Rather than rigorously proving that this sequence is not universal under BP decoding, we proceed informally and show by illustrations why the conditions in Theorem~\ref{thm:univ+generic} do hold under the current setting. We recall that the right-regular sequence $\{(\lambda_{(N)}, \rho_{(N)})\}$ is specified in Example~\ref{exa:right+regular}. With
\begin{equation}
\label{eqn:rr+f+def}
f_{(N)}(x) = \epsilon \lambda_{(N)} \big( 1 - (1 - x)^{\ttr_N} \big),
\end{equation}
where $\ttr_N = \ttr(\epsilon, N) = \log_{1/\bar{\epsilon}}(N)$ and $\lambda_{i+1} ^{(N)}$ is given in (\ref{eqn:rr+lambda+rho+def}), in the following we show that the sequence $\{f_{(N)} ^{\prime\prime}\}$ of second-order derivatives exhibits a similar behavior as the one for the heavy-tail Poisson sequence. More precisely, the following Lemma~\ref{lem:rr+f2+formula} characterizes the sequence of second-order derivatives of $\{f_{(N)}\}$ in (\ref{eqn:rr+f+def}). The proof of this lemma is given in Appendix~\ref{appdix:rr+univ}.

\begin{lemma}
\label{lem:rr+f2+formula}
Given $\epsilon \in (0,1)$, consider the right-regular sequence $\{(\lambda_{(N)}, \rho_{(N)})\}$ which is specified according to  (\ref{eqn:rr+lambda+rho+def}). For every integer $N \geqslant 3$, define functions $F_{(N)} ^0$ and $F_{(N)} ^1$ over $[0,1]$ as
\begin{align}
\label{eqn:rr+F0+def}
F_{(N)} ^0 (x) &= - \lambda_2 ^{(N)} (\ttr_N - 1) ; \\
\label{eqn:rr+F1+def}
F_{(N)} ^1 (x) &= \sum_{i=2}^{N-1} i \lambda_{i+1} ^{(N)} \big( 1 - (1-x)^{\ttr_N} \big)^{i-2} \big( (i\ttr_N - 1)(1-x)^{\ttr_N} - \ttr_N + 1 \big).
\end{align}
Then the second-order derivative of each $f_{(N)}$ in (\ref{eqn:rr+f+def}) can be expressed as
\begin{equation}
\label{eqn:rr+F012}
f_{(N)} ^{\prime \prime} (x) = \epsilon \ttr_N (1-x)^{\ttr_N - 2} \big( F_{(N)} ^0 (x) + F_{(N)} ^1 (x) \big).
\end{equation}
\end{lemma}

We notice that the second-order derivative of each $f_{(N)}$ at $x=0$ is $0$. Indeed, by (\ref{eqn:rr+F0+def}) and (\ref{eqn:rr+F1+def}) we have 
\begin{equation}
\label{eqn:rr+F0+F1}
F^0 _{(N)}(0) = - \lambda_2 ^{(N)} (\ttr_N - 1) \text{ and } F^1 _{(N)}(0) = 2 \lambda_3 ^{(N)} \ttr_N.
\end{equation}
Now if we let $\lambda_{\alpha} ^{(N)} = 1 - (N/\alpha) \binom{\alpha}{N} (-1)^{N-1}$ and use the fact that $\alpha \ttr_N = 1$, then by (\ref{eqn:rr+lambda+rho+def}) and (\ref{eqn:rr+F0+F1}) we can obtain
\begin{align*}
F^0 _{(N)}(0) + F^1 _{(N)}(0) &= \big( -\alpha (\ttr_N - 1) - 2 \alpha(\alpha-1)/2 \ttr_N \big) \big/ \lambda_{\alpha} ^{(N)} \\
&= ( -1 + \alpha + 1 - \alpha ) / \lambda_{\alpha} ^{(N)} = 0.
\end{align*}

Figure~\ref{fig:univ+rr+f2+component} illustrates (\ref{eqn:rr+F012}) for the case $\epsilon = 0.5$. As we can observe in Figure~\ref{fig:univ+rr+f2+component}, when $x$ is below a certain value, the contribution to ``$f^{\prime\prime}$'' from ``$F^0$,'' which is shown in the left subfigure of Figure~\ref{fig:univ+rr+f2+component}, is almost identical to the negative of the contribution to ``$f^{\prime\prime}$'' from ``$F^1$,'' which is shown in the middle subfigure of Figure~\ref{fig:univ+rr+f2+component}. Moreover, when $x$ is close to $\epsilon$, the only contribution to ``$f^{\prime\prime}$'' are those from ``$F^1$.'' 

\begin{figure}[H]
\centering
\subfigure{
\includegraphics{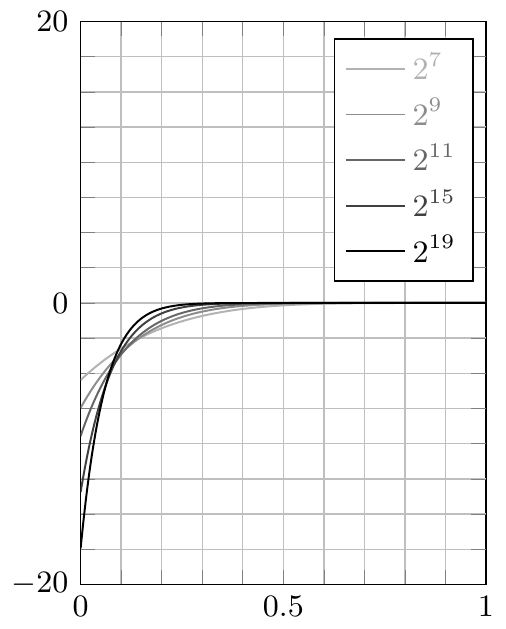}
}
\subfigure{
\includegraphics{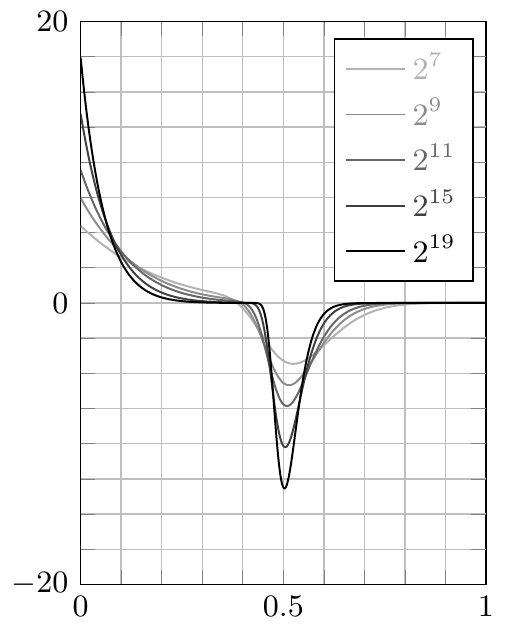}
}
\caption{Illustration of (\ref{eqn:rr+F012}) for the right-regular sequence with $\epsilon=0.5$.}
\label{fig:univ+rr+f2+component}
\end{figure}

Figure~\ref{fig:univ+rr+f2} plots the sequence $\{f_{(N)} ^{\prime\prime}\}$ of second-order derivatives in (\ref{eqn:rr+F012}) for the cases where $\epsilon \in \{0.3,0.5,0.7\}$. As we can see from Figure~\ref{fig:univ+rr+f2}, these curves are extremely close to $0$ when $x$ is bounded away from $\epsilon$, which suggests that in this case the sequence of second-order derivatives is indeed uniformly equibounded on some closed interval in ${[0,\epsilon)}$, and thus the conditions in Theorem~\ref{thm:univ+generic} applies to the right-regular sequence as well.

\begin{figure}[H]
\centering
\subfigure{
\includegraphics{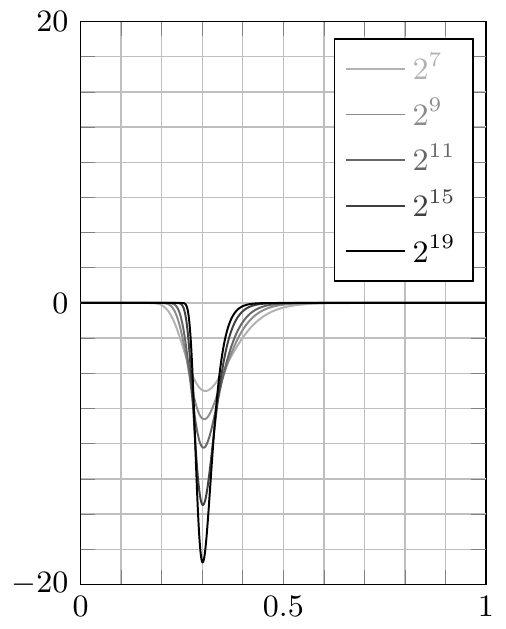}
}
\subfigure{
\includegraphics{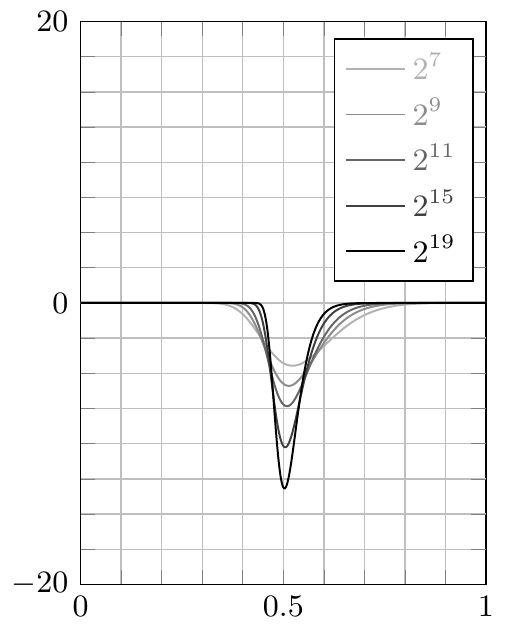}
}
\subfigure{
\includegraphics{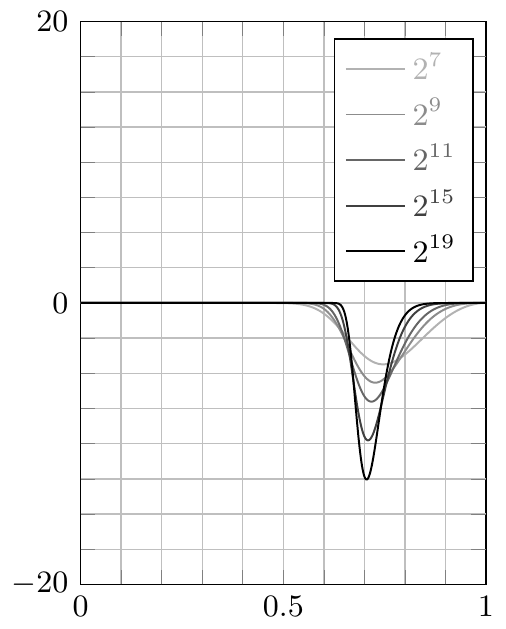}
}
\caption{$\{f_{(N)} ^{\prime\prime}\}$ of the right-regular sequence for $\epsilon \in \{0.3, 0.5, 0.7\}$.}
\label{fig:univ+rr+f2}
\end{figure}

We close this section by pointing out that for LDPC codes, one way of designing capacity-achieving sequences for the BEC is by the \textit{matching condition}~\cite{Sho01}. The first one of this matching condition corresponds exactly to the equality in (\ref{eqn:rr+flatness}), and in particular both heavy-tail Poisson sequence and right-regular sequence satisfy this condition. In fact, it was proved in~\cite{Sho01} that any capacity-achieving sequence of LDPC codes for the BEC under BP decoding necessarily satisfy this condition. We leave it as future work to investigate the universality of other capacity-achieving sequences under BP decoding such as the one proposed in \cite{SB10}, which will give us a better understanding of design principles on how to construct universal capacity-achieving sequences under BP decoding.

\section{Stability under MAP Decoding}
\label{sec:stab+cond+map}
For a design degree distribution pair $(\lambda,\rho)$, the quantity $\lambda^{\prime}(0) \rho^{\prime}(1)$ captures many important characteristics of the underlying LDPC codes as well as their encoding operations. More precisely, if $\lambda^{\prime}(0) \rho^{\prime}(1) > 1$, then such codes are linear-time encodable with high probability \cite{RU01b}, while the minimum distance grows sublinearly with respect to the blocklength; see, \eg, \cite{DRU06} for more details. Meanwhile, capacity-achieving LDPC codes under BP decoding known to date necessarily satisfy this condition \cite{Sho01}. On the other hand, if $\lambda^{\prime}(0) \rho^{\prime}(1) < 1$, then with probability $\sqrt{1 - \lambda^{\prime}(0) \rho^{\prime}(1)}$ the minimum distance grows linearly with respect to the blocklength; see, \eg, \cite{DRU06} for more details.

The quantity $\lambda^{\prime}(0) \rho^{\prime}(1)$ is also crucial for the analysis of BP decoding and as we have seen in Section~\ref{sec:univ+bp}, it is a key component for the universality result of modern capacity-achieving sequences of LDPC codes. In this section we connect the quantity $\lambda^{\prime}(0) \rho^{\prime}(1)$ to the analysis of blockwise MAP or bitwise MAP decoding when transmission takes place over a BMS channel. Our main result of this section depends on the following definition of stability threshold.

\begin{definition}
\label{def:stab}
Given a design degree distribution pair $(\lambda,\rho)$ and an ordered and complete BMS channel family $\{\cc_{\tth}\}$ indexed by entropy $\tth$, the \textit{stability threshold} $\tth^{\stab} = \tth^{\stab}(\{\cc_{\tth}\}, \lambda, \rho)$ associated with $\{\cc_{\tth}\}$ and $(\lambda,\rho)$ is defined as
\begin{equation*}
\tth^{\stab} = \inf \big\{\tth \in [0,1] \, \big| \, \bhatta(\cc_{\tth}) \lambda^{\prime}(0) \rho^{\prime}(1) > 1 \big\}.
\end{equation*}
If there is no $\tth \in [0,1]$ such that the condition is fulfilled, $\tth^{\stab}$ is defined to be $1$.
\end{definition}

Here is a discussion of the notation that we will use in the rest of this section. When we speak of a \textit{single design} degree distribution, we typically write $\lambda = \lambda_{(N)}$ or $\rho = \rho_{(N)}$ for some $N \in \NN$ and do not explicitly mention the underlying sequence of design degree distribution pairs that $(\lambda,\rho)$ belongs to. This will keep the notational burden to a minimum. We also recall from Section~\ref{subsec:prelim} that the degree distribution pairs $(\lambda_n, \rho_n)$ are such that $\lambda_n \longrightarrow_{\rrD} \lambda$ and $\rho_n \longrightarrow_{\rrD} \rho$. Here the subscript ``$_n$'' stands for the blocklength and $\LDPC(\lambda_n,\rho_n)$ denotes the ensemble of Tanner graphs constructed according to the configuration model. Moreover, $\ttl$ is the maximum variable-node degree and $\ttr$ is the maximum check-node degree.

The organization of this section is the following. In Section~\ref{subsec:exploration} we describe an exploration process that reveals the randomness of channel log-likelihood ratios as well as the connections of half-edges of variable nodes and check nodes in the Tanner graph. Basically, the exploration process reveals a specific graphical structure around a generic variable node in the Tanner graph. Notice that the order of the revelation in such an exploration does not corresponds to the actual transmission, but is very crucial for our analysis of blockwise or bitwise MAP decoding. Then in Section~\ref{subsec:cycle+deg2} we analyze the exploration process, where we will see an operational significance of the stability threshold defined in Definition~\ref{def:stab}. More precisely, the exploration process behaves in a fundamentally different way, depending on whether the channel entropy is larger than or smaller than the corresponding stability threshold. Subsequently in Section~\ref{subsec:lbd+block+error+prob}, we show that the stability threshold is an upper bound on the corresponding blockwise MAP threshold that was defined in (\ref{eqn:defn+blockmap}). Finally in Section~\ref{subsec:lbd+bit+error+prob}, we connect the stability threshold to the corresponding bitwise MAP threshold which was defined in (\ref{eqn:defn+bitmap}). More concretely, we present how the stability threshold determines an upper bound on the bitwise MAP threshold, and to illustrate the idea, we concentrate on transmission over the BSC. We then use a particular type of symmetry for a generic code in the ensemble to transform the original problem to the problem of finding a maximum matching for a suitably constructed graph, and we present how the cardinality of the maximum matching is related to bit error probability under bitwise MAP decoding, which concludes this section. 


\subsection{Exploration Process}
\label{subsec:exploration}

In the exploration process we will start with an arbitrarily chosen variable node of degree two, call it $v$, and then reveal the connections of edges as well as the channel log-likelihood ratios. The exploration is performed in discrete \textit{stage}s indexed by $k \in \NN$ and each stage consists of $l$ discrete \textit{step}s indexed by $t \in \NN$, where $l \in \NN$ is a prescribed parameter that depends on the underlying BMS channel and the design degree distribution pair $(\lambda,\rho)$. As we explore the neighborhoods of $v$, we label the half-edges of variable and check nodes as \textit{explored}, \textit{neutral}, \textit{open} or \textit{active}. Moreover, a half-edge of a variable or check node will be called \textit{unexplored} if its status is either neutral, or open, or active. At any particular stage or step, any half-edge only has one status. Initially, all half-edges are neutral and we only observe the node degrees and the half-edges, but not how these half-edges are connected. The following gives the precise description of the exploration process.

\begin{enumerate}[label*={\scshape Stage \arabic*},leftmargin=*,align=left]
\item[\refstepcounter{enumi}\scshape Stage $0$] Reveal the channel log-likelihood ratio $L_v$ associated with variable node $v$. Denote by $e_1$ and $e_2$ the two variable half-edges associated with $v$. Now uniformly at random choose a check half-edge and connect it to $e_1$ and change the status of these two half-edges from neutral to explored. Also change the status of the other check half-edges of this particular check node from neutral to active; we also mark this newly explored check node as \textit{active}. This corresponds to the initialization stage $k=0$ of the exploration process.
\item[\refstepcounter{enumi}\scshape Stage $k$] At stage $k \geqslant 1$, as long as there is any active check half-edge remaining at stage $k-1$, the exploration takes the following steps. 
\begin{enumerate}[label={\scshape Step \arabic*},nosep,leftmargin=*,align=left]
\item[\refstepcounter{enumi}\scshape Step\;$0$] Select the active check half-edge according to the breadth-first order.
\item[\refstepcounter{enumi}\scshape Step $t$] At step $t \in [l] \coloneqq \{1,\dots,l\}$, as long as there is any active check half-edge that is remaining at step $t-1$, the exploration is performed according to the following description.
\begin{enumerate}[label={\arabic*)},leftmargin=*]
\item Pick an active check half-edge according to the breadth-first order and then connect it to an unexplored variable half-edge uniformly at random. If this variable half-edge is neutral and belongs to a variable node of degree two, then the status of this variable half-edge and the active check half-edge are changed to explored, and then we change the status of the remaining variable half-edge of the same variable node of degree two to active. Otherwise, \ie, if this variable half-edge is either open or active, or if this variable half-edge belongs to a variable node of degree strictly larger than two, then we change the status of this variable half-edge and the active check half-edge to explored, and change the status of the remaining variable half-edges of the same variable node to open as well. 
\item As long as there is any variable half-edge that is active, we pick an variable half-edge according to the breadth-first order and then we uniformly at random connect it to an unexplored check half-edge. If this check half-edge is neutral and the variable half-edge is active, then the status of these two half-edges are changed to explored, and we also change the status of the remaining check half-edges of the same check node to active. Otherwise, that is, if this check half-edge is either open or active, then the status of these two half-edges are changed to explored, and we also change the status of the remaining check half-edges of the same check node to open. Now we mark a check node as active if and only if it has active check half-edges.
\item If $t = l$, an additional step is performed. Specifically, we reveal all channel log-likelihood ratios of the variable nodes between step $0$ and step $l$. Notice that for any given active check half-edge after performing step $l$, there exists a unique path of length $2l$, call it $\cP$, involving active check nodes only. We then sum up the log-likelihood ratios along this path. If the sum is strictly negative, the status of that particular check half-edge is kept unchanged. If the sum is strictly positive, the status of that check half-edge is changed to open. Otherwise, we independently flip a coin with uniform probability and let $B_{\cP} \in {\{0,1\}}$ denote the result to decide whether to change its status or not.
\end{enumerate}
\end{enumerate}

We continue this procedure until we have used $\varepsilon n$ check half-edges, where $\varepsilon$ is a prescribed parameter that depends on the underlying BMS channel and the pair $(\lambda,\rho)$. At this point we start to reveal the edge connection  to the variable half-edge $e_2$ of $v$ in the standard way. In this case we will not use labels anymore. If we have prematurely used up active check half-edges, meaning that the number of active check half-edges becomes zero before the total number of explored check half-edges reaches $\varepsilon n$, then we pick a degree-two variable node with neutral half-edges only and then restart the exploration process from {\scshape Stage 0}.
\end{enumerate}

This puts an end to the exploration process. Let us remark the reader that the only randomness in the exploration process are those due to channel fluctuations and the Tanner graph generation. Therefore if the Tanner graph in the ensemble and the channel realizations are fixed and if we run the exploration process starting from $v$ several times, then the orders of the half-edges that the process have revealed are the same.

Figure~\ref{fig:stab+onestage} exhibits an illustration of the ``one-step construction.'' In this example $l=3$ so that it consists of ``three steps'' that reveal the local neighborhoods of a check half-edge in the breadth-fist order. At step $t=l$ the channel log-likelihood ratios are revealed and among all three paths revealed just after step $3$, there are two paths such that each path contains $3$ variable nodes and the sum of log-likelihood ratios along these $3$ variable nodes is negative.

\begin{figure}[H]
\centering
\includegraphics{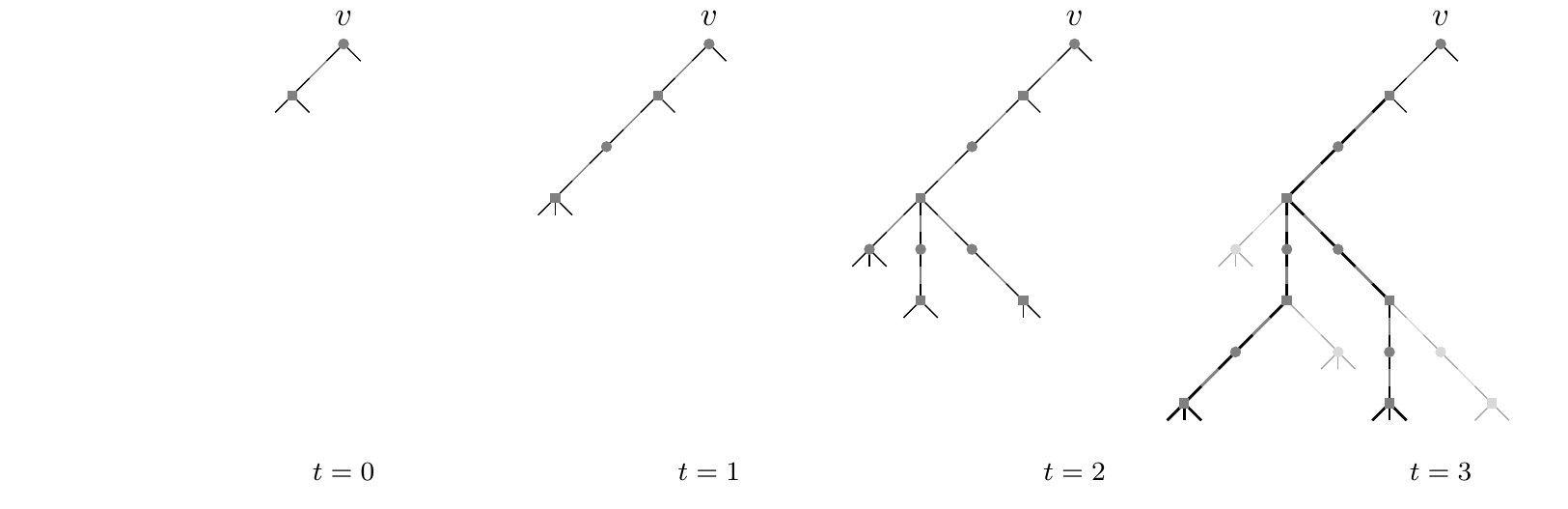}
\caption{Illustration of ``one-stage'' in the exploration process.}
\label{fig:stab+onestage}
\end{figure}

Figure~\ref{fig:stab+multistages} shows an illustration of ``multiple stages'' in the exploration process. In this illustration there are in total $k=13$ stages and those paths where the sum of log-likelihood ratios is negative are highlighted; \eg, there are $2$ such new paths revealed in stage $12$ and $3$ such new paths revealed in stage $13$. After these $13$ stages of exploration, there are in total $10$ paths, each of which connects the variable half-edge $e_1$ to an active check half-edge. Here we recall that $e_1$ is one of the two half-edges belonging to the variable node $v$.

\begin{figure}[H]
\centering
\includegraphics{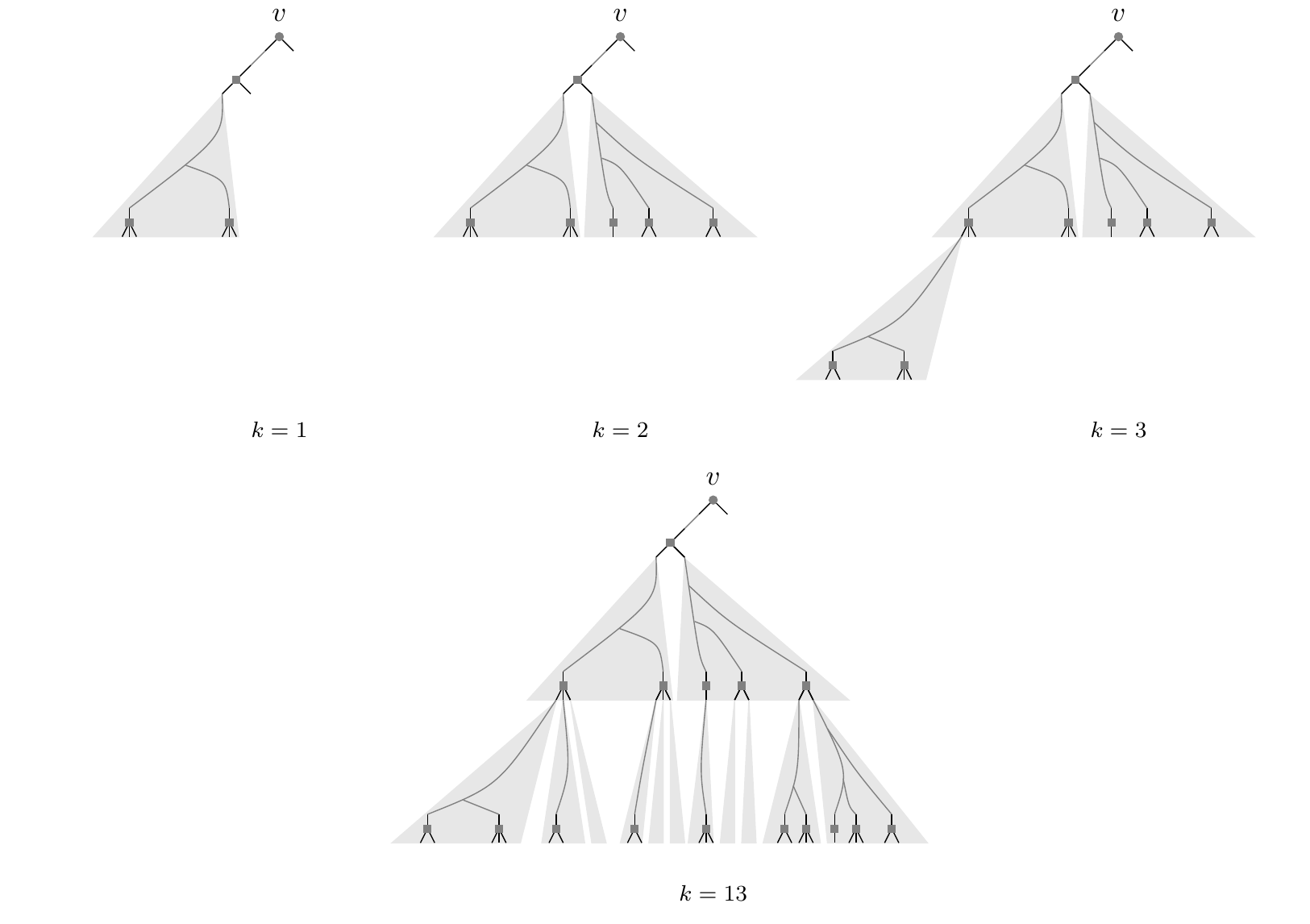}
\caption{Illustration of ``multiple stages'' in the exploration process.}
\label{fig:stab+multistages}
\end{figure}

In the next section we analyze the exploration process, where we will see, after performing a certain stage and depending on how large the channel entropy is, how the number of active check half-edges behaves. This has an important consequence for the analysis of MAP decoding.  


\subsection{Cycles with Degree-Two Variable Nodes}
\label{subsec:cycle+deg2}

We define a stochastic process $(A_k)_{k \geqslant 0}$ characterizing the exploration process that evolves in stage, where $A_k$ is the number of active check half-edges at stage $k$. In particular, at stage $0$ and for each $i$, $A_0$ is equal to $i-1$ with probability $\rho_{i,n}$. By the construction of the exploration process in Section~\ref{subsec:exploration}, at stage $k \geqslant 1$, \textit{either} there exists at least one active check half-edge remaining at stage $k-1$, \ie, $A_{k-1} \geqslant 1$, so that we perform $l$ steps of exploration, \textit{or} $A_{k-1} = 0$ so that we will immediately restart the exploration process from a degree-two variable node with neutral half-edges only. In the first case we denote by $Z_k$ the increment of active check half-edges, while in the second case $Z_k-1$ plays the same role as $A_0$ and the process restarts at stage $k$. Therefore, with $A_{-1} \coloneqq 0$, the stochastic process $(A_k)_{k \geqslant 0}$ satisfies the recursion
\begin{equation}
\label{eqn:A+process}
A_k = A_{k-1} + Z_k - 1.
\end{equation}
Notice that in the first case where $A_{k-1} \geqslant 1$, the ``$-1$ term'' in (\ref{eqn:A+process}) takes into account the fact that we need to change the status of the check half-edge we started at stage $k-1$ from active to explored. Further define 
\begin{equation}
\label{eqn:first+zero}
K = \inf\{k \, | \, A_k = 0\},
\end{equation}
which indicates the first stage that the process $(A_k)_{k \geqslant 0}$ returns to $0$. We set $K$ to be $\infty$ if no such $k$ exists.

To see why $\bhatta(\cc) \lambda^{\prime}(0) \rho^{\prime} (1)$ is a meaningful parameter to analyze, we recall that $\error(\cc^{\varoast l}) \leqslant \bhatta(\cc)^l$ for any $l \geqslant 1$. Given a design degree distribution pair $(\lambda, \rho)$ such that $\bhatta(\cc) \lambda^{\prime}(0) \rho^{\prime} (1) < 1$, it holds that $\error(\cc^{\varoast l}) (\lambda^{\prime}(0) \rho^{\prime} (1))^l < 1$. As we will see next, this condition gives rise to an operational meaning to the process $(A_k)_{k \geqslant 0}$. In the following and for notational simplicity, we define 
\begin{equation*}
\ttd = \ttr-1, 
\end{equation*}
where we recall from Section~\ref{subsec:prelim} that $\ttr$ is the maximum check-node degree.

\begin{theorem}
\label{thm:subcritical}
Given a design degree distribution pair $(\lambda,\rho)$, consider transmission over a BMS channel that is characterized by its $L$-density $\cc$. If $\bhatta(\cc) \lambda^{\prime}(0) \rho^{\prime} (1) < 1$, then there exist constants $\gamma = \gamma(\cc,\lambda,\rho) \in (0,1)$ and $\delta = \delta(\cc,\lambda,\rho) > 0$ such that for any integer $l \geqslant 1$ and any $a \in (0,1/2)$, it holds for all sufficiently large $n$ that
\begin{equation}
\PP(K > n^a) \leqslant \ee^{-n^a \delta^2 \gamma^{2l}/\ttd^{2l}}.
\end{equation}
\end{theorem}

Theorem~\ref{thm:subcritical} reveals that with probability tending to $1$ as the blocklength tends to $\infty$, the first stage that the process $(A_k)_{k \geqslant 0}$ returns to $0$ cannot be later than $\cO(n^a)$ for any $a \in (0,1/2)$, as long as $\bhatta(\cc) \lambda^{\prime}(0) \rho^{\prime} (1) < 1$ holds. The proof of this theorem is given in Appendix~\ref{appdix:stab+map+cycle}.

The rest of this section focuses on the case $\bhatta(\cc) \lambda^{\prime}(0) \rho^{\prime} (1) > 1$. First we notice that since $\bhatta(\cc) \geqslant \error(\cc)$, $\error(\cc) \lambda^{\prime}(0) \rho^{\prime} (1) > 1$ implies $\bhatta(\cc) \lambda^{\prime}(0) \rho^{\prime} (1) > 1$ while the opposite direction does not necessarily holds. Consequently, $\bhatta(\cc) \lambda^{\prime}(0) \rho^{\prime} (1) > 1$ gives a more general condition on the pair $(\lambda, \rho)$. On the other hand, $\error(\cc^{\varoast \ell})$ behaves essentially like a constant multiple of $\bhatta(\cc) ^\ell$ so that it is natural to expect that there exists an $l = l(\cc, \lambda, \rho) \in \NN$ such that $\bhatta(\cc) \lambda^{\prime}(0) \rho^{\prime}(1) > 1$ implies $\error(\cc^{\varoast l}) (\lambda^{\prime}(0) \rho^{\prime}(1))^l > 1$. As we shall see, contrary to Theorem~\ref{thm:subcritical}, this condition implies that with strictly positive probability, the first stage that the process $(A_k)_{k \geqslant 0}$ returns to $0$ is linear in the blocklength $n$.

We start by introducing a notion of residual degree distribution which quantifies how degree distribution changes after some very small perturbation on the ``original'' degree distribution has been added, and this notion will be used for subsequent analysis of the exploration process. We recall from Section~\ref{subsec:prelim} that $L$/$R$ is the design degree distribution from a node perspective corresponding to $\lambda$/$\rho$, and $r = 1 - L^{\prime}(1)/R^{\prime}(1)$ is the design rate. Meanwhile, the generating function $L(x)$/$R(x)$ has finite maximum degree $\ttl$/$\ttr$.

\begin{definition}[Residual Degree Distribution]
\label{defn:residual}
Given an $\varepsilon = \varepsilon(\cc,\lambda,\rho) > 0$, let $\{ \phi_i | 2 \leqslant i \leqslant \ttl \}$ and $\{ \psi_i | 2 \leqslant i \leqslant \ttr \}$ be such that $0 \leqslant \phi_i \leqslant i L_i$ and $0 \leqslant \psi_i \leqslant i R_i (1 - r)$ and satisfy $\sum_{i=2}^{\ttl} \phi_i = \sum_{i=2}^{\ttr} \psi_i \leqslant \varepsilon$. For a blocklength $n \in \NN$ and for every $i$, define
\begin{align}
\label{eqn:hat+lambda}
\hat{\lambda}_{i,n} &= \frac{i L_{i,n} - \phi_i}{L^{\prime} _n (1)}; \\ 
\label{eqn:hat+rho}
\hat{\rho}_{i,n} &= \frac{i R_{i,n} - \psi_i/(1-r_n)}{R^{\prime} _n (1)}.
\end{align}
Further define $\hat{\lambda}_{1,n} = 1 - \sum_{i=2}^{\ttl} \hat{\lambda}_{i,n}$ and define $\hat{\rho}_{1,n} = 1 - \sum_{i=2}^{\ttr} \hat{\rho}_{i,n}$. Denote by $\hat{\lambda}_n$ and $\hat{\rho}_n$ the variable and check \textit{residual degree distributions from an edge perspective}, respectively. The variable and check residual \textit{design} degree distributions from an edge perspective, $\hat{\lambda}$ and $\hat{\rho}$, are defined similarly.
\end{definition}

The next lemma and theorem make the above discussion precise and the proofs are given in Appendix~\ref{appdix:stab+map+cycle}. Let us remind the reader that in (\ref{eqn:hat+lambda}) and (\ref{eqn:hat+rho}), the residual degree distributions are normalized with respect to the exact number of edges for the ``original'' degree distributions $\lambda_n$ and $\rho_n$. As we shall see later, these normalizations simplify our calculations when we analyze the exploration process.

\begin{lemma}
\label{lem:bhatta+error}
Suppose that $\bhatta(\cc) \lambda^{\prime}(0) \rho^{\prime} (1) > 1$. Then there exist $\gamma = \gamma(\cc, \lambda, \rho) > 1$, $l = l(\cc, \lambda, \rho) \in \NN$, and $\varepsilon = \varepsilon(\cc, \lambda, \rho) > 0$ such that the degree distribution pairs $(\lambda_n,\rho_n)$ satisfy
\begin{equation}
\label{eqn:lem:E+ver+B}
\error(\cc^{\varoast l}) \big( \hat{\lambda}^{\prime} _n (0) \hat{\rho}^{\prime} _n (1) \big)^l \geqslant \gamma^l,
\end{equation}
for all large enough blocklength $n$.
\end{lemma}

Theorem~\ref{thm:supercritical} below is one of the main results of this Section~\ref{subsec:cycle+deg2}. It characterizes the exploration process in Section~\ref{subsec:exploration} when the condition $\bhatta(\cc) \lambda^{\prime}(0) \rho^{\prime} (1) > 1$ holds and also reveals that the parameter $\bhatta(\cc) \lambda^{\prime}(0) \rho^{\prime} (1)$ has an operational significance to the decoding analysis.

\begin{theorem}
\label{thm:supercritical}
Given a design degree distribution pair $(\lambda,\rho)$, consider transmission over a BMS channel characterized by its $L$-density $\cc$. If $\bhatta(\cc) \lambda^{\prime}(0) \rho^{\prime} (1) > 1$ holds, then there exist constants $\gamma = \gamma(\cc,\lambda,\rho) > 1$, $l = l(\cc, \lambda, \rho) \in \NN$, and $\epsilon = \epsilon(\cc, \lambda, \rho) > 0$ such that for any $\delta \in (0,1)$ satisfying $\bar{\delta}\gamma^l > 1$ with $\bar{\delta} \coloneqq 1 - \delta$, it holds for all large enough blocklength $n$ that
\begin{equation}
\label{eqn:exp+bound+Ak}
\PP\big( A_{\epsilon n} \leqslant (\bar{\delta} \gamma^l - 1) \epsilon n\big) \leqslant \ee^{- \epsilon n \delta^2 \gamma^l / (2 \ttd^l)}.
\end{equation}
Further let $K = \inf\{k \, | \, A_k = 0\}$. Then there exist a $\kappa = \kappa(\cc,\lambda,\rho) \in \NN$ and a constant $c_{\kappa} = c_{\kappa}(\cc,\lambda,\rho) \in (0,1)$ such that
\begin{equation}
\label{eqn:exp+bound+K}
\PP( K \leqslant \epsilon n \, | \, K > \kappa ) \leqslant c_{\kappa} \big( 1 - \ee^{- (\epsilon n - \kappa) \delta^2 \gamma^l / (2 \ttd^l)} \big),
\end{equation}
for all large enough blocklength $n$.
\end{theorem}

Theorem~\ref{thm:supercritical} illustrates that the number of active check half-edges at stage $\epsilon n$ is with high probability at least as large as a multiple of $\gamma^l \epsilon n$, which is linear in the blocklength $n$. In fact the convergence rate is exponentially fast in $n$. Furthermore, with probability strictly smaller than $1$, the first stage that the process $(A_k)_{k \geqslant 0}$ returns to $0$ is before stage $\epsilon n$. This in particular implies that with strictly positive probability, the first stage of returning to $0$ is after stage $\epsilon n$, which means that the variable node $v$ of degree two in stage $0$, when we started the exploration process, is connected to at least one active check node in stage $\epsilon n$.

Now if we connect the remaining variable half-edge $e_2$ belonging to the variable node $v$ uniformly at random to an unexplored check half-edge in stage $\epsilon n + 1$, then with positive probability the connected check half-edge is an active one in stage $\epsilon n$; see Figure~\ref{fig:stab+cycle} for an illustration.

\begin{figure}[H]
\centering
\includegraphics{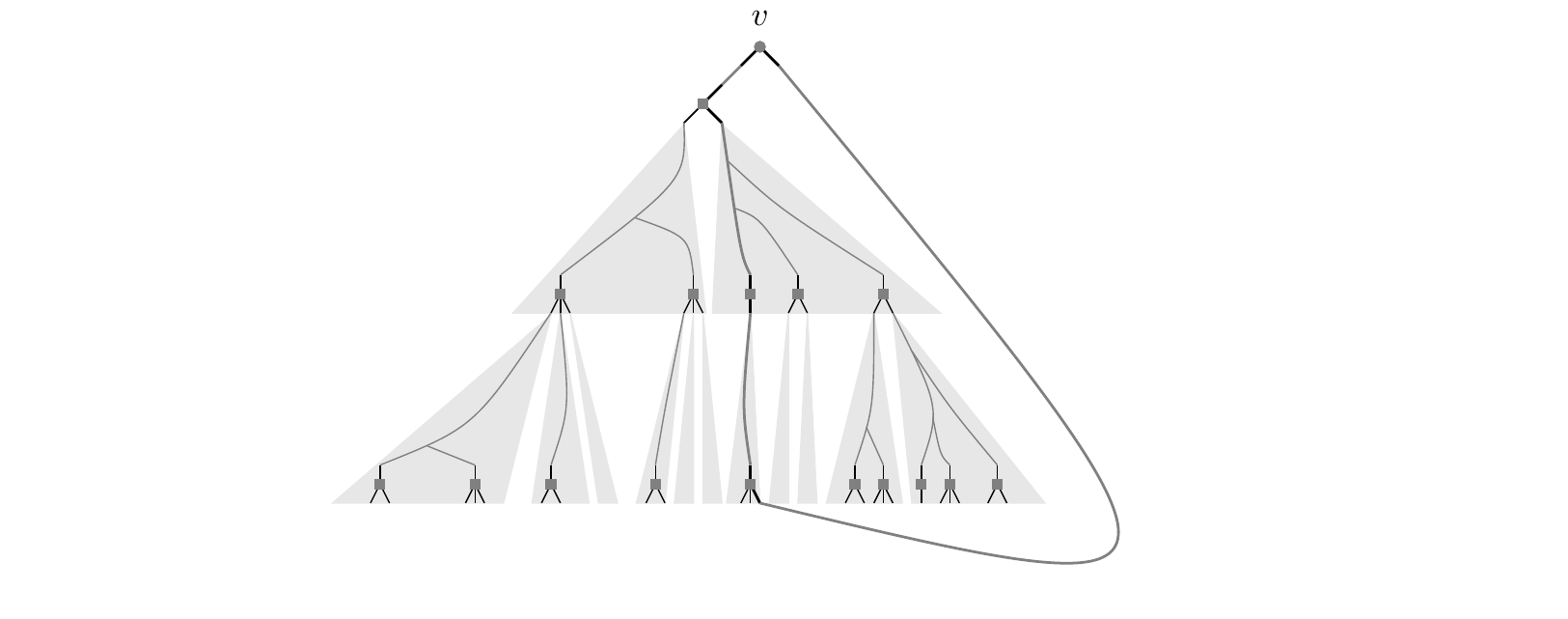}
\caption{Illustration of the ``final step'' of the exploration process.}
\label{fig:stab+cycle}
\end{figure}

Notice also that the sum of log-likelihood ratios along any path between the check half-edge $e_1$ and any active check half-edge at stage $\epsilon n$ is negative. Furthermore, if the variable half-edge $e_2$ of $v$ is connected to one of those active check half-edges at stage $\epsilon n$, then a cycle involving $v$ is created. The consequence of this ``final step'' is made precise by the following theorem.

\begin{theorem}
\label{thm:cycle}
Given a design degree distribution pair $(\lambda,\rho)$, consider transmission over a BMS channel characterized by its $L$-density $\cc$. Let $v$ be an arbitrary variable node of degree two and denote by $\cC_v$ be the event that $v$ lies on a bipartite cycle such that all variable nodes belonging to the cycle are of degree two and that the corresponding sum of log-likelihood ratios is negative. If $\bhatta(\cc) \lambda^{\prime}(0) \rho^{\prime}(1) > 1$, then there exists a constant $c_{\mathrm{cyc}} = c_{\mathrm{cyc}}(\cc, \lambda, \rho) > 0$ such that $\liminf_{n \to \infty} \PP(\cC_v) \geqslant c_{\mathrm{cyc}}$.
\end{theorem}

\begin{proof}
Since we get a desired cycle if $e_2$ is connected to an active check half-edge in stage $\epsilon n$ and if the log-likelihood ratio of the variable node $v$ is negative, it follows that
\begin{align*}
\PP(\cC_v) &\geqslant \error(\cc) \PP \big( \{A_{\epsilon n} > (\bar{\delta}\gamma^l-1) \epsilon n \} \cap \{ K > \epsilon n\} \big) \\
&\geqslant \error(\cc) \PP(K > \kappa) \PP \big( \{A_{\epsilon n} > (\bar{\delta}\gamma^l-1) \epsilon n \} \cap \{ K > \epsilon n\}  \, | \, \{ K > \kappa\} \big),
\end{align*}
where $\gamma$, $\epsilon$, $l$ and $\kappa$ are specified in Theorem~\ref{thm:supercritical}. Then it holds that
\begin{align*}
p_{\cC_v} &\coloneqq \PP \big( \{A_{\epsilon n} \leqslant (\bar{\delta}\gamma^l-1) \epsilon n \} \cup \{ K \leqslant \epsilon n\}  \, | \, \{ K > \kappa \} \big) \\
&\leqslant \PP(A_{\epsilon n} \leqslant (\bar{\delta}\gamma^l-1) \epsilon n \, | \, K > \kappa ) + \PP(K \leqslant \epsilon n \, | \, K > \kappa) \\
&\leqslant \PP(A_{\epsilon n} \leqslant (\bar{\delta}\gamma^l-1) \epsilon n) / \PP(K > \kappa) +  \PP(K \leqslant \epsilon n \, | \, K > \kappa),
\end{align*}
so that
\begin{equation*}
\PP(\cC_v) \geqslant \error(\cc) \PP(K > \kappa) (1 - p_{\cC_v}).
\end{equation*}
Therefore by (\ref{eqn:exp+bound+Ak}) and (\ref{eqn:exp+bound+K}) in Theorem~\ref{thm:supercritical}, it follows that
\begin{equation*}
\liminf_{n \to \infty} \PP(\cC_v) \geqslant \error(\cc) (1- c_{\kappa}) \liminf_{n \to \infty} \PP(K > \kappa) \eqqcolon c_{\mathrm{cyc}}.
\end{equation*}
Since $\kappa = \kappa(\cc,\lambda,\rho)$ is fixed and does not depend on $n$, $\liminf_{n \to \infty} \PP(K > \kappa) > 0$ so that $c_{\mathrm{cyc}} > 0$.
\end{proof}

Theorem~\ref{thm:cycle} has the following immediate consequence.

\begin{theorem}
\label{thm:cycle+deg2}
Suppose that $\bhatta(\cc) \lambda^{\prime}(0) \rho^{\prime}(1) > 1$. Let $C_n$ denote the number of degree-two variable nodes lying on cycles such that, on each cycle, all variable nodes are of degree two and the corresponding sum of log-likelihood ratios is negative. Then $\liminf_{n \to \infty} \mathbb{P}(C_n \geqslant c_{\mathrm{cyc}} n L_{2,n} / 2) \geqslant c_{\mathrm{cyc}} / 2$.
\end{theorem}

\begin{proof}
Since $\EE[C_n] = n L_{2,n} \PP(\cC_v)$, we conclude by Theorem \ref{thm:cycle} that 
\begin{equation*}
\liminf_{n \to \infty} \EE[C_n] / (n L_{2,n}) \geqslant c_{\mathrm{cyc}}. 
\end{equation*}
Meanwhile, since $C_n$ is nonnegative and is upper bounded by $n L_{2,n}$, we can obtain
\begin{equation}
\label{eqn:C_n+bound}
\EE[C_n] \leqslant c_{\mathrm{cyc}} n L_{2,n} / 2 + n L_{2,n} \PP(C_n \geqslant c_{\mathrm{cyc}} n L_{2,n} / 2).
\end{equation}
Now assume to the contrary that $\liminf_{n \to \infty} \PP(C_n \geqslant c_{\mathrm{cyc}} n L_{2,n} / 2) < c_{\mathrm{cyc}} / 2$, so that the inequality in (\ref{eqn:C_n+bound}) implies
\begin{align*}
\liminf_{n \to \infty} \frac{\EE[C_n]}{n L_{2,n}} &\leqslant c_{\mathrm{cyc}} / 2 + \liminf_{n \to \infty} \PP(C_n \geqslant c_{\mathrm{cyc}} n L_{2,n} / 2) \\
&< c_{\mathrm{cyc}} / 2 + c_{\mathrm{cyc}} / 2,
\end{align*}
which leads to a contradiction.
\end{proof}


\subsection{Lower Bound on Probability of Block Error}
\label{subsec:lbd+block+error+prob}

We present a simple consequence of Theorem~\ref{thm:cycle}, which reveals how stability is related to the probability of block error under blockwise MAP decoding. Consider transmission over a BMS channel with transition probability $p_{Y|X}$ and let $\cC$ be a binary code of blocklength $n$. Let $\uX$ be the codeword chosen with uniform distribution from $\cC$ and sent through the channel, and denote by $\uY$ the channel output. For $\uy \in \cY^n$, the blockwise MAP decoding chooses
\begin{equation}
\label{eqn:block+map+def}
\hat{\ux}^{\ML}(\uy) \coloneqq \underset{\ux \in \cC}{\argmax} \sum_{i=1}^n \ux_i l(\uy_i),
\end{equation}
where $l(y) = \ln(p_{Y|X}(y|1)/p_{Y|X}(y|-1))$ is the log-likelihood ratio function. Denote by ${\rP}_{\rB} ^{\ML}$ the probability of block error under blockwise MAP decoding for the code $\cC$, where the average is over the choice of the codeword as well as the channel fluctuations.

Consider transmission over a BMS channel with $L$-density $\cc_{\tth}$ using a design degree distribution pair $(\lambda, \rho)$. We denote by $\tG$ the Tanner graph chosen with uniform distribution from the $\LDPC(\lambda_n,\rho_n)$ ensemble, and denote by ${\rP}_{\rB} ^{\ML} (\tG, \tth)$ the conditional probability of block error under blockwise MAP decoding, conditioned on the all-one codeword being transmitted. We further define ${\rP}_{\rB} ^{\ML} = {\rP}_{\rB} ^{\ML}(\lambda_n, \rho_n) = \EE_{\LDPC(\lambda_n,\rho_n)}[{\rP}_{\rB} ^{\ML} (\tG, \tth)]$, where the average is taken over the randomness of the Tanner graph $\tG$. The main result of this section is the following theorem.

\begin{theorem}
\label{thm:block+map+block+error}
Let $(\lambda, \rho)$ be a design degree distribution pair and let $\{\cc_{\tth}\}$ be an ordered and complete family of BMS channels indexed by entropy $\tth$. Further denote by $\tth^{\stab} = \tth^{\stab}(\{\cc_{\tth}\}, \lambda, \rho)$ the stability threshold associated with $\{\cc_{\tth}\}$ and $(\lambda, \rho)$. Consider transmission over a BMS channel with $L$-density $\cc_{\tth}$. If $\tth > \tth^{\stab}$ so that $\bhatta(\cc_{\tth}) \lambda^{\prime}(0) \rho^{\prime}(1) > 1$, then $\liminf_{n \to \infty} {\rP}_{\rB} ^{\ML}(\lambda_n, \rho_n) > 0$.
\end{theorem}

\begin{proof}
By Theorem~\ref{thm:cycle+deg2} we know that with strictly positive probability and for all sufficiently large blocklength $n$, there exists a cycle on $\tG$ such that the cycle only contains variable nodes of degree two and the sum of log-likelihood ratios along these variable nodes is negative. Let us denote by $\cI_{\mathrm{cyc}}$ the index set of the variable nodes along the cycle. Then we have $\sum_{i \in \cI_{\mathrm{cyc}}} l(\uY_i) \leqslant 0$.

Furthermore, the index set $\cI_{\mathrm{cyc}}$ gives rise to a codeword in the code $\cC = \cC(\tG)$. More precisely, we let $\ux^{\circ}$ be such that $\ux_i ^{\circ} = -1$ if $i \in \cI_{\mathrm{cyc}}$ and $\ux_i ^{\circ} = 1$ if $i \notin \cI_{\mathrm{cyc}}$. That $\ux^{\circ}$ is a codeword in $\cC$ can be verified as follows. Consider any check node along the cycle. Then it is connected with two variable nodes with indices in $\cI_{\mathrm{cyc}}$ along the cycle and all other connected variable nodes, if any, are off the cycle. This means that for any particular check node on the cycle, the pointwise multiplication of the corresponding $\ux_i ^{\circ}$s is equal to $1$, satisfying the corresponding parity-check equation. Now since $\sum_{i \in \cI_{\mathrm{cyc}}} l(\uY_i) \leqslant 0$ and $\ux_i ^{\circ} = -1$ for all $i \in \cI_{\mathrm{cyc}}$, we can conclude that
\begin{equation*}
\sum_{i \in \cI_{\mathrm{cyc}}} \ux_i ^{\circ} l(\uY_i) \geqslant \sum_{i \in \cI_{\mathrm{cyc}}} l(\uY_i),
\end{equation*}
which implies that $\ux^{\circ}$ gives rise to a larger sum in (\ref{eqn:block+map+def}) than the sum corresponding to the transmitted all-one codeword, resulting in a block error.
\end{proof}


\subsection{Lower Bound on Probability of Bit Error}
\label{subsec:lbd+bit+error+prob}

We investigate stability under bitwise MAP decoding when transmission takes place on a BMS channel using a design degree distribution pair $(\lambda,\rho)$. We recall from Theorem~\ref{thm:cycle} in Section~\ref{subsec:cycle+deg2} that we let $v$ be an arbitrary variable node of degree two of the Tanner graph $\tG$ that is chosen uniformly at random from the $\LDPC(\lambda_n,\rho_n)$ ensemble. Then a codeword from the code $\cC = \cC(\tG)$ is chosen uniformly at random and is transmitted through a BMS channel with $L$-density $\cc$. Assuming that the all-one codeword was transmitted, if the condition $\bhatta(\cc) \lambda^{\prime}(0) \rho^{\prime}(1) > 1$ holds, then with strictly positive probability as the blocklength tends to $\infty$, the variable node $v$ is lying on a cycle such that all variable nodes on the cycle are of degree two and the sum of log-likelihood ratios along these variable nodes is negative.

Let us show that, for a broad family of BMS channels, we are able to conclude that the sum of log-likelihood ratios along the cycle, where $v$ lies on, is not only negative but also \textit{bounded from below} by a constant that only depends on $\cc$ and $(\lambda,\rho)$. We start by considering the case of the BSC with parameter $p$ and $L$-density $\cc_p$. Recall that each stage of the exploration process consists of $l$ steps, where $l$ is the parameter defined in Lemma~\ref{lem:bhatta+error}. Without loss of generality we assume that $l$ is an odd integer. The $L$-density of the sum of $l$ log-likelihood ratios of the BSC is the $l$-fold $\varoast$-convolution of $\cc_p$ and reads
\begin{equation}
\cc_p ^{\varoast l} (y) = \sum_{i = 0}^l \binom{l}{i} p^{l - i} {(1-p)}^{i} \Delta_{-(l - 2i) \ln((1-p)/p)} (y).
\end{equation}
Suppose that we have performed $l$ steps of stage $k$ for some $k \geqslant 1$. Then there exists a unique path that involves active check nodes only and starts from variable half-edge $e_1$ belonging to variable node $v$. Let us assume without loss of generality that the log-likelihood ratio takes values in $\{-1,1\}$ rather than $\{-\ln(\bar{p}/p), \ln(\bar{p}/p)\}$. Now suppose that the ``partial'' sum of log-likelihood ratios along this path after we performed step $0$ but before we perform step $1$ of stage $k$ is equal to $-l + \ell$, where $0 \leqslant \ell \leqslant l$ is arbitrary. Clearly the claim holds for $k=1$. Then rather than requiring that the sum of log-likelihood ratios along the $l$ variable nodes explored in stage $k$ is negative, we require that the sum is in the closed interval $[-\ell,-\ell+l]$. In this way, the ``total'' sum along the path, which starts from the variable half-edge $e_1$ belonging to $v$, is bounded in $[-l,0]$. It remains to show that, in this case, the conditional probability can only increase.

Now some calculation reveals that, for any integer $\ell \in [0,l]$, there exists a unique integer $k = k(\ell)$, which satisfies $0 \leqslant k \leqslant (l+1)/2$, such that the interval $[-\ell,-\ell+l]$ contains $(l+1)/2$ odd integers of the form $-(l-2i)$, where $k \leqslant i \leqslant k + (l-1)/2$. Consequently if we define
\begin{equation*}
B(k_0; k_1) = \sum_{i = k_0}^{k_1} \binom{l}{i} p^{l - i} {(1-p)}^{i},
\end{equation*}
then $B(0; (l-1)/2)$ is the probability that the sum of $l$ log-likelihood ratios, each of which has an $L$-density equal to $p \Delta_{-1} + (1-p) \Delta_{1}$, is negative, while $B(k(\ell); k(\ell) + (l-1)/2)$ is the probability that the sum of $l$ log-likelihood ratios is bounded in $[-\ell,-\ell+l]$. The next lemma shows that for any integer $k$ satisfying $0 \leqslant k \leqslant (l+1)/2$, it holds that $B(0; (l-1)/2) \leqslant B(k; k + (l-1)/2)$.

\begin{lemma}
\label{lem:bsc+density}
For any $p \in [0,1/2]$ and any odd integer $l \geqslant 2$, it holds that
\begin{equation}
\label{eqn:bsc+binomial}
B\big( 0; (l-1)/2 \big) \leqslant B\big( k; k+(l-1)/2 \big),
\end{equation}
for every integer $k$ satisfying $0 \leqslant k \leqslant (l+1)/2$.
\end{lemma}

The proof of Lemma~\ref{lem:bsc+density} is given in Appendix~\ref{appdix:bsc+sum}. This lemma implies that for the case of the BSC, in Theorem~\ref{thm:cycle} we can further require that the sum of log-likelihood ratios along the cycle is negative but also bounded in the interval $[-(l+1)\ln(\bar{p}/p), 0]$, while the conclusion still holds.

In fact, a similar conclusion holds for a class of BMS channels other than the BSC. To verify this, we need the following lemma, which is a variant of Lemma~\ref{lem:bhatta+error}.

\begin{lemma}
\label{lem:truncation+error}
Let $\cc$ be an $L$-density and $(\lambda,\rho)$ a design degree distribution pair. Suppose $\bhatta(\cc) \lambda^{\prime}(0) \rho^{\prime} (1) > 1$ holds. Then there exist $M = M(\cc, \lambda, \rho) > 0$, $\gamma = \gamma(\cc, \lambda, \rho) > 1$, $l = l(\cc, \lambda, \rho) \in \NN$, and $\varepsilon = \varepsilon(\cc, \lambda, \rho) > 0$ such that for all $n$ large enough, it holds that
\begin{equation}
\label{eqn:new+E+ver+B}
\hat{\error}(\cc^{\varoast l})\big(\hat{\lambda}^{\prime} _n (0) \hat{\rho}^{\prime} _n (1)\big)^l \geqslant \gamma^l,
\end{equation}
where $(\hat{\lambda}_n, \hat{\rho}_n)$ is the residual degree distribution associated with $(\lambda_n,\rho_n)$ according to Definition~\ref{defn:residual}, and $\hat{\error}(\cc^{\varoast l})$ is defined by
\begin{equation*}
\hat{\error}(\cc^{\varoast l}) = \frac{1}{2} \int_{{-M}}^{M} \ee^{-(|y/2| + y/2)} \cc^{\varoast l} (y) \diff y.
\end{equation*}
\end{lemma}

\begin{proof}
Let $\beta$ and $\bhatta$ be both specified according to Lemma~\ref{lem:bhatta+error}. By Lemma \ref{lem:bhatta+error}, there exist $\gamma = \gamma(\cc, \lambda, \rho) > 1$, $l = l(\cc, \lambda, \rho) \in \NN$, and $\varepsilon = \varepsilon(\cc, \lambda, \rho) > 0$ such that there exists a $\acute{\beta} \in {(0, \beta)}$ with $\error(\cc^{\varoast l}) \geqslant \beta \bhatta^l > \acute{\beta} \bhatta^l$ and for all $n$ large enough, it holds that
\begin{equation}
\label{eqn:acute+beta}
\acute{\beta} \big(\bhatta \hat{\lambda}^{\prime} _n (0) \hat{\rho}^{\prime} _n (1)\big)^l > \gamma^l. 
\end{equation}
Pick such an $l$, so that we can select $M = M(\cc, \lambda, \rho) > 0$ such that $\hat{\error}(\cc^{\varoast l}) \geqslant \acute{\beta} \bhatta^l$. Therefore (\ref{eqn:new+E+ver+B}) follows by combining this and (\ref{eqn:acute+beta}).
\end{proof}

Next we define admissible BMS channels and as we shall see later, the boundedness condition on the sum of the log-likelihood ratios along a cycle is satisfied for this class of channels.

\begin{definition}
\label{def:admissible+bms}
Given a design degree distribution pair $(\lambda,\rho)$, a BMS channel with $L$-density $\cc$ is called \textit{admissible} if there exist $l = l(\cc, \lambda, \rho) \in \NN$ and $M = M(\cc,\lambda, \rho) > 0$ such that 
\begin{equation}
\label{eqn:cond+exp+bounded}
\int_{-M}^0 \cc^{\varoast l} (y) \diff y \leqslant \int_{-m}^{M-m} \cc^{\varoast l} (y) \diff y,
\end{equation}
for every $m \in [0,M]$.
\end{definition}

Now recall the exploration process in Section~\ref{subsec:exploration}. Suppose we have performed $l$ steps of stage $k$ for some $k \geqslant 1$. Then there exists a unique path that involves active check nodes only and starts from variable half-edge $e_1$ belonging to variable node $v$. Further assume that the ``partial'' sum of log-likelihood ratios along this path after we performed step $0$ but before we perform step $1$ of stage $k$ is equal to $-M + m$, where $0 \leqslant m \leqslant M$ is arbitrary. Then rather that requiring that the sum of log-likelihood ratios along the $l$ variable nodes explored in stage $k$ is in $[-M,0]$, we require that the sum is in the closed interval $[-m,-m+M]$. In this way, the ``total'' sum along the path, which starts from the variable half-edge $e_1$ belonging to $v$, is bounded in $[-M,0]$. Then (\ref{eqn:cond+exp+bounded}) in Definition~\ref{def:admissible+bms} and Lemma~\ref{lem:truncation+error} imply that the conditional probability that the status of the active check half-edge after performing step $l$ in stage $k$ remains active is at least as large as $\hat{\error}(\cc^{\varoast l})$. This leads us to the following theorem.

\begin{theorem}
\label{thm:bounded+cycle}
Consider transmission over the BSC or an admissible BMS channel characterized by its $L$-density $\cc$ and let $(\lambda,\rho)$ be a design degree distribution pair. Let $v$ be an arbitrary variable node of degree two and denote by $\cC_v$ be the event that $v$ lies on a bipartite cycle such that all variable nodes on it are of degree two and that the corresponding sum of log-likelihood ratios is negative and bounded from below by a constant that only depends on $\cc$ and $(\lambda,\rho)$. If $\bhatta(\cc) \lambda^{\prime}(0) \rho^{\prime}(1) > 1$, then there exists a constant $c_{\mathrm{cyc}} = c_{\mathrm{cyc}}(\cc, \lambda, \rho) > 0$ such that $\liminf_{n \to \infty} \PP(\cC_v) \geqslant c_{\mathrm{cyc}}$.
\end{theorem}

Similar to Theorem~\ref{thm:cycle+deg2}, a direct consequence of Theorem~\ref{thm:bounded+cycle} is the following. With strictly positive probability as the blocklength $n$ tends to $\infty$, there is a linear (in the blocklength) fraction of variable nodes of degree two such that the following property holds. For each such variable node, there is a cycle such that all variable nodes on it are of degree two and that the sum of log-likelihood ratios along them is not only negative but also bounded from below.

Let us now show that admissible BMS channels do exist. Consider the BAWGNC with variance $\sigma^2$ and $L$-density $\mathsf{c}_{\sigma} (y) = \sqrt{\sigma^2 / 8} \exp(- (y - 2 / \sigma^2)^2 \sigma^2/ 8)$. We denote by $I_{\sigma}(m)$ the corresponding difference of the left-hand side of (\ref{eqn:cond+exp+bounded}) and the right-hand side of (\ref{eqn:cond+exp+bounded}), \ie,
\begin{equation*}
I_{\sigma}(m) = \int_{-M}^0 \cc_{\sigma} ^{\varoast l} (y) \diff y - \int_{-m}^{M-m} \cc_{\sigma} ^{\varoast l} (y) \diff y.
\end{equation*}
Then we have
\begin{align}
I_{\sigma}(m) &= \int_{-M}^{-m} \cc_{\sigma} ^{\varoast l} (y) \diff y - \int_{0}^{M-m} \cc_{\sigma} ^{\varoast l} (y) \diff y \nonumber \\
\label{eqn:integral+error+1}
&= \int_{-M}^{-m} \cc_{\sigma} ^{\varoast l} (y) \diff y - \int_{0}^{M-m} \ee^y \cc_{\sigma} ^{\varoast l} (-y) \diff y \\ 
\label{eqn:integral+error+2}
&\leqslant \int_{-M}^{-m} \cc_{\sigma} ^{\varoast l} (y) \diff y - \int_{-M+m}^{0} \cc_{\sigma} ^{\varoast l} (y) \diff y \\
\label{eqn:integral+error+3}
&\leqslant \int_{-M}^{-m} \cc_{\sigma} ^{\varoast l} (y) \diff y - \int_{-M}^{-m} \cc_{\sigma} ^{\varoast l} (y) \diff y \\
&= 0, \nonumber
\end{align}
where (\ref{eqn:integral+error+1}) is due to the symmetry of $\cc_{\sigma} ^{\varoast l}$, (\ref{eqn:integral+error+2}) follows by the fact that $\ee^{-y} \geqslant 1$ for $y \in {[{-M+m}, 0]}$, and (\ref{eqn:integral+error+3}) follows by the fact that the $L$-density $\cc_{\sigma} ^{\varoast l}$ is increasing on ${(-\infty,0]}$.

Finally we connect the exploration process to the probability of bit error under bitwise MAP decoding. We restrict ourselves to the case where transmission takes place over the BSC. 

\begin{example}
\label{exa:pattern+mapping}
Consider the Tanner graph $\tG$ shown in Figure~\ref{fig:stab+example1}, where the variable nodes are labelled by $\{1,2,3,4\}$. Besides, there are $2$ cycles involving variable nodes $1$, \ie, one cycle, call it $\cC_1$, contains $\{1,2\}$ and another one, call it $\cC_2$, contains $\{1,3,4\}$.

\begin{figure}[H]
\centering
\includegraphics{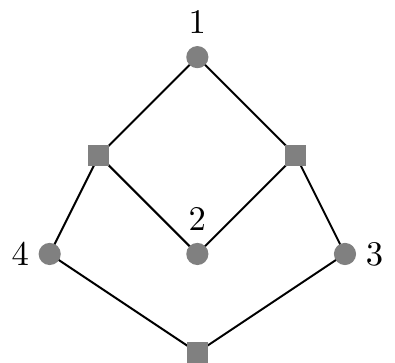}
\caption{Tanner graph $\tG$ in Example~\ref{exa:pattern+mapping}.}
\label{fig:stab+example1}
\end{figure}

From graph $\tG$ we construct a graph $\tR$ as follows. For $j \in \{0,1,\dots,15\}$, let $\ub_j$ be the vector in $\{0,1\}^4$ that corresponds to the binary expansion of $j$, and the last coordinate of $\ub_j$ corresponds to the least significant bit of $j$ in base $2$. Let $\ujj$ be the vector in $\{-1,1\}^4$ by pointwise apply the mapping $x \mapsto 1-2x$ to $\ub_j$. The edge set of $\tR$ is defined as follows. Let $\cI_1 = \{1,2\}$ and $\cI_2 = \{1,3,4\}$. For $\uii, \ujj \in \{-1,1\}^4$, they are connected if there is $\cI \in \{\cI_1, \cI_2\}$ such that $\uii_k = - \ujj_k$ for $k \in \cI$, $\uii_k = \ujj_k$ for $k \in [4] \backslash \cI$ and $|\sum_{k \in \cI} \uii_k| = |\sum_{k \in \cI} \ujj_k|$ is bounded by $1$; see Figure~\ref{fig:stab+pairing1}.   

\begin{figure}[H]
\centering
\includegraphics{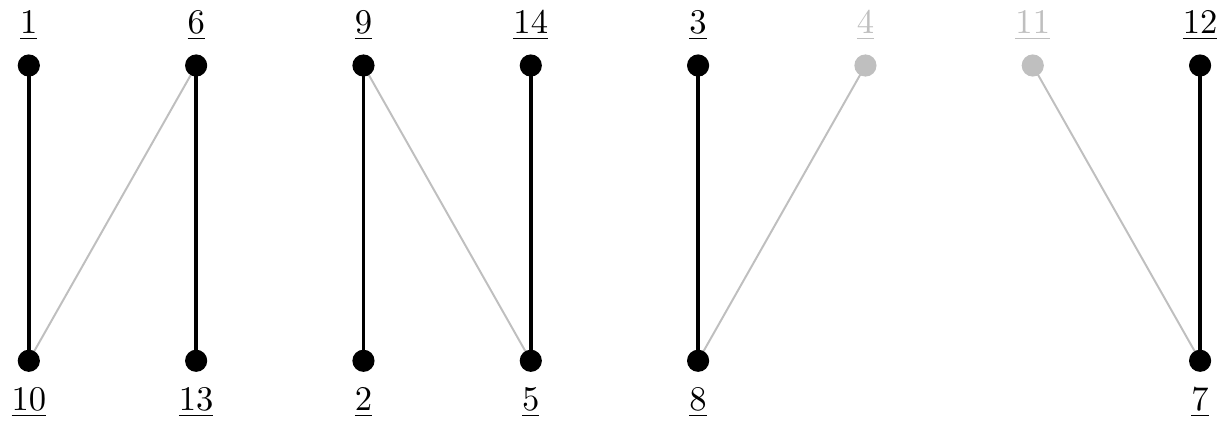}
\caption{Graph $\tR$ and an associated maximum matching in Example~\ref{exa:pattern+mapping}.}
\label{fig:stab+pairing1}
\end{figure}

The vertex set of $\tR$ is the largest subset of $\{-1,1\}^4$ such that each vertex has degree at least $1$. The graph $\tR$ has the following property for transmission over the BSC. Each vertex in $\tR$ corresponds to a realization of the channel output in $\{-1, 1\}^4$, and in the sequel we will call it a \textit{pattern}. For any two \textit{connected} patterns in $\tR$, the probability that the channel output is equal to one pattern is at most $\bar{p}/p$ times and at least $p/\bar{p}$ times the probability that the channel output is equal to the other pattern. Furthermore from Figure~\ref{fig:stab+pairing1} we see that a \textit{maximum matching} on the graph $\tR$ consists of $6$ edges, \eg, $\{\{\underline{1},\underline{10}\},\{\underline{2},\underline{9}\},\{\underline{3},\underline{8}\},\{\underline{5},\underline{14}\},\{\underline{6},\underline{13}\},\{\underline{7},\underline{12}\}\}$. As we shall see next, two connected patterns in $\tR$ give rise to opposite signs under bitwise MAP decoding. We further denote by $\tP$ the graph consisting of any maximum matching of $\tR$ together with all involved vertices belong to the matching, and call $\tR$ the \textit{realization graph} and $\tP$ the \textit{pattern graph} associated with $\tG$.  
\end{example}

We note that the ``negativity and boundedness constraint'' in Theorem~\ref{thm:bounded+cycle} can in fact be relaxed; \ie, we only require that the sum of log-likelihood ratios along the cycle described in Theorem~\ref{thm:bounded+cycle} to be \textit{bounded} in $[-M,M]$ for some $M = M(\cc, \lambda, \rho) > 0$. Clearly, the conclusion of Theorem~\ref{thm:bounded+cycle} still holds, as this event contains the original event. Some thought also reveals that the length of the cycle is of order $\Theta(2 \log_{\gamma}(\epsilon n))$. Now let us consider transmission over the BSC using the degree distribution pair $(\lambda_n,\rho_n)$. Fix a variable node index and call it $v$. For a Tanner graph $\tG \in \LDPC(\lambda_n,\rho_n)$, we construct the realization graph $\tR_v$ associated with $\tG$ as in Example~\ref{exa:pattern+mapping}, but only consider cycles on $\tG$ that only contains variable nodes of degree two, including $v$, and is of order $\Theta(2 \log_{\gamma}(\epsilon n))$. Moreover, let $\tP_v$ be the pattern graph associated with $\tG$, \ie, the edge set of $\tP_v$ is a maximum matching of $\tR_v$. Then we let $\cL^{\tR_v}$ be the vertex set of $\tR_v$ and let $\cL^{\tP_v}$ be the vertex set of $\tP_v$.

\begin{theorem}
\label{thm:stab+cond+bitwise+map}
Given a design degree distribution pair $(\lambda,\rho)$, consider transmission over the BSC with $L$-density $\cc_p$. Let $\tG$ be the Tanner graph chosen uniformly at random from the $\LDPC(\lambda_n,\rho_n)$ ensemble. A codeword from the code $\cC = \cC(\tG)$ is chosen uniformly at random and is then transmitted through the BSC. Let $\uL$ be the log-likelihood ratios associated with the channel output $\uY$ and let $v$ be any variable node of degree two in $\tG$. If $\bhatta(\cc_p) \lambda^{\prime}(0) \rho^{\prime}(1) > 1$, then $\liminf_{n \to \infty}\PP( \uL \in \cL^{\tR_v} \,|\, \uX = \uone ) > 0$. Further denote by $\rP_{\rb} ^{\MAP}(v)$ the bit error probability of $v$ under bitwise MAP decoding, conditioned on $\uone$ being transmitted. If $\liminf_{n \to \infty}\PP( \uL \in \cL^{\tP_v} \,|\, \uX = \uone ) > 0$, then it holds that $\liminf_{n \to \infty} \rP_{\rb} ^{\MAP}(v) > 0$.
\end{theorem}

\begin{proof}
That $\bhatta(\cc_p) \lambda^{\prime}(0) \rho^{\prime}(1) > 1$ implies $\liminf_{n \to \infty}\PP( \uL \in \cL^{\tR_v} \,|\, \uX = \uone ) > 0$ is a direct consequence of Theorem~\ref{thm:bounded+cycle}. For any $\ul \in \{-1,1\}^n$, let us define the bitwise MAP decoding function $\hat{x}_v ^{\MAP} (\ul)$ as
\begin{equation*}
\hat{x}_v ^{\MAP} (\ul) = \ln\Bigg( \frac{\sum_{\ux \in \cC: \ux_v = +1} \exp\big( \sum_{i=1}^n \ux_i \ul_i / 2 \big)}{\sum_{\ux \in \cC: \ux_v = -1} \exp\big( \sum_{i=1}^n \ux_i \ul_i / 2 \big)} \Bigg).
\end{equation*}
Then $v$ is decoded to be $1$ if $\hat{x}_v ^{\MAP}(\ul) > 0$ and is decoded to be $-1$ if $\hat{x}_v ^{\MAP}(\ul) < 0$. Otherwise, an independent fair coin is flipped to make a decision. Then we let $\ul^+$ and $\ul^-$ be such that $\{\ul^+,\ul^-\}$ is an edge of $\tP_v$. We know that there exists a codeword in $\cC(\tG)$, call it $\ux^{\circ}$, such that $\ux^{\circ} _v = -1$ and that $\ul^+ = \ux^{\circ} \cdot \ul^-$, where ``$\cdot$'' is pointwise multiplication. But
\begin{align*}
\hat{x}_v ^{\MAP} (\ul^+) &= \ln\Bigg( \frac{\sum_{\ux \in \cC: \ux_v = +1} \exp\big( \sum_{i=1}^n \ux_i {(\ux^{\circ} \cdot \ul^-)}_i / 2 \big)}{\sum_{\ux \in \cC: \ux_v = -1} \exp\big( \sum_{i=1}^n \ux_i {(\ux^{\circ} \cdot \ul^-)}_i / 2 \big)} \Bigg) \\
&= \ln\Bigg( \frac{\sum_{\ux \in \cC: \ux_v = +1} \exp\big( \sum_{i=1}^n {(\ux \cdot \ux^{\circ})}_i \ul^- _i / 2 \big)}{\sum_{\ux \in \cC: \ux_v = -1} \exp\big( \sum_{i=1}^n {(\ux \cdot \ux^{\circ})}_i \ul^- _i / 2 \big)} \Bigg) \\
&= \ln\Bigg( \frac{\sum_{\ux \in \cC: \ux_v = -1} \exp\big( \sum_{i=1}^n \ux_i \ul^- _i / 2 \big)}{\sum_{\ux \in \cC: \ux_v = +1} \exp\big( \sum_{i=1}^n \ux_i \ul^- _i / 2 \big)} \Bigg) = - \hat{x}_v ^{\MAP} (\ul^-).
\end{align*}
Thus $\hat{x}_v ^{\MAP} (\ul^+) < 0$ if and only if $\hat{x}_v ^{\MAP} (\ul^-) > 0$, so that one of them gives correct decoding while the other one gives incorrect decoding. But we also know that the ratio of $\PP(\uL = \ul^+ \,|\, \uX = \uone)$ and $\PP(\uL = \ul^- \,|\, \uX = \uone)$ is bounded between $(p/\bar{p})^l$ and $(\bar{p}/p)^l$, where $l$ is the parameter introduced in Lemma~\ref{lem:bhatta+error}. Consequently we can conclude that $\rP_{\rb} ^{\MAP}(v) \geqslant \PP( \uL \in \cL^{\tP_v} \,|\, \uX = \uone ) / (1 + (\bar{p}/p)^l)$, from which we complete the proof.
\end{proof}

Two remarks regarding Theorem~\ref{thm:stab+cond+bitwise+map} are the following.
\begin{enumerate}[nosep,leftmargin=*]
\item[--] \textit{Discussion on maximum matching}. Here we discuss why we expect that the number of vertices of $\tR_v$ should not differ significantly from that of $\tP_v$. Consider a vertex of $\tR_v$, call it $\underline{i}$, that has a large degree. This means that for this particular realization $\underline{i}$, there are many cycles surrounding $v$ in $\tG$, each of which contains variable nodes of degree two only and the sum of log-likelihood ratios is bounded. Intuitively we expect that each neighbor of $\underline{i}$ in $\tR_v$ should have a large degree, too. This is because each neighbor of $\underline{i}$ differs from $\underline{i}$ by a sign change along an index set corresponding to exactly one cycle. In principle, such a sign change due to a single cycle does not make its neighboring vertices significantly different from $\underline{i}$, so that, for each neighbor of $\underline{i}$, the number of cycles such that the corresponding sums of log-likelihood ratios are bounded should not be small, \ie, the degree of each neighbor of $\underline{i}$ should not be small. Thus we expect that each component of the realization graph $\tR_v$ should not be star-like, but should be ``regular enough'' so that the size of the vertex set of the maximum matching should be at least a \textit{constant} fraction of the size of vertex set of $\tR_v$.    
\item[--] \textit{Simplification to cyclic codes}. From the above analysis and discussion we can see that the construction of the realization graph $\tR$ and its maximum matching(s) only depends on the subgraph of the corresponding $\tG$ that is induced by all variable nodes of degree two and their connected check nodes. Variable nodes of degree three or higher in $\tG$ do not play a role. Consider then the following communication scenario. For a blocklength $n \in \NN$ and degree distribution pair $(\lambda_n, \rho_n)$, a Tanner graph $\tG$ is generated uniformly at random according to $(\lambda_n, \rho_n)$ and variable nodes are labelled by $[n]$ in an arbitrary but fixed way. Denote by $\uX = \uX(\tG)$ the codeword chosen uniformly at random from the code corresponding to $\tG$. The codeword $\uX$ is then transmitted over a BMS channel with $L$-density $\cc$ and we implement a \textit{genie-aided} bitwise MAP decoding algorithm which is described as follows. Let $\cI_2$ denote the index set of variable nodes of degree two. Moreover, the genie has access to $\uX_{[n] \backslash \cI_2}$, \ie, the genie knows the exact values of all bits belonging to variable nodes of degree three or higher. Therefore for any particular check node, call it $c$, the genie can perform pointwise multiplication of all bits belonging to variable nodes of degree at least three and are connected to $c$, and then stores the result at $c$. The genie does so for all check nodes. Consequently for the genie, all variable nodes of degree three or higher have served their purposes and can thus be eliminated from $\tG$. Equivalently the genie only needs to decode $\uX_{\cI_2}$ based on the graph $\tG_2$ that is the induced subgraph of $\tG$ containing all variable nodes of degree two and all connected check nodes. For each of these connected check nodes, the genie stores the ``parity'' stemming from the connected variable nodes of degree three or higher in $\tG$. Clearly, if the probability of bit error for the genie-aided bitwise MAP decoding is lower bounded away from $0$, then this is so for the standard bitwise MAP decoding. Codes whose Tanner graphs only contain degree-two variable nodes are often referred to as \textit{cyclic} codes in the literature~\cite{Zem01}, and we conclude that the stability analysis of such cyclic codes under bitwise MAP decoding ``directly translates'' to that of codes specified by general degree distributions according to the above simplification. 
\end{enumerate} 

\begin{example}
\label{exa:pattern+mapping2}
Consider the Tanner graph $\tG$ shown in Figure~\ref{fig:stab+example2}. There are three cycles consisting of variable node $1$. Let $\cI_1 = \{1,2\}$, $\cI_2 = \{1,3,4\}$ and $\cI_3 = \{1,3,5\}$; these give the index sets of variable nodes which correspond to the three cycles. 

\begin{figure}[H]
\centering
\includegraphics{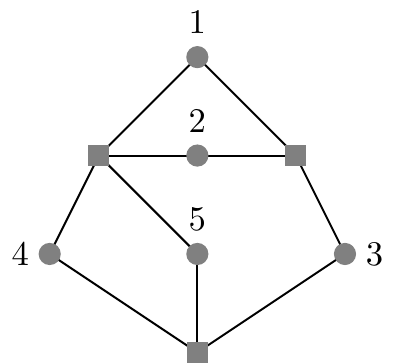}
\caption{Tanner graph $\tG$ in Example~\ref{exa:pattern+mapping2}.}
\label{fig:stab+example2}
\end{figure}

We adopt the same notation as in Example~\ref{exa:pattern+mapping} to construct the realization graph $\tR$ as well as the pattern graph $\tP$ associated with variable node $1$ accordingly. For instance, consider $\underline{20} = (-1,1,-1,1,1)$. If we sum the entries of $\underline{20}$ along the index sets $\cI_1$, $\cI_2$ and $\cI_3$, respectively, the results are $0$, $-1$ and $-1$, the absolute values of which are bounded by $1$. By multiplying its entries corresponding to $\cI_1$, $\cI_2$ and $\cI_3$, respectively, by $-1$ while keeping the rest entries unchanged, we see that $\underline{20}$ has $\underline{1}$, $\underline{2}$ and $\underline{12}$ as its neighbors in $\tR$. Moreover, if $\underline{20}$ gives rise to correct bitwise MAP decoding of variable node $1$ in $\tG$, then $\underline{1}$, $\underline{2}$ and $\underline{12}$ give rise to incorrect bitwise MAP decoding of $1$, and vice versa. The resulting graph $\tR$ together with a pattern graph $\tP$ is shown in Figure~\ref{fig:stab+pairing2}. 

\begin{figure}[H]
\centering
\includegraphics{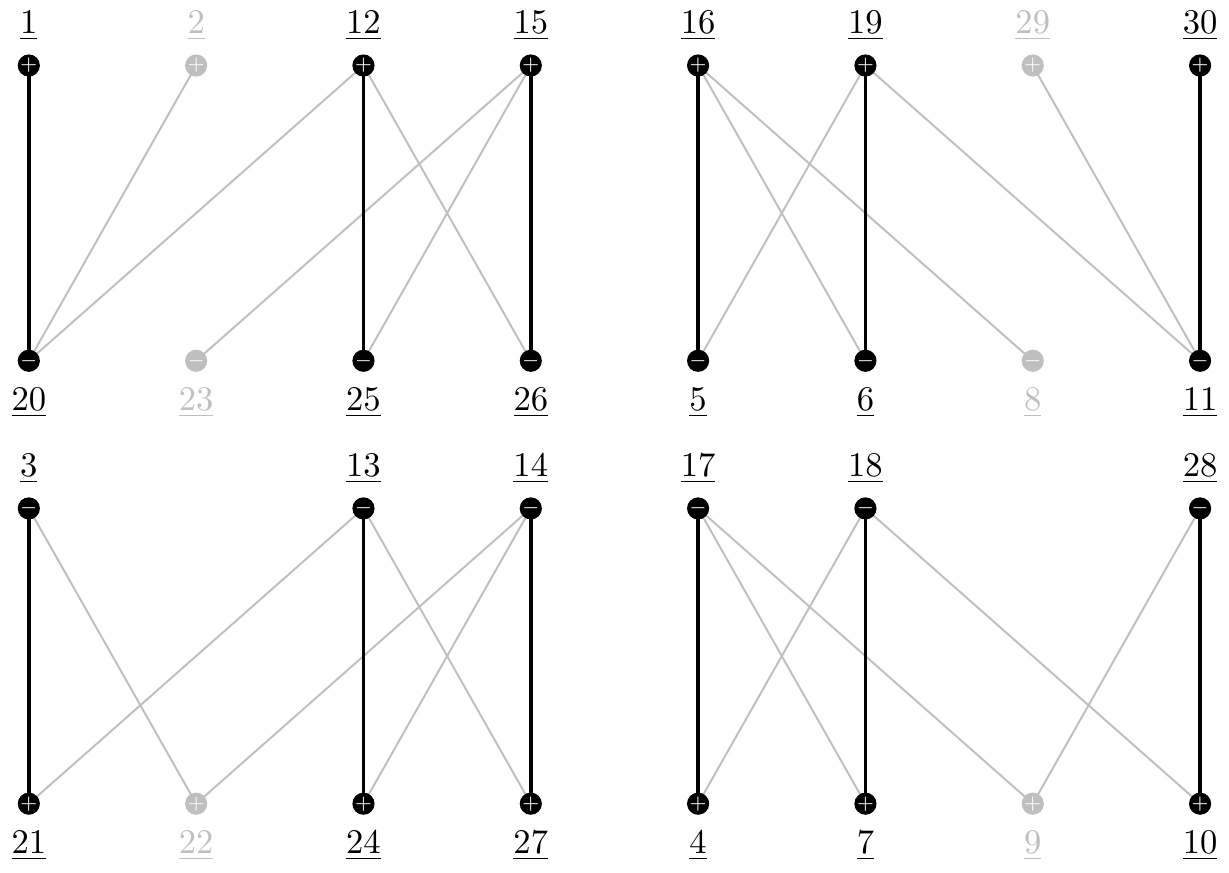}
\caption{Graph $\tR$ and an associated maximum matching in Example~\ref{exa:pattern+mapping2}.}
\label{fig:stab+pairing2}
\end{figure}

We emphasize that the $+/-$ labels shown in Figure~\ref{fig:stab+pairing2} do not necessarily correspond to the actual bitwise MAP decoding signs for variable node $1$. They are chosen only subject to the constraint that every pair of adjacent vertices in $\tR$ can only have different decoding signs for variable node $1$. 
\end{example}

\section{Conclusion}
\label{sec:stab+concl}
In this paper we performed decoding analysis of LDPC codes when transmission takes place over a BMS channel. We addressed the problem of universality of capacity-achieving LDPC codes under BP decoding, and showed that many existing such codes are not universally capacity-achieving under BP decoding. We then revealed that the key to this argument is the stability threshold under BP decoding. We further showed that the same stability threshold applies to blockwise or bitwise MAP decoding as well. We  presented how stability can determine an upper bound on the corresponding blockwise or bitwise MAP threshold, revealing the operational significance of the stability threshold.


\section*{Acknowledgment}

This work was done while W.~Liu was affiliated with EPFL and was supported in part by Swiss National Science Foundation under grant No.~200021-166106. W.~Liu would also like to acknowledge Tom Richardson for very helpful discussions.


\appendix

\section{Appendix to Section \ref{sec:ca}}
\label{appdix:ca}
In this appendix we collect the proofs of lemmas in Section~\ref{sec:ca}.


\subsection{Appendix to Section \ref{subsec:ca+exa}}
\label{appdix:inter+ca}

\begin{proof}[Proof of Lemma~\ref{lem:poi+f1+limit}]
For every $N \geqslant 2$, the first-order derivative of $f_{(N)}$ can be computed as
\begin{align}
f_{(N)} ^{\prime} (x) &= \frac{\epsilon}{H_{N-1}} \sum_{i=1}^{N-1} \frac{1}{i} \big( ( 1 - \ee^{- \alpha x} )^i \big)^{\prime} = \sum_{i=1}^{N-1} \ee^{-\alpha x} ( 1 - \ee^{-\alpha x} )^{i-1} \nonumber \\
\label{eqn:poi+1st+der}
&= \ee^{- \alpha x} + \sum_{i=2}^{N-1} \ee^{- \alpha x} ( 1 - \ee^{-\alpha x} )^{i-1} = 1 - ( 1 - \ee^{-\alpha x} )^{N-1}.
\end{align}
Denoting by $\gamma$ the Euler-Mascheroni constant and for $N \geqslant 2$, we recall the bounds on $H_N$ given by
\begin{equation}
\label{eqn:harmonic+bound}
\ln(N) + \gamma < H_N = \sum_{k=1}^N \frac{1}{k} < \ln(N) + \gamma + \frac{1}{2N}.
\end{equation}
Therefore we can write
\begin{equation}
\label{eqn:poi+bound+exp}
\frac{\ee^{- (\gamma + 1/(2(N-1))) x / \epsilon}}{(N-1)^{x/\epsilon}} \leqslant \ee^{-\alpha x} \leqslant \frac{\ee^{- \gamma x / \epsilon}}{(N-1)^{x / \epsilon}}.
\end{equation}
Take any $\kappa \in {[0,\epsilon)}$ and let $x \in [0,\kappa]$. Using $1 - x \leqslant \ee^{-x}$, we have $1 - \ee^{-\alpha x} \leqslant \exp(-\ee^{-\alpha x}) \leqslant \exp(-\ee^{-\alpha \kappa})$ so that
\begin{equation*}
(1 - \ee^{- \alpha x})^{N-1} \leqslant \exp\big(-\ee^{-(\gamma+1/2)\kappa/\epsilon} (N-1)^{1 - \kappa/\epsilon}\big),
\end{equation*}
the right-hand side of which tends to $0$ as $N$ tends to $\infty$. We conclude that the sequence of first-order derivatives on ${[0,\epsilon)}$ converges to $1$ as $N$ tends to $\infty$. Next we show that the sequence of first-order derivatives on ${(\epsilon,1]}$ converges to $0$ for $x \in {(\epsilon,1]}$. To this end, fix any $\kappa \in {(\epsilon,1]}$. Since $\ln(1-y) \geqslant -y - y^2$ for $y \in [0,1/2]$, for any $n \in \NN$ we have $(1 - y)^n = \exp(n \ln(1-y)) \geqslant \exp(-n(y+y^2))$. Using the upper bound on $\ee^{-\alpha x}$ in (\ref{eqn:poi+bound+exp}), for any $x \in [\kappa,1]$ we have the inequalities $\ee^{-\alpha x} \leqslant \ee^{-\gamma x/\epsilon} (N-1)^{-x/\epsilon} \leqslant \ee^{-\gamma}(N-1)^{-\kappa/\epsilon} \leqslant \ee^{-\gamma}(N-1)^{-1} \leqslant 1/2$ for all $N \geqslant 3$. Thus
\begin{equation*}
(1 - \ee^{- \alpha x})^{N-1} \geqslant \exp\big(-\ee^{-\gamma}(N-1)^{1-\kappa/\epsilon} - \ee^{-2\gamma} (N-1)^{1-2\kappa/\epsilon}\big),
\end{equation*}
the right-hand side of which tends to $1$ as $N$ tends to $\infty$ and this proves the claim. Finally let us consider the case where $x = \epsilon$. By (\ref{eqn:poi+1st+der}) we can write $f_{(N)} ^{\prime} (\epsilon) = 1 - ( 1 - \ee^{-H_{N-1}} )^{N-1}$. Using (\ref{eqn:poi+bound+exp}) we have
\begin{equation*}
(1 - \ee^{-H_{N-1}}) \geqslant \big(1 - \ee^{-\gamma} / (N-1)\big)^{N-1},
\end{equation*}
the right-hand side of which has limit $\ee^{-\ee^{-\gamma}}$. Besides, by (\ref{eqn:poi+1st+der}) we have $1 - \ee^{-H_{N-1}} \leqslant 1 - \ee^{-\gamma-1/(2(N-1))} / (N-1)$ so that
\begin{align*}
( 1 - \ee^{-H_{N-1}} )^{N-1} &\leqslant \big( 1 - \ee^{-\gamma-1/(2(N-1))} / (N-1) \big)^{N-1} \\
&\leqslant \exp\big(-\ee^{-\gamma-1/(2(N-1))}  \big) \\
&= \ee^{-\ee^{-\gamma}} \exp\Big( \ee^{-\gamma}\big(1 - \ee^{-1/(2(N-1))}\big) \Big) \\
&\leqslant \ee^{-\ee^{-\gamma}} \exp\big( \ee^{-\gamma} / (2(N-1)) \big),
\end{align*}
the right-hand side of which has limit $\ee^{-\ee^{-\gamma}}$ as well, and this completes the proof of Lemma~\ref{lem:poi+f1+limit}.
\end{proof}

\begin{proof}[Proof of Lemma~\ref{lem:ca+rr+mu}]
We follow the lead of~\cite{Sho99} and bound the quantity $\binom{\alpha}{N}(-1)^{N-1}$ by
\begin{equation}
\label{eqn:c0+c1+lower+bound}
\frac{c_0 \alpha}{N^{\alpha+1}} \leqslant \binom{\alpha}{N} (-1)^{N-1} \leqslant \frac{c_1 \alpha}{N^{\alpha+1}},
\end{equation}
where $c_0 = c_0(\alpha,N)$ and $c_1 = c_1(\alpha,N)$ are defined as 
\begin{align}
\label{eqn:ca+rr+c0}
c_0 &= (1-\alpha)^{\pi^2/6} \exp\big( \alpha (\pi^2/6 - \gamma + 1/(2N)) \big); \\
\label{eqn:ca+rr+c1}
c_1 &= (1-\alpha) \exp\big( \alpha (1 - \gamma + 1/N) \big).
\end{align}
Notice that since $\alpha$ tends to $0$ as $N$ tends to $\infty$, it follows that both $c_0$ and $c_1$ tend to $1$ as $N$ tends to $\infty$. Moreover,
\begin{equation}
\label{eqn:rr+one+over+f}
1 \big/ f_{(N)} ^{\prime} (0) = 1 \big/ \big( \lambda_{(N)} ^{\prime}(0) \rho_{(N)} ^{\prime}(1) \big) = 1 - \frac{N}{\alpha} \binom{\alpha}{N} (-1)^{N-1}.
\end{equation}
Therefore (\ref{eqn:c0+c1+lower+bound}) and (\ref{eqn:rr+one+over+f}) lead to 
\begin{equation*}
1 - \frac{c_1}{N^{\alpha}} \leqslant 1 \big/ f_{(N)} ^{\prime} (0) \leqslant 1 - \frac{c_0}{N^{\alpha}}.
\end{equation*}
By definition $\alpha = \ln(1/\bar{\epsilon}) / \ln(N)$ so that $N^{\alpha} = 1/\bar{\epsilon}$. Thus we get $1 - c_1 (1 - \epsilon) \leqslant 1 / f_{(N)} ^{\prime} (0) \leqslant 1 - c_0 (1 - \epsilon)$. Taking the limit $N \to \infty$ completes the proof of Lemma~\ref{lem:ca+rr+mu}.
\end{proof}


\subsection{Appendix to Section \ref{subsec:ca+deg2}}
\label{appdix:ca+deg2}

\begin{proof}[Proof of Lemma \ref{lem:ca+frac+deg2}]
We start by showing $\lim_{N \to \infty} L_2 ^{(N)} = 1/2$ in (\ref{eqn:ca+poi+l2}) for the heavy-tail Poisson sequence. By definition we can write
\begin{equation}
\label{eqn:rr+poi+l2}
L_2 ^{(N)} = \frac{\lambda_2 ^{(N)} / 2}{\sum_{i=1}^{N-1} \lambda_{i+1} ^{(N)} / (i+1)}.
\end{equation}
Specializing to the heavy-tail Poisson sequence by (\ref{eqn:poi+lambda+rho+def}), we have
\begin{align*}
\sum_{i=1}^{N-1} \frac{\lambda_{i+1} ^{(N)}}{i+1} &= \frac{1}{H_{N-1}} \sum_{i=1}^{N-1} \frac{1}{(i+1)i} = \frac{1}{H_{N-1}} \sum_{i=1}^{N-1} \Big( \frac{1}{i} - \frac{1}{i+1} \Big) \\
&= \frac{1}{H_{N-1}} \Big( 1 - \frac{1}{N} \Big).
\end{align*}
Using $\lambda_2 ^{(N)} = 1 / H_{N-1}$, we can get $\lim_{N \to \infty} L_2 ^{(N)} = \lim_{N \to \infty} (1/2) / (1 - 1/N) = 1/2$, which verifies (\ref{eqn:ca+poi+l2}) for the heavy-tail Poisson sequence.

Next we show that (\ref{eqn:ca+rr+l2}) holds for the right-regular sequence. By (\ref{eqn:rr+lambda+rho+def}) and for $1 \leqslant i \leqslant N-1$, we can write
\begin{equation}
\label{eqn:rr+lambda+def+new}
\lambda_{\alpha} ^{(N)} \lambda_{i+1} ^{(N)} = \binom{\alpha}{i} (-1)^{i-1},
\end{equation}
where $\lambda_{\alpha} ^{(N)} \coloneqq 1 - (N/\alpha) \binom{\alpha}{N} (-1)^{N-1}$. We recall from the proof of Lemma~\ref{lem:ca+rr+mu} that for every $i \geqslant 1$, it holds that
\begin{equation}
\label{eqn:rr+binom+new}
\frac{c_0 \alpha}{i^{\alpha+1}} \leqslant \binom{\alpha}{i} (-1)^{i-1} \leqslant \frac{c_1 \alpha}{i^{\alpha+1}}  ,
\end{equation}
where $c_0 = c_0(\alpha, i)$ and $c_1 = c_1(\alpha, i)$ are specified in (\ref{eqn:ca+rr+c0}) and (\ref{eqn:ca+rr+c1}) in the proof of Lemma~\ref{lem:ca+rr+mu}, respectively. By defining
\begin{equation}
\label{eqn:rr+cplus+cminus}
c^+ = (1-\alpha) \ee^{\alpha(3-\gamma)} \text{ and } c^- = (1-\alpha)^2 \ee^{\alpha(1-\gamma)},
\end{equation}
we can then obtain that $c^- \leqslant c_0(\alpha, i) \leqslant c^+$ and $c^- \leqslant c_1(\alpha, i) \leqslant c^+$ for all $i$. Notice that both $c^+$ and $c^-$ converge to $1$ as $N$ tends to $\infty$, since $\alpha = \alpha(\epsilon, N)$ tends to $0$ as $N$ tends to $\infty$; see (\ref{eqn:ca+rr+alpha+def}). Now combining (\ref{eqn:rr+lambda+def+new}), (\ref{eqn:rr+binom+new}) and (\ref{eqn:rr+cplus+cminus}) and after some manipulation, we get
\begin{equation}
\label{eqn:rr+new+lambda+bound}
\frac{c^- \alpha}{i^{\alpha+1}} \leqslant \lambda_{\alpha} ^{(N)} \lambda_{i+1} ^{(N)} \leqslant \frac{c^+ \alpha}{i^{\alpha+1}}.
\end{equation}
Now we use (\ref{eqn:rr+poi+l2}) to get
\begin{equation}
\label{eqn:rr+new+l2}
L_2 ^{(N)} = \frac{\lambda_{\alpha} ^{(N)} \lambda_2 ^{(N)} / 2}{\sum_{i=1}^{N-1} \lambda_{\alpha} ^{(N)} \lambda_{i+1} ^{(N)} / (i+1)} = \frac{\alpha / 2}{S_{\alpha} ^{(N)}},
\end{equation}
where in (\ref{eqn:rr+new+l2}), $S_{\alpha} ^{(N)}$ is defined as 
\begin{equation*}
S_{\alpha} ^{(N)} = \sum_{i=1}^{N-1} \frac{\lambda_{\alpha} ^{(N)} \lambda_{i+1} ^{(N)}}{i+1}.
\end{equation*}
It remains to work with $S_{\alpha} ^{(N)}$. By (\ref{eqn:rr+new+lambda+bound}) we can write
\begin{equation}
\label{eqn:bounds+hat+S}
c^- \alpha \hat{S}_{\alpha} ^{(N)} \leqslant S_{\alpha} ^{(N)} \leqslant c^+ \alpha \hat{S}_{\alpha} ^{(N)},
\end{equation}
where
\begin{equation}
\label{eqn:rr+sum+identity}
\hat{S}_{\alpha} ^{(N)} \coloneqq \sum_{i=1}^{N-1} \frac{1}{(i+1)i^{\alpha+1}} = \frac{1}{2} + \sum_{i=2}^{N-1} \frac{1}{(i+1)i^{\alpha+1}}.
\end{equation}
In the remaining of the proof, we will get uniform in $\alpha$ estimate of the sum on the right-hand side of (\ref{eqn:rr+sum+identity}). This will lead us to the claim in (\ref{eqn:ca+rr+l2}) in Lemma~\ref{lem:ca+frac+deg2}. To this end, we let $N \geqslant 3$ and for $i \in \{2,3,\dots,N-1\}$, we write
\begin{equation}
\label{eqn:expansion+i+1}
\frac{1}{i+1} = \frac{1}{i} \sum_{w=0}^{\infty} \frac{(-1)^w}{i^w} = \sum_{w=1}^{\infty} \frac{(-1)^{w-1}}{i^w}.
\end{equation} 
Let $W \geqslant 3$ be an odd integer and be fixed. Then we can get from (\ref{eqn:expansion+i+1}) that
\begin{equation}
\label{eqn:rr+W+formula}
\frac{1}{i+1} \leqslant \sum_{w=1}^W \frac{(-1)^{w-1}}{i^w} = \sum_{w=1}^{(W-1)/2} \Big( \frac{1}{i^{2w-1}} - \frac{1}{i^{2w}} \Big) + \frac{1}{i^{W}}.
\end{equation}
Plugging (\ref{eqn:rr+W+formula}) into the right-hand side of (\ref{eqn:rr+sum+identity}), the partial sum from $i=3$ to $N-1$ can be upper bounded by 
\begin{align}
\sum_{i=3}^{N-1} \frac{1}{(i+1) i^{\alpha+1}} &\leqslant \sum_{w=1}^{(W-1)/2} \sum_{i=3}^{N-1} \Big( \frac{1}{i^{\alpha+2w}} - \frac{1}{i^{\alpha+2w+1}} \Big) + \sum_{i=3}^{N-1} \frac{1}{i^{\alpha+W+1}} \nonumber \\
&\leqslant \sum_{w=1}^{(W-1)/2} \int_2^{N-1} \Big( \frac{1}{x^{\alpha+2w}} - \frac{1}{x^{\alpha+2w+1}} \Big) \diff x + \int_2^{N-1} \frac{1}{x^{\alpha+W+1}} \diff x \nonumber \\
&= \sum_{w=1}^W \int_2^{N-1} \frac{(-1)^{w-1}}{x^{\alpha+w+1}} \diff x \nonumber \\
\label{eqn:sum+upper+bound+rr}
&= \sum_{w=1}^W \frac{(-1)^{w-1}}{\alpha+w} \big( 2^{-(\alpha+w)} - (N-1)^{-(\alpha+w)} \big).
\end{align}
Now we notice that for fixed $W$, each individual term in the sum on the right-hand side of (\ref{eqn:sum+upper+bound+rr}) converges to $(-1)^{w-1} 2^{-w} / w$, as $N$ tends to $\infty$ so that $\alpha$ tends to $0$. By (\ref{eqn:rr+sum+identity}) and (\ref{eqn:sum+upper+bound+rr}) we can get
\begin{align*}
\limsup_{N \to \infty} \hat{S}_{\alpha} ^{(N)} &\leqslant \frac{1}{2} + \frac{1}{6} + \lim_{N \to \infty} \sum_{w=1}^W \frac{(-1)^{w-1}}{\alpha+w} \big( 2^{-(\alpha+w)} - (N-1)^{-(\alpha+w)} \big) \\
&= \frac{2}{3} + \sum_{w=1}^W \frac{(-1)^{w-1} 2^{-w}}{w}.
\end{align*}
Since $W$ is arbitrary, by letting $W$ tend to $\infty$ we obtain that
\begin{equation}
\label{eqn:rr+S+limsup}
\limsup_{N \to \infty} \hat{S}_{\alpha} ^{(N)} \leqslant \frac{2}{3} + \sum_{w=1}^{\infty} \frac{(-1)^{w-1} 2^{-w}}{w} = \frac{2}{3} + \ln(3/2) = \ln(3 \ee^{2/3} / 2).
\end{equation}
Now by (\ref{eqn:rr+new+l2}) and the right-hand side of (\ref{eqn:bounds+hat+S}) we have
\begin{equation}
\label{eqn:rr+l2+lower+bound+new}
L_2 ^{(N)} = \frac{\alpha/2}{S_{\alpha} ^{(N)}} \geqslant \frac{\alpha/2}{c^+ \alpha \hat{S}_{\alpha} ^{(N)}} = \frac{1/2}{c^+ \hat{S}_{\alpha} ^{(N)}}.
\end{equation}
Recall that $c^+$ converges to $1$ as $N$ tends to $\infty$. Therefore the lower bound in (\ref{eqn:ca+rr+l2}) in Lemma~\ref{lem:ca+frac+deg2} follows, since (\ref{eqn:rr+S+limsup}) and (\ref{eqn:rr+l2+lower+bound+new}) lead to
\begin{equation*}
\liminf_{N \to \infty} L_2 ^{(N)} \geqslant \frac{1/2}{\ln(3 \ee^{2/3} / 2)} = \frac{1}{\ln(9 \ee^{4/3} / 4)}.
\end{equation*}
We next show that the upper bound in (\ref{eqn:ca+rr+l2}) holds by a similar argument. We use (\ref{eqn:expansion+i+1}) to get
\begin{equation*}
\frac{1}{i+1} \geqslant \sum_{w=1}^{W+1} \frac{(-1)^{w-1}}{i^w} = \sum_{w=1}^{(W+1)/2} \Big( \frac{1}{i^{2w-1}} - \frac{1}{i^{2w}} \Big),
\end{equation*}
which then leads to
\begin{align}
\sum_{i=2}^{N-1} \frac{1}{(i+1) i^{\alpha+1}} &\geqslant \sum_{w=1}^{(W+1)/2} \sum_{i=2}^{N-1} \Big( \frac{1}{i^{\alpha+2w}} - \frac{1}{i^{\alpha+2w+1}} \Big) \nonumber \\
&\geqslant \sum_{w=1}^{(W+1)/2} \int_2^N \Big( \frac{1}{x^{\alpha+2w}} - \frac{1}{x^{\alpha+2w+1}} \Big) \diff x \nonumber \\
&= \sum_{w=1}^{W+1} \int_2^N \frac{(-1)^{w-1}}{x^{\alpha+w+1}} \diff x \nonumber \\
\label{eqn:sum+lower+bound+rr}
&= \sum_{w=1}^{W+1} \frac{(-1)^{w-1}}{\alpha+w} \big( 2^{-(\alpha+w)} - N^{-(\alpha+w)} \big).
\end{align}
For fixed $W$, every term in the sum on the right-hand side of (\ref{eqn:sum+lower+bound+rr}) converges to $(-1)^{w-1}2^{-w} / w$, as $N$ tends to $\infty$. Thus by (\ref{eqn:rr+sum+identity}) and (\ref{eqn:sum+lower+bound+rr}) we can derive that
\begin{align*}
\liminf_{N \to \infty} \hat{S}_{\alpha} ^{(N)} &\geqslant \frac{1}{2} + \lim_{N \to \infty} \sum_{w=1}^{W+1} \frac{(-1)^{w-1}}{\alpha+w} \big( 2^{-(\alpha+w)} - N^{-(\alpha+w)} \big) \\
&= \frac{1}{2} + \sum_{w=1}^{W+1} \frac{(-1)^{w-1} 2^{-w}}{w}.
\end{align*}
By letting $W$ tend to $\infty$ we arrive at
\begin{equation}
\label{eqn:rr+S+liminf}
\liminf_{N \to \infty} \hat{S}_{\alpha} ^{(N)} \geqslant \frac{1}{2} + \sum_{w=1}^{\infty} \frac{(-1)^{w-1} 2^{-w}}{w} = \frac{1}{2} + \ln(3/2) = \ln(3 \sqrt{\ee}/2).
\end{equation} 
Now by (\ref{eqn:rr+new+l2}) and the left-hand side of (\ref{eqn:bounds+hat+S}) we have
\begin{equation}
\label{eqn:rr+l2+upper+bound+new}
L_2 ^{(N)} = \frac{\alpha/2}{S_{\alpha} ^{(N)}} \leqslant \frac{\alpha/2}{c^- \alpha \hat{S}_{\alpha} ^{(N)}} = \frac{1/2}{c^- \hat{S}_{\alpha} ^{(N)}}.
\end{equation}
Since $c^-$ converges to $1$ as $N$ tends to $\infty$, the upper bound in (\ref{eqn:ca+rr+l2}) follows since (\ref{eqn:rr+S+liminf}) and (\ref{eqn:rr+l2+upper+bound+new}) lead to
\begin{equation*}
\limsup_{N \to \infty} L_2 ^{(N)} \leqslant \frac{1/2}{\ln(3 \sqrt{\ee} / 2)} = \frac{1}{\ln(9\ee/4)},
\end{equation*}
which completes the proof.
\end{proof}

\section{Appendix to Section \ref{sec:univ+bp}}
\label{appdix:bp}
In this appendix we give the proofs of some lemmas and theorems in Section~\ref{sec:univ+bp}.


\subsection{Appendix to Section \ref{subsec:univ}}
\label{appdix:generic}

\begin{proof}[Proof of Theorem \ref{thm:univ+generic}]
Since $\lim_{\ell \to \infty} \ln( \error(\cc^{\varoast \ell}) ) / \ell = \ln( \bhatta(\cc) )$, the condition $\bhatta(\cc) \mu_{\infty} > 1$ implies that there exists an $l = l(\cc, \mu_{\infty}) \in \NN$ such that for all $\ell \geqslant l$, we have $\error(\cc^{\varoast \ell}) \mu_{\infty}^{\ell} > 1$. In the following we fix such an $l$ and we further define two sequences $\{\beta(\ell)\}_{0 \leqslant \ell \leqslant l}$ and $\{\theta(\ell)\}_{0 \leqslant \ell \leqslant l}$, the operational significance of which will be clear as we proceed. More precisely, for $\ell \in [l]$, let
\begin{equation}
\label{eqn:univ+beta+def}
\beta(\ell) = \mu_{\infty} ^{\ell} - \mu_{\infty} ^{\ell} / 2^{l - \ell + 2},
\end{equation}
and
\begin{equation}
\label{eqn:univ+theta+def}
\theta(\ell) = 0.5 \mu_{\infty} ^{\ell+1} / 2^{l - \ell + 2}.
\end{equation}
Furthermore, let $\beta(0) = 1$ and $\theta(0) = 0.5 \mu_{\infty} / 2^{l+1}$. Then for the given $\kappa \in (0,1)$, we define
\begin{equation}
\label{eqn:univ+xi}
\xi = \min\Big\{ \kappa/\mu_{\infty}^{l+1}, \min\big\{ 2 \theta(\ell) / \beta(\ell)^2 \, | \,  0 \leqslant \ell \leqslant l \big\}  \Big\}.
\end{equation} 

We first consider the initial density $\bb_N(0) = \beta(0) \epsilon \Delta_0 + (1 - \beta(0) \epsilon) \Delta_{\infty}$ for  $N \geqslant 2$, where $\epsilon \in {(0, \xi]} \subseteq [0, \kappa]$ is arbitrary. Without loss of generality let us assume that $M = 1$ so that by (\ref{eqn:generic+univ+g2}) we have
\begin{equation*}
g_{(N)} ^{\prime \prime}(\epsilon) \geqslant -1.
\end{equation*}
Consider the $1$-st iteration of density evolution given by $\bb_N(1) = \TT_{(N)}(\bb_N(0))$, \ie, after one iteration of density evolution, we get the $L$-density $\bb_N(1)$ by applying $\TT_{(N)}$ to $\bb_N(0)$, which is given by
\begin{equation*}
\bb_N(1) = g_{(N)}( \beta(0) \epsilon ) \cc + \big( 1 - g_{(N)}( \beta(0) \epsilon ) \big) \Delta_{\infty}.
\end{equation*}
Since $g_{(N)}(0) = 0$, by Taylor's formula we can write for each $N$ that
\begin{equation}
\label{eqn:univ+taylor}
g_{(N)}( \beta(0) \epsilon ) = \beta(0) \epsilon \mu_N + \frac{1}{2} g_{(N)} ^{\prime\prime} (\epsilon_N ^*) \beta(0)^2 \epsilon^2,
\end{equation}
where $\epsilon_N ^* \in [0,\epsilon]$. Since $\epsilon \leqslant \xi$, by the definition of $\xi$ in (\ref{eqn:univ+xi}) we have the inequality $\epsilon \leqslant 2 \theta(0) / \beta(0)^2$ so that $- \beta(0)^2 \epsilon^2 / 2 \geqslant - \theta(0) \epsilon$. Now since the second-order derivative at $\epsilon_N ^*$ is bounded from below by $-1$, (\ref{eqn:univ+taylor}) leads to
\begin{equation}
\label{eqn:univ+g+beta0+taylor}
g_{(N)}( \beta(0) \epsilon ) \geqslant \beta(0)\epsilon \mu_N - \theta(0)\epsilon = \big( \beta(0)\mu_N - \theta(0) \big) \epsilon.
\end{equation}
Furthermore, by $\lim_{N \to \infty} \mu_N = \mu_{\infty}$ we have $\lim_{N \to \infty} \beta(0)\mu_N - \theta(0) = \beta(0)\mu_{\infty} - \theta(0)$, so that there exists an integer $N_1 = N_1(\xi)$ such that for all $N \geqslant N_1$, it holds that
\begin{align*}
\beta(0)\mu_N - \theta(0) &\geqslant \beta(0) \mu_{\infty} - \theta(0) - \theta(0) \\
&= \beta(0) \mu_{\infty} - 0.5 \mu_{\infty}/2^{l+1} - 0.5 \mu_{\infty}/2^{l+1} \\
&= \mu_{\infty} - \mu_{\infty}/2^{l+1} = \beta(1),
\end{align*}
where the last equality is due to the definition of $\beta(1)$ in (\ref{eqn:univ+beta+def}). Besides, the condition $\epsilon \leqslant \xi$ implies that $\epsilon \leqslant \kappa / \mu_{\infty}$, which together with the inequality $\beta(1) < \mu_{\infty}$, implies that $0 < \beta(1) \epsilon \leqslant \kappa$. Consequently for all $N \geqslant N_1$, we have
\begin{equation*}
\bb_N(1) \leftarrowtriangle \beta(1)\epsilon \cc + \big( 1 - \beta(1)\epsilon \big) \Delta_{\infty} \eqqcolon \tilde{\bb}_N(1).
\end{equation*}
Thus we obtain the upgraded $L$-density $\tilde{\bb}_N(1)$ with respect to $\bb_{N}(1)$ for all $N \geqslant N_1$, after the $1$-st iteration of density evolution initialized by $\bb_{N}(0)$. Consider next the $2$-nd iteration of density evolution for any $N \geqslant N_1$. We have
\begin{align*}
\bb_N(2) &= \TT_{(N)}\big( \bb_N(1) \big) \leftarrowtriangle \TT_{(N)}\big( \tilde{\bb}_N(1) \big) \\
&= g_{(N)}( \beta(1) \epsilon ) \cc^{\varoast 2} + \big( 1 - g_{(N)}( \beta(1) \epsilon ) \big) \Delta_{\infty}.
\end{align*}
Arguing similarly to (\ref{eqn:univ+g+beta0+taylor}), by definition of $\xi$ in (\ref{eqn:univ+xi}) we get $\epsilon \leqslant 2 \theta(1) / \beta(1)^2$ so that $- \beta(1)^2 \epsilon^2 / 2 \geqslant - \theta(1) \epsilon$. This leads to
\begin{equation*}
g_{(N)}( \beta(1) \epsilon ) \geqslant \beta(1)\epsilon \mu_N - \theta(1)\epsilon = \big( \beta(1)\mu_N - \theta(1) \big) \epsilon.
\end{equation*}
Therefore by $\lim_{N \to \infty} \mu_N = \mu_{\infty}$ we have 
\begin{align*}
\lim_{N \to \infty} \beta(1) \mu_N - \theta(1) &= \beta(1) \mu_{\infty} - \theta(1) \\
&= \mu_{\infty} ^2 - \mu_{\infty} ^2 / 2^{l+1} - \theta(1).
\end{align*}
Thus there exists an integer $N_2 = N_2(\xi)$ so that for all $N \geqslant \max\{N_1, N_2\}$, it holds that
\begin{align*}
\beta(1)\mu_N - \theta(1) &\geqslant \mu_{\infty} ^2 - \mu_{\infty} ^2 / 2^{l+1} - \theta(1) - \theta(1) \\
&= \mu_{\infty} ^2 - \mu_{\infty} ^2 / 2^{l+1} - 0.5 \mu_{\infty} ^2 / 2^{l+1} - 0.5 \mu_{\infty} ^2 / 2^{l+1} \\
&= \mu_{\infty} ^2 - \mu_{\infty} ^2 / 2^{l+1} - \mu_{\infty} ^2 / 2^{l+1} \\
&= \mu_{\infty} ^2 - \mu_{\infty} ^2 / 2^l = \beta(2).
\end{align*}
Moreover, $\epsilon \leqslant \xi$ implies that $\epsilon \leqslant \kappa / \mu_{\infty} ^2$, which together with $\beta(2) < \mu_{\infty} ^2$, implies that $0 < \beta(2) \epsilon \leqslant \kappa$. Therefore for all $N \geqslant \max\{N_1, N_2\}$, we have
\begin{equation*}
\bb_N(2) \leftarrowtriangle \beta(2) \epsilon \cc^{\varoast 2}+ \big( 1 - \beta(2)\epsilon \big) \Delta_{\infty} \eqqcolon \tilde{\bb}_N(2).
\end{equation*}
Consequently we can obtain the upgraded $L$-density $\tilde{\bb}_N(2)$ with respect to $\bb_{N}(2)$ for all $N \geqslant \max\{N_1, N_2\}$, after the $2$-nd iteration of density evolution initialized by $\bb_{N}(0)$.

In general, we can obtain similar arguments as above so that after $\ell \leqslant l$ iterations of density evolution initialized by $\bb_{N}(0)$, there exist integers $\{N_{\ell^{\prime}}\}_{1 \leqslant \ell^{\prime} \leqslant \ell}$ such that for all $N \geqslant \max_{1 \leqslant \ell^{\prime} \leqslant \ell} \{N_{\ell^{\prime}}\}$, we have
\begin{equation*}
\bb_N(\ell) \leftarrowtriangle \beta(\ell) \epsilon \cc^{\varoast \ell} + \big( 1 - \beta(\ell)\epsilon \big) \Delta_{\infty} \eqqcolon \tilde{\bb}_N(\ell).
\end{equation*}
Specializing to the case $\ell = l$ gives for all $N \geqslant \max_{1 \leqslant \ell^{\prime} \leqslant l} \{N_{\ell^{\prime}}\}$ that 
\begin{align}
\bb_N(l) &\leftarrowtriangle \beta(l) \epsilon \cc^{\varoast l} + \big( 1 - \beta(l)\epsilon \big) \Delta_{\infty} \nonumber \\
\label{eqn:univ+b+l}
&= \big( \mu_{\infty} ^l  - \mu_{\infty} ^l / 4 \big) \epsilon \cc^{\varoast l} + \big( 1 - (\mu_{\infty} ^l  - \mu_{\infty} ^l / 4) \epsilon \big) \Delta_{\infty}  \\
&\eqqcolon \tilde{\bb}_N(l) \nonumber.
\end{align}
For the $(l+1)$-th iteration we get
\begin{align*}
\bb_N(l+1) &= \TT_{(N)}\big( \bb_N(l) \big) \leftarrowtriangle \TT_{(N)}\big( \tilde{\bb}_N(l) \big) \\
&= g_{(N)}( \beta(l) \epsilon ) \cc^{\varoast (l+1)} + \big( 1 - g_{(N)}( \beta(l) \epsilon ) \big) \Delta_{\infty}.
\end{align*}
By the definition of $\xi$ in (\ref{eqn:univ+xi}) we have $\epsilon \leqslant 2 \theta(l) / \beta(l)^2$ so that $- \beta(l)^2 \epsilon^2 / 2 \geqslant - \theta(l) \epsilon$. This leads to
\begin{equation*}
g_{(N)}( \beta(l) \epsilon ) \geqslant \beta(l) \epsilon \mu_N - \theta(l) \epsilon = \big( \beta(l) \mu_N - \theta(l) \big) \epsilon.
\end{equation*}
Thus there exists an integer $N_{l+1} = N_{l+1}(\xi)$ so that for all $N \geqslant \max_{1 \leqslant \ell^{\prime} \leqslant l+1} \{N_{\ell^{\prime}}\}$,
\begin{align*}
\beta(l) \mu_N - \theta(l) &\geqslant \mu_{\infty} ^{l+1} - \mu_{\infty} ^{l+1} / 4 - \theta(l) - \theta(l) \\
&= \mu_{\infty} ^{l+1} - \mu_{\infty} ^{l+1} / 4  - 0.5 \mu_{\infty} ^{l+1} / 4 - 0.5 \mu_{\infty} ^{l+1} / 4 \\
&= \mu_{\infty} ^{l+1} - \mu_{\infty} ^{l+1} / 2,
\end{align*}
so that
\begin{align}
\label{eqn:univ+b+l+1}
\bb_N(l+1) &\leftarrowtriangle  \big( \mu_{\infty} ^{l+1}  - \mu_{\infty} ^{l+1} / 2 \big) \epsilon \cc^{\varoast (l+1)} + \big( 1 - (\mu_{\infty} ^{l+1}  - \mu_{\infty} ^{l+1} / 2) \epsilon \big) \Delta_{\infty} \\
&\eqqcolon \tilde{\bb}_N(l+1). \nonumber
\end{align}
Now recall that $\error(\cc^{\varoast \ell}) \mu_{\infty} ^{\ell} > 1$ for $\ell \in \{l, l+1\}$. As a consequence, (\ref{eqn:univ+b+l}) and (\ref{eqn:univ+b+l+1}), together with the fact that the order of error probability functional is preserved under degradation, gives
\begin{align}
\label{eqn:error+l}
\error \big( \bb_N (l) \big) &\geqslant \error\big( \tilde{\bb}_N (l) \big) > \frac{3 \epsilon}{4} > \frac{\epsilon}{2} = \error \big( \bb(0) \big); \\
\label{eqn:error+l+1}
\error\big( \bb_N (l+1) \big) &\geqslant \error\big( \tilde{\bb}_N (l+1) \big) > \frac{\epsilon}{2} = \error\big( \bb(0) \big).
\end{align}
Take $N_{\infty} = \max_{0 \leqslant \ell \leqslant l+1} \{N_{\ell}\}$. By the erasure decomposition lemma \cite{RU08} and the fact that the BEC family is naturally ordered by degradation with respect to the erasure probability, (\ref{eqn:error+l}) and (\ref{eqn:error+l+1}) imply that both $\bb_N (l)$ and $\bb_N (l+1)$ are degraded with respect to $\bb_N(0)$, for all $N \geqslant N_{\infty}$. Therefore it follows that the sequence $\{\bb_N (\ell)\}$ converges to a fixed point $\bb_N (\infty) = \bb_N ^{\epsilon} (\infty)$ and satisfies $\error(\bb_N ^{\epsilon} (\infty)) > \epsilon / 2$, for all $N \geqslant N_{\infty}$. In fact we also have 
\begin{equation}
\label{eqn:surprise+error}
\error\big( \bb_N ^{\epsilon} (\infty) \big) > \frac{\xi}{2}, 
\end{equation}
for all $N \geqslant N_{\infty}$. To verify that (\ref{eqn:surprise+error}) holds, let us assume to the contrary that $\epsilon / 2 < \error(\bb_N ^{\epsilon} (\infty)) \leqslant \xi / 2$. Let $\error(\bb_N ^{\epsilon} (\infty)) = \varepsilon / 2$. We know that $\bb_N ^{\epsilon}(\infty) \leftarrowtriangle \bb^{\varepsilon}(0) \coloneqq \varepsilon \Delta_0 + (1 - \varepsilon) \Delta_{\infty}$ by the erasure decomposition lemma. Because $0 < \epsilon < \varepsilon \leqslant \xi$, by applying the same reasoning as above, we also have $\error(\bb_N ^{\varepsilon} (l)) > \varepsilon / 2$. Since $\bb_N ^{\epsilon} (\infty)$ is a fixed point, by induction it holds $\bb_N ^{\epsilon} (\infty) \leftarrowtriangle \bb_N ^{\varepsilon} (l)$ so that $\varepsilon / 2 = \error(\bb_N ^{\epsilon} (\infty)) \geqslant \error(\bb_N ^{\varepsilon} (l)) >\varepsilon / 2$, \ie, $\varepsilon > \varepsilon$, leading to a contradiction and thus verifying (\ref{eqn:surprise+error}). Consequently $\bb_N ^{\epsilon} (\infty)$ must be degraded with respect to $\bb^{\xi} (0) \coloneqq \xi \Delta_0 + (1 - \xi) \Delta_{\infty}$ for all $N \geqslant N_{\infty}$. 

So $\bb_N ^{\epsilon} (\infty)$ is degraded with respect to $\bb_N ^{\xi} (\ell)$ for all $\ell$ and thus to $\bb_N ^{\xi} (\infty)$, which follows by the monotonicity property of the density evolution operator and the fact that $\bb_N ^{\epsilon} (\infty)$ is a fixed point. Thus
\begin{equation}
\label{eqn:left+relation}
\bb_N ^{\epsilon} (\infty) \leftarrowtriangle \bb_N ^{\xi} (\infty).
\end{equation}
On the other hand, since $\bb^{\epsilon}(0)$ is upgraded with respect to $\bb^{\xi}(0)$, it again follows by the monotonicity property of the density evolution operator that
\begin{equation}
\label{eqn:right+relation}
\bb_N ^{\xi}(\infty) \leftarrowtriangle \bb_N ^{\epsilon} (\infty).
\end{equation}
Therefore (\ref{eqn:left+relation}) and (\ref{eqn:right+relation}) show that $\bb_N ^{\epsilon} (\infty)$ and $\bb_N ^{\xi} (\infty)$ are equal for all $\epsilon \in {(0,\xi]}$ and all $N \geqslant N_{\infty}$.

Now let $\xx_N(0)$ be an arbitrary $L$-density satisfying $\error(\xx_N(0)) = \epsilon \in (0,\xi]$. By the erasure decomposition lemma, $\xx_N(0)$ is degraded with respect to $\bb^{\epsilon}(0)$. It then follows by induction that $\xx_N(\ell)$ must be degraded with respect to $\bb_N ^{\epsilon}(\ell)$ for all $N \geqslant N_{\infty}$, and therefore
\begin{equation}
\liminf_{\ell \to \infty} \error\big(\xx_N (\ell)\big) \geqslant \error\big(\bb_N ^{\epsilon} (\infty)\big) = \error\big(\bb_N ^{\xi} (\infty)\big) > \xi.
\end{equation}
On the other hand, if $\error(\xx_N(0)) > \xi$, then from the erasure decomposition lemma we know that $\xx_N(0)$ is degraded with respect to $\bb^{\xi}(0)$, and it follows again by induction that the same conclusion holds, completing the proof of Theorem~\ref{thm:univ+generic}.
\end{proof}


\subsection{Appendix to Section \ref{subsec:poi+univ}}
\label{appdix:poi+univ}

\begin{proof}[Proof of Lemma~\ref{lem:poi+f2}]
By definition of $f_{(N)}$ in (\ref{eqn:poi+f+def}) we can derive that
\begin{equation}
\label{eqn:poi+f2+simplified}
f_{(N)} ^{\prime \prime} (x) = - \alpha \ee^{-\alpha x} (N-1) ( 1 - \ee^{-\alpha x} )^{N-2}.
\end{equation}
Since $0 \leqslant \ee^{-\alpha x} \leqslant 1$ for any $\alpha \geqslant 0$ and $x \geqslant 0$, $|f_{(N)} ^{\prime \prime} (x)|$ at each $x$ can be upper bounded as
\begin{equation}
\label{eqn:poi+f2+basic}
\big| f_{(N)} ^{\prime \prime} (x)  \big| \leqslant \alpha N ( 1 - \ee^{-\alpha x} )^{N-2}.
\end{equation}
Take any $\kappa \in {[0,\epsilon)}$ and consider $x \in [0,\kappa]$. Proceeding similarly as the proof of Lemma~\ref{lem:poi+f1+limit} in Appendix~\ref{appdix:inter+ca}, we can bound $(1 - \ee^{-\alpha x})^{N-2}$ from above by $\exp(2 - \ee^{-(\gamma+1/2)\kappa/\epsilon} N^{1 - \kappa/\epsilon})$. Therefore from (\ref{eqn:poi+f2+basic}) we can obtain 
\begin{equation}
\label{eqn:poi+f2+further}
\big| f_{(N)} ^{\prime \prime} (x)  \big| \leqslant \ee^2 \alpha N \exp\big(- \ee^{- \kappa (\gamma + 1/2) / \epsilon} N^{1 - \kappa/\epsilon} \big).
\end{equation}
Finally, since $N \geqslant 3$ and $H_{N-1} < \ln(N)+ \gamma + 1/2$ by (\ref{eqn:harmonic+bound}), we get $\ln(N) > \gamma + 1/2$ so that the inequality $\alpha \leqslant (\ln(N)+ \gamma + 1/2)/\epsilon \leqslant 2 \ln(N)/\epsilon$ holds. Combining this with (\ref{eqn:poi+f2+further}) gives (\ref{eqn:poi+f2+lemma}). The right-hand side of (\ref{eqn:poi+f2+lemma}) is essentially of order $N \ln(N) \exp(-N^{1 - \kappa/\epsilon})$, while $0 < \kappa/\epsilon < 1$. It then follows that this term converges to $0$ as $N$ tends to $\infty$ and this holds uniformly for all $x \in [0, \kappa]$. 

Finally let us verify that the sequence of second-order derivative at $x = \epsilon$ diverges to $-\infty$. We use the bounds $\ln(N-1) + 1 \leqslant H_{N-1} \leqslant \ln(N-1) + 2$ and (\ref{eqn:poi+f2+simplified}) to get
\begin{align*}
- \epsilon f_{(N)} ^{\prime\prime} (\epsilon) &= H_{N-1} \ee^{- H_{N-1}} (N-1) \big( 1 - \ee^{- H_{N-1}} \big)^{N-2} \\
&\geqslant H_{N-1} \ee^{-\ln(N-1)-2} (N-1) \big( 1 - \ee^{-1}/(N-1) \big)^{N-1} \\
&\geqslant \ee^{-2} (1 - \ee^{-1}) H_{N-1}.
\end{align*}
Thus $f_{(N)} ^{\prime\prime} (\epsilon)$ can be upper bounded by $-\epsilon^{-1} \ee^{-2} (1 - \ee^{-1}) ( \ln(N-1) + 1)$, which diverges to $-\infty$.
\end{proof}


\subsection{Appendix to Section \ref{subsec:rr+univ}}
\label{appdix:rr+univ}

\begin{proof}[Proof of Lemma \ref{lem:rr+f2+formula}]
For simplicity we suppress the subscript ``$_{(N)}$'' or superscript ``$^{(N)}$'' and write
\begin{equation*}
f(x) = \epsilon \lambda_2 \big( 1 - (1-x)^\ttr \big) + \epsilon \sum_{i=2}^{N-1} \lambda_{i+1} \big( 1 - (1-x)^\ttr \big)^i.
\end{equation*} 
Then the first-order derivative of $f$ can be computed as
\begin{equation*}
f^{\prime}(x) = \epsilon \lambda_2 \ttr (1-x)^{\ttr-1} + \epsilon \ttr \sum_{i=2}^{N-1} i \lambda_{i+1} \big( 1 - (1-x)^\ttr \big)^{i-1} (1-x)^{\ttr-1}.
\end{equation*} 
Therefore we get $( \lambda_2 \ttr (1-x)^{\ttr-1} )^{\prime} = - \lambda_2 \ttr (\ttr-1) (1-x)^{\ttr-2}$ and this gives rise to (\ref{eqn:rr+F0+def}). For $i \geqslant 2$, we have
\begin{equation}
\label{eqn:rr+f2+lambda3}
\Big( \big( 1 - (1-x)^\ttr \big)^{i-1} (1-x)^{\ttr-1} \Big)^{\prime} = \big( 1 - (1-x)^\ttr \big)^{i-2} (1-x)^{\ttr-2} \big( (i \ttr - 1)(1 - x)^\ttr - \ttr + 1 \big). 
\end{equation}
This gives rise to (\ref{eqn:rr+F1+def}). Combining these gives the identity in (\ref{eqn:rr+F012}).
\end{proof}

\section{Appendix to Section \ref{sec:stab+cond+map}}
\label{appdix:map}
In this appendix we give the proofs of some lemmas and theorems in Section~\ref{sec:stab+cond+map} on the stability under blockwise or bitwise MAP decoding.


\subsection{Appendix to Section \ref{subsec:cycle+deg2}}
\label{appdix:stab+map+cycle}

In this appendix we provide the proof of Theorem~\ref{thm:subcritical}, Lemma~\ref{lem:bhatta+error} and Theorem~\ref{thm:supercritical} in Section~\ref{subsec:cycle+deg2}. These theorems describes how the exploration process behaves, depending on whether the channel entropy is above or below the stability threshold. We start by proving Theorem~\ref{thm:subcritical}, which characterizes the exploration process described in Section~\ref{subsec:exploration} when the condition $\bhatta(\cc) \lambda^{\prime}(0) \rho^{\prime}(1) < 1$ holds.

\begin{proof}[Proof of Theorem~\ref{thm:subcritical}]
Consider the exploration process where each stage consists of $l \geqslant 1$ steps. Here $l$ is an arbitrary but fixed integer.  For any fixed $a \in (0,1/2)$, we consider the process $(A_k)_{k \geqslant 0}$ and the associated random time $K = \inf\{k \, | \, A_k = 0\}$. We want to bound the probability $\PP(K > n^a)$. Notice first that for any $k \geqslant 1$, we have
\begin{equation*}
A_k = A_{k-1} + Z_k - 1 = \sum_{i=1}^k Z_i - (k-1).
\end{equation*}
Therefore, we can bound $\PP(K > k)$ as
\begin{align*}
\PP(K > k) &= \PP \bigg( \bigcap_{i=1}^k \{A_i \geqslant 1\} \bigg) \leqslant \PP(A_k \geqslant 1) \\
&= \PP \bigg( \sum_{i=1}^k Z_i - (k-1) \geqslant 1 \bigg) = \PP \bigg( \sum_{i=1}^k Z_i \geqslant k \bigg).
\end{align*}
We claim that the process $(A_k)_{0 \leqslant k \leqslant n^a}$ behaves essentially like a branching process $(\tilde{A}_k)_{k \geqslant 0}$, when $0 \leqslant k \leqslant n^a$ with $a \in (0,1/2)$ fixed. More precisely, we let
\begin{equation*}
\tilde{A}_k = \tilde{A}_{k-1} + \tilde{Z}_k - 1, 
\end{equation*}
where $\tilde{A}_{-1} = 0$ while the random variables $\tilde{Z}_k$s are independent and identically distributed. To see this, let us consider stage $1$ of the exploration process that consists of $l$ steps, where we pick active variable or check half-edges and uniformly connect it to unexplored check or variable half-edges. It is possible that the connected ones are not neutral but it is clear that the expected number of such occurrences is of order $\cO(1/n)$, where the implicit constant depends on $l$ as well as the maximum degrees $\ttl$ and $\ttr$. In general when we perform stage $i \leqslant k$, the expected number of such occurrences is of order $\cO(i/n)$, since we need to take into account the active half-edges in the previous $i-1$ stages. Therefore, if we perform in total $k$ stages, the expected number of such occurrences in at most of order $\cO(k^2/n)$ which is $o(1)$ as long as $k = o(\sqrt{n})$.

Therefore, we can couple the processes $(A_k)_{0 \leqslant k \leqslant n^a}$ and $(\tilde{A}_k)_{0 \leqslant k \leqslant n^a}$ such that
\begin{equation*}
\PP \bigg( \sum_{i=1}^k Z_i \geqslant k \bigg) = \PP \bigg( \sum_{i=1}^k \tilde{Z}_i \geqslant k \bigg) + o(1).
\end{equation*}
Now if we let $\gamma = \gamma(\cc,\lambda,\rho)$ be such that $\gamma^l = \error(\cc^{\varoast l}) (\lambda^{\prime}(0) \rho^{\prime} (1))^l < 1$, then there exists a $\delta = \delta(\cc,\lambda,\rho) > 0$ such that $(1+\delta)\gamma^l < 1$ and therefore we can get
\begin{equation*}
\PP \bigg( \sum_{i=1}^k \tilde{Z}_i \geqslant k \bigg) \leqslant \PP \bigg( \sum_{i=1}^k \tilde{Z}_i \geqslant (1+\delta) \gamma^l k \bigg).
\end{equation*}
By Hoeffding's inequality we conclude that
\begin{equation*}
\PP \bigg( \sum_{i=1}^k \tilde{Z}_i \geqslant (1+\delta) \gamma^l k \bigg) \leqslant \exp(-2 \delta^2 k^2 \gamma^{2l} / (k \ttd^{2l})) = \ee^{- 2 k \delta^2 \gamma^{2l}/\ttd^{2l}},
\end{equation*}
completing the proof of Theorem~\ref{thm:subcritical}.
\end{proof}

\begin{proof}[Proof of Lemma~\ref{lem:bhatta+error}]
By assumption, there exists a $\gamma = \gamma(\cc, \lambda, \rho) \in {(1, \ttd)}$ such that $\bhatta(\cc) \lambda^{\prime}(0) \rho^{\prime} (1) > \gamma$. Thus there exists a $\bhatta = \bhatta(\mathsf{c}, \lambda, \rho)$ such that $0 < \bhatta < \bhatta(\cc)$ satisfying $\bhatta \lambda^{\prime}(0) \rho^{\prime} (1) > \gamma$. Definition \ref{defn:residual} gives
\begin{align*}
\hat{\lambda}^{\prime}(0) \hat{\rho}^{\prime}(1) &= \frac{2 L_2 - \phi_2}{L^{\prime} (1)} \sum_{i = 1}^{\ttr-1} \frac{(i+1) R_{i+1} - \psi_{i+1}/(1 - r)}{R^{\prime}(1)} i \\
&\geqslant \frac{2 L_2 - \varepsilon}{L^{\prime}(1)} \sum_{i = 1}^{\ttr-1} \frac{(i+1) R_{i+1} - \psi_{i+1}/(1 - r)}{R^{\prime}(1)} i \\
&\geqslant \lambda^{\prime}(0) \rho^{\prime}(1) - \bigg(\frac{\varepsilon}{L^{\prime}(1)} \sum_{i=1}^{\ttr-1} \frac{(i+1) R_{i+1}}{R^{\prime}(1)} i + \frac{2 L_2}{L^{\prime}(1)} \sum_{i = 1}^{\ttr-1} \frac{\psi_{i+1}/(1 - r)}{R^{\prime}(1)} i\bigg) \\
&\geqslant \lambda^{\prime}(0) \rho^{\prime}(1) - \bigg(\frac{\rho^{\prime}(1)}{L^{\prime}(1)} + \frac{2 L_2 \ttr^2}{L^{\prime}(1) ^2}\bigg) \varepsilon.
\end{align*}
Since $\lambda^{\prime}(0) \rho^{\prime}(1) > \gamma / \bhatta$, there exists a $\xi = \xi(\cc, \lambda, \rho) > 0$ so that we have $\lambda^{\prime}(0) \rho^{\prime}(1) > \gamma / \bhatta + \xi$. So
\begin{equation}
\label{eqn:lambda+rho+bound}
\hat{\lambda}^{\prime}(0) \hat{\rho}^{\prime}(1) > \frac{\gamma}{\mathfrak{B}} + \xi - \bigg(\frac{\rho^{\prime}(1)}{L^{\prime}(1)} + \frac{2 L_2 \ttr^2}{L^{\prime}(1) ^2}\bigg) \varepsilon.
\end{equation}
Now we pick 
\begin{equation}
\label{eqn:def+varepsilon}
\varepsilon = \min \Big\{(\xi/2) \big/ (\rho^{\prime}(1) / L^{\prime}(1) + 2 L_2 \ttr^2 / L^{\prime}(1) ^2), L_2 \Big\}. 
\end{equation}
Then (\ref{eqn:lambda+rho+bound}) gives the inequality $\hat{\lambda}^{\prime}(0) \hat{\rho}^{\prime}(1) > \gamma/\bhatta + \xi/2$. This implies that there exists a $\delta = \delta(\mathsf{c}, \lambda, \rho) > \xi / 2$ such that $\hat{\lambda}^{\prime}(0) \hat{\rho}^{\prime}(1) > \gamma/\mathfrak{B} + \delta$. Similarly, we have
\begin{align}
\hat{\lambda}^{\prime} _n (0) \hat{\rho}^{\prime} _n (1) &\geqslant \lambda^{\prime} _n (0) \rho^{\prime} _n (1) - \bigg(\frac{\rho^{\prime} _n (1)}{L^{\prime} _n (1)} + \frac{2 L_{2,n} \ttr^2}{L^{\prime} _n (1) ^2}\bigg) \varepsilon \nonumber \\
&\geqslant \hat{\lambda}^{\prime}(0) \hat{\rho}^{\prime}(1) + \lambda^{\prime} _n (0) \rho^{\prime} _n (1) - \lambda^{\prime}(0) \rho^{\prime}(1) - \bigg(\frac{\rho^{\prime} _n (1)}{L^{\prime} _n (1)} + \frac{2 L_{2,n} \ttr^2}{L^{\prime} _n (1) ^2}\bigg) \varepsilon \nonumber \\
\label{eqn:conseq+def+varepsilon}
&\geqslant \hat{\lambda}^{\prime}(0) \hat{\rho}^{\prime}(1) + \lambda^{\prime} _n (0) \rho^{\prime} _n (1) - \lambda^{\prime}(0) \rho^{\prime}(1) \nonumber \\
&\quad\quad\quad\quad\quad\,\,\, - \frac{\rho^{\prime} _n (1) / L^{\prime} _n (1) + 2 L_{2,n} \ttr^2 / L^{\prime} _n (1) ^2}{\rho^{\prime}(1) / L^{\prime}(1) + 2 L_2 \ttr^2 / L^{\prime}(1) ^2} \frac{\xi}{2} \\
\label{eqn:asymptotic+lower+bound+xi}
&\geqslant \hat{\lambda}^{\prime}(0) \hat{\rho}^{\prime}(1) - \delta > \gamma / \bhatta,
\end{align}
where (\ref{eqn:conseq+def+varepsilon}) follows by definition of $\varepsilon$ in (\ref{eqn:def+varepsilon}), and the second to last inequality in (\ref{eqn:asymptotic+lower+bound+xi}) follows by the convergence of $(\lambda_n,\rho_n)$ in Section~\ref{subsec:prelim}. Therefore for the given $\beta = \beta(\cc, \lambda, \rho) = (\bhatta(\cc)\ee)^{3/2} \sqrt{2 \ln(\bhatta(\cc) / \bhatta)} / (9 \pi)$, we can select an $l = l(\cc, \lambda, \rho) \in \NN$ such that 
\begin{equation*}
\beta \big(\bhatta \hat{\lambda}^{\prime} _n (0) \hat{\rho}^{\prime} _n (1)\big)^l \geqslant \beta \big(\bhatta \hat{\lambda}^{\prime}(0) \hat{\rho}^{\prime}(1) - \delta \big)^l > \gamma^l. 
\end{equation*}
By \cite{RU08} we know that $\error(\cc^{\varoast l}) \geqslant \beta \bhatta^l$, which completes the proof.
\end{proof}

\begin{proof}[Proof of Theorem~\ref{thm:supercritical}]
Let $\varepsilon = \varepsilon(\cc,\lambda,\rho) > 0$, $l = l(\cc,\lambda,\rho) \in \NN$, $\gamma = \gamma(\cc,\lambda,\rho) > 1$ be determined by Lemma~\ref{lem:bhatta+error} and fixed. First notice that every time we perform $l$ steps in a stage in the exploration process, the maximum number of check half-edges that we have explored is upper bounded by $\ttd^l$, where we recall that $\ttd$ is the maximum check-node degree $\ttr$ minus $1$. Consequently, if we only consider the process $(A_k)_{k \geqslant 0}$ up to stage $\varepsilon n / \ttd^l$, then the total number of check half-edges that we have explored does not exceed $\varepsilon n$ and the same goes for the total number of variable half-edges. For notational simplicity, we define 
\begin{equation*}
\epsilon = \varepsilon / \ttd^{l}.
\end{equation*}
Let $\rrF_k$ be the $\sigma$-algebra generated by the status of variable and check half-edges up to and including stage $k$, and in particular $\rrF_k$ contains the information about the number of active, open, neutral and explored variable and check half-edges. We claim that $\EE[Z_k | \rrF_{k-1}] \geqslant \gamma^l$ as long as $k \leqslant \epsilon n$. To see this, we recall the modification of status of active check half-edges in the additional step after step $l$ in stage $k$. We can write
\begin{equation*}
Z_k = \sum_{j = 1}^{N_k(l)}\Big( \II_{\{\sum_{u \in \mathcal{P}_j} L_u < 0\}} + \II_{\{B_{\mathcal{P}_j} = 0\}} \II_{\{\sum_{u \in \mathcal{P}_j} L_u = 0\}} \Big),
\end{equation*}
where $N_k(l)$ is the number of active check half-edges at stage $k$ before the additional step and after step $l$, and $\mathcal{P}_j$ is the set of variable nodes of degree two which are contained in the path of length $2l$ that contains active check nodes only and also contains the $j$-th active check half-edge. By the tower property of expectation, we have
\begin{align*}
\EE[Z_k | \rrF_{k-1}] &= \EE \big[ \EE[ Z_k | \rrF_k ] \big| \rrF_{k-1} \big] \\
&= \error(\cc^{\varoast l}) \EE[N_k(l) | \rrF_{k-1}].
\end{align*}
Now write 
\begin{equation*}
N_k(l) = \sum_{j = 1}^{N_k(l-1)} X_j(l), 
\end{equation*}
where $N_k(l-1)$ is the number of active check half-edges from which we directly grew out all $N_k(l)$ active check half-edges at step $l$ in stage $k$, \ie, these half-edges are at graph distance $2$ from the $N_k(l)$ half-edges at step $l-1$ in stage $k$; and $X_j(l)$ is the number of active check half-edges at stage $k$ when performing step $l$ and which are originated from the $j$-th active check half-edge after step $l-1$ but before step $l$. If we denote by $\rrF_k(l-1)$ the $\sigma$-algebra generated by $\rrF_{k-1}$ and the status of variable and check half-edges up to and including step $l-1$ in stage $k$, then we can continue writing
\begin{align*}
\EE[Z_k | \rrF_{k-1}] &= \error(\cc^{\varoast l}) \EE\big[ \EE[N_k(l) | \rrF_{k} (l-1)] \big| \rrF_{k-1} \big] \\
&\geqslant \error(\cc^{\varoast l}) \hat{\lambda}^{\prime} _n (0) \hat{\rho}^{\prime} _n (1) \EE[N_k(l-1) | \rrF_{k-1}],
\end{align*}
where the last inequality holds since $\EE[X_j(l) | \rrF_{k}(l-1)]$ is lower bounded by $\hat{\lambda}^{\prime} _n (0) \hat{\rho}^{\prime} _n (1)$. Indeed, because $k \leqslant \varepsilon n / \ttd^l$, the number of explored check half-edges has not exceeded $\varepsilon n$. The claim then follows by recalling Definition \ref{defn:residual}.

In general and for any $t \in \{0,1,2,\dots,l-1\}$, we denote by $N_k(l-t)$ the number of active check half-edges after performing step $l-t$ and denote by $\rrF_k(l-t)$ the $\sigma$-algebra generated by $\rrF_{k-1}$ and the status of variable and check half-edges up to and including step $l-t$ in stage $k$. Notice that $N_k(0) = 1$ and $\rrF_k(0) = \rrF_{k-1}$ since in step $0$ we pick the active check half-edge in the breadth-first order which is known based on $\rrF_{k-1}$. Continuing the computation in a similar fashion, we obtain that for any $t \in \{0,1,2,\dots,l-1\}$,
\begin{equation}
\label{eqn:Zk+t}
\EE[Z_k | \rrF_{k-1}] \geqslant \error(\cc^{\varoast l}) \big( \hat{\lambda}^{\prime} _n (0) \hat{\rho}^{\prime} _n (1) \big)^{t} \EE[N_k(l-t) | \rrF_{k-1}].
\end{equation}
Since $\EE[N_k(1) | \rrF_{k-1}] = \EE[N_k(1) | \rrF_k(0)] \geqslant \hat{\lambda}^{\prime} _n (0) \hat{\rho}^{\prime} _n (1)$, the inequality in (\ref{eqn:Zk+t}) leads to
\begin{equation}
\label{eqn:cond+exp+lb}
\EE[Z_k | \rrF_{k-1}] \geqslant \error(\cc^{\varoast l}) \big( \hat{\lambda}^{\prime} _n (0) \hat{\rho}^{\prime} _n (1) \big)^l \geqslant \gamma^l,
\end{equation}
where the last inequality in (\ref{eqn:cond+exp+lb}) is due to Lemma \ref{lem:bhatta+error}. On the other hand, it is not difficult to see that $\EE[Z_k | \rrF_{k-1}] \leqslant \ttd^l$ holds.

Now recall that the process $(A_k)_{k \geqslant 0}$ satisfies the recursion $A_k = A_{k-1} + Z_k - 1$ so that we can further write $A_k = \sum_{i=1}^k Z_i - (k - 1)$. Therefore for any $k \geqslant 1$, we have that
\begin{align}
\PP \big( A_k \leqslant (\bar{\delta} \gamma^l - 1) k \big) &= \PP \bigg( \sum_{i=1}^k Z_i - (k-1) \leqslant (\bar{\delta} \gamma^l - 1) k \bigg) \nonumber \\
&= \PP \bigg( \sum_{i=1}^k Z_i \leqslant \bar{\delta} \gamma^l k - 1 \bigg) \nonumber \\
\label{eqn:Ak+bound}
&\leqslant \PP \bigg( \sum_{i=1}^k Z_i \leqslant \bar{\delta} \gamma^l k \bigg).
\end{align}
Thus for any $s > 0$, we can further bound (\ref{eqn:Ak+bound}) by 
\begin{align}
\PP \big( A_k \leqslant (\bar{\delta} \gamma^l - 1) k \big) &\leqslant \PP \bigg( -s \sum_{i=1}^k Z_i \geqslant -s \bar{\delta} \gamma^l k \bigg) \nonumber \\
&= \PP \Bigg( \exp \bigg( -s \sum_{i=1}^k Z_i \bigg) \geqslant \exp(-s \bar{\delta} \gamma^l k) \Bigg) \nonumber \\
\label{eqn:Ak+markov}
&\leqslant \ee^{s \bar{\delta} \gamma^l k} \EE \Bigg[ \exp \bigg( -s \sum_{i=1}^k Z_i \bigg) \Bigg],
\end{align}
where (\ref{eqn:Ak+markov}) is due to Markov's inequality. Notice that the expectation in the last inequality involves $(Z_i)_{1 \leqslant i \leqslant k}$ and to get a further upper bound, we write
\begin{equation}
\label{eqn:chernoff+Zk}
\EE \Bigg[ \exp \bigg( -s \sum_{i=1}^k Z_i \bigg) \Bigg] = \EE \Bigg[ \exp \bigg( -s \sum_{i=1}^{k-1} Z_i \bigg) \EE \big[ \ee^{-s Z_k} \big| \rrF_{k-1} \big] \Bigg].
\end{equation}
According to (\ref{eqn:cond+exp+lb}) we know that $\gamma^l \leqslant \EE[Z_k | \rrF_{k-1}] \leqslant \ttd^l$. Given these bounds on the conditional expectation of $Z_k$ given $\rrF_k$, we can characterize ``feasible'' conditional distributions of $Z_k$ given $\rrF_{k-1}$, which then enables us to derive an upper bound on $\EE[e^{-s Z_k} | \rrF_{k-1}]$. 

More concretely, $\EE[e^{-s Z_k} | \rrF_{k-1}]$ can be further upper bounded by the negative of the optimum of the following linear program $\mathsf{LP}^{\mathrm{prim}}$ in the standard primal form as
\begin{align}
\label{eqn:lp+prim}
\begin{split}
\min_{\{p_j\}} \quad & - \sum_{j=0}^{\ttd^l} p_j \ee^{-s j} \\
\mathrm{s.t.}  \quad & - \sum_{j=0}^{\ttd^l} p_j j + \gamma^l \leqslant 0; \\
                     & \sum_{j=0}^{\ttd^l} p_j - 1 \leqslant 0; \\
                     & - p_j \leqslant 0, \forall j.
\end{split}
\end{align}
Notice that in the above linear program we used the inequality constraint $\sum_{j=0}^{\ttd^l} p_j \leqslant 1$ to replace the equality constraint requiring $\{p_j\}$ a valid distribution, since in both cases the constraint sets are bounded and closed so that the optimal value is obtained on the boundary of the corresponding constraint set. Furthermore, the Lagrangian function $\cL = \cL(p_0,\dots,p_{\ttd^l},a,b,c_0,\dots,c_{\ttd^l})$ is defined as
\begin{equation*}
\cL = - \sum_{j=0}^{\ttd^l} p_j \ee^{-s j} - a \bigg( \sum_{j=0}^{\ttd^l} p_j j - \gamma^l \bigg) + b \bigg( \sum_{j=0}^{\ttd^l} p_j - 1 \bigg) - \sum_{j=0}^{\ttd^l} c_j p_j,
\end{equation*}
where the Lagrangian multipliers $a$, $b$ and $c_j$s are nonnegative. Therefore we can write the Lagrangian dual function $\cD = \cD(a,b,c_0,\dots,c_{\ttd^l})$ as
\begin{align}
\cD(a,b,c_0,\dots,c_{\ttd^l}) &= \inf_{\{p_j\}} \cL(p_0,\dots,p_{\ttd^l},a,b,c_0,\dots,c_{\ttd^l}) \nonumber \\
\label{eqn:dual}
&= \inf_{\{p_j\}} \sum_{j=0}^{\ttd^l} \big( - \ee^{-s j} - a j + b - c_j \big) p_j + a \gamma^l - b.
\end{align}
We notice that the infimum in (\ref{eqn:dual}) attains $-\infty$ unless $- \ee^{-s j} - a j + b - c_j = 0$ for every $j \in \{0,1,\dots,\ttd^l\}$. Since the $c_j$s are nonnegative, the above conditions imply that $b - a j \geqslant \ee^{-s j}$ for every $j$. Consequently we obtain the dual linear program $\mathsf{LP}^{\mathrm{dual}}$ associated to the linear program $\mathsf{LP}^{\mathrm{prim}}$ as
\begin{align}
\label{eqn:lp+dual}
\begin{split}
\max_{a,b}    \quad & a \gamma^l - b \\
\mathrm{s.t.} \quad & b - a j \geqslant \ee^{-s j}, \forall j; \\
			     & a \geqslant 0. 
\end{split}
\end{align}
Since the function $x \mapsto b - a x$ is linear and the function $x \mapsto \ee^{-s x}$ is convex-$\cup$, it follows by convexity that the conditions $b - a j \geqslant \ee^{-s j}$ are fulfilled for all $j$ if and only if the conditions are fulfilled for $j=0$ and $j=\ttd^l$, which are given by $b \geqslant 1$ and $b - a \ttd^l \geqslant \ee^{-s \ttd^l}$. Now for any fixed $b \geqslant 1$, the second condition gives $a \leqslant (b - \ee^{-s \ttd^l}) / \ttd^l$ and the maximum value of the linear program $\mathsf{LP}^{\mathrm{dual}}$ in (\ref{eqn:lp+dual}) in this case is 
\begin{equation*}
(b - \ee^{-s \ttd^l}) \gamma^l / \ttd^l - b = - b (1 - \gamma^l / \ttd^l) - \ee^{-s \ttd^l} \gamma^l / \ttd^l,
\end{equation*}
which is obtained by setting $a = (b - \ee^{-s \ttd^l}) / \ttd^l$. Notice that since $b \geqslant 1$, the condition $a \geqslant 0$ is satisfied. Consequently we get the optimal value of the linear program $\mathsf{LP}^{\mathrm{dual}}$ in (\ref{eqn:lp+dual}) by setting $b = 1$ and the optimum is equal to $-1 + \gamma^l / \ttd^l - \ee^{-s \ttd^l} \gamma^l / \ttd^l$. This in particular implies that
\begin{equation*}
\EE[e^{-s Z_k} | \rrF_{k-1}] \leqslant 1 - \gamma^l / \ttd^l + \ee^{-s \ttd^l} \gamma^l / \ttd^l,
\end{equation*}
and as a consequence, (\ref{eqn:chernoff+Zk}) becomes
\begin{equation*}
\EE \Bigg[ \exp \bigg( -s \sum_{i=1}^k Z_i \bigg) \Bigg] \leqslant \Bigg( 1 - \frac{\gamma^l}{\ttd^l} + \frac{\gamma^l}{\ttd^l} \ee^{-s \ttd^l} \Bigg) \EE \Bigg[ \exp \bigg( -s \sum_{i=1}^{k-1} Z_i \bigg) \Bigg].
\end{equation*}
Now we continue the computation in a similar approach by sequentially bounding the conditional expectations $\EE[\ee^{-s Z_{k-t}} | \rrF_{k-t-1}]$, where $t$ increases from $1$ to $k-1$. Since $k \leqslant \epsilon n$, we know that $\gamma^l$ is a valid lower bound on $\EE[\ee^{-s Z_{k-t}} | \rrF_{k-t-1}]$ so that the negative of the optimum of the previous linear program $\mathsf{LP}^{\mathrm{prim}}$ gives an upper bound of it as well. Therefore, we get
\begin{equation*}
\EE \Bigg[ \exp \bigg( -s \sum_{i=1}^k Z_i \bigg) \Bigg] \leqslant \Bigg( 1 - \frac{\gamma^l}{\ttd^l} + \frac{\gamma^l}{\ttd^l} \ee^{-s \ttd^l} \Bigg)^k,
\end{equation*}
so that (\ref{eqn:Ak+markov}) becomes
\begin{align}
\PP \big( A_k \leqslant (\bar{\delta} \gamma^l - 1) k \big) &\leqslant \ee^{s \bar{\delta} \gamma^l k} \big( 1 - \gamma^l / \ttd^l + \ee^{-s \ttd^l} \gamma^l / \ttd^l \big)^k \nonumber \\
\label{eqn:s+bound}
&= \Big( \ee^{s \bar{\delta} \gamma^l} \big( 1 - \gamma^l / \ttd^l + \ee^{-s \ttd^l} \gamma^l / \ttd^l \big) \Big)^k.
\end{align}
In Lemma~\ref{lem:gs} and Lemma~\ref{lem:exp+g+bound} at the end of the proof of this theorem we show that by deriving a uniform in $s$ bound, the left-hand side of (\ref{eqn:s+bound}) can be bounded by the exponentially decreasing function $\exp(-k \delta^2 \gamma^l/(2\ttd^l))$, and this gives the desired result (\ref{eqn:exp+bound+Ak}) in Theorem~\ref{thm:supercritical}. In the following we prove (\ref{eqn:exp+bound+K}) by a similar line of arguments as above.

Recall that by definition $K = \inf\{k \, | \, A_k = 0\}$ and for some $\kappa = \kappa(\cc,\lambda,\rho) \in \NN$ that will be specified later, we want to show that conditional on the event $\{K > \kappa\}$, the probability of the event $\{K \leqslant \epsilon n\}$ cannot tend to $1$ as the blocklength tends to $\infty$. We start by writing
\begin{align}
\PP( K \leqslant k \, | \, K > \kappa ) &\leqslant \PP \bigg( \bigcup_{j=\kappa+1}^k \{A_j = 0\} \, \bigg| \, \bigcap_{j=1}^{\kappa} \{A_j \geqslant 1\} \bigg) \nonumber \\
&\leqslant \sum_{j=\kappa+1}^k \PP \bigg( A_j = 0 \, \bigg| \, \bigcap_{j=1}^{\kappa} \{A_j \geqslant 1\} \bigg) \nonumber \\
&= \sum_{j=\kappa+1}^k \PP \bigg( \sum_{i=1}^j Z_i = j-1 \, \bigg| \, \bigcap_{j=1}^{\kappa} \{A_j \geqslant 1\} \bigg) \nonumber \\
\label{eqn:K+Zi}
&\leqslant \sum_{j=\kappa+1}^k \PP \bigg( \sum_{i=1}^j Z_i \leqslant j \, \bigg| \, \bigcap_{j=1}^{\kappa} \{A_j \geqslant 1\}\bigg).
\end{align}
Since the parameters $\gamma$ and $l$ satisfy $\gamma^l > 1$, for any $\delta = \delta(\cc,\lambda,\rho) \in (0,1)$ such that $\bar{\delta} \gamma^l > 1$, by (\ref{eqn:K+Zi}) we get
\begin{equation*}
\PP( K \leqslant k \, | \, K > \kappa) \leqslant \sum_{j=\kappa+1}^k \PP \bigg( \sum_{i=1}^j Z_i \leqslant \bar{\delta} \gamma^l j \, \bigg| \, \bigcap_{j=1}^{\kappa} \{A_j \geqslant 1\} \bigg).
\end{equation*}
Therefore for any $s > 0$, we get 
\begin{align}
\PP( K \leqslant k \, | \, K > \kappa ) &\leqslant \sum_{j=\kappa+1}^k \PP \bigg( -s \sum_{i=1}^j Z_i \geqslant -s \bar{\delta} \gamma^l j \, \bigg| \, \bigcap_{j=1}^{\kappa} \{A_j \geqslant 1\} \bigg) \nonumber \\
&= \sum_{j=\kappa+1}^k \PP \Bigg( \exp\bigg( -s \sum_{i=1}^j Z_i  \bigg) \geqslant \ee^{-s \bar{\delta} \gamma^l j} \, \bigg| \, \bigcap_{j=1}^{\kappa} \{A_j \geqslant 1\} \Bigg) \nonumber \\
\label{eqn:cond+markov}
&\leqslant \sum_{j=\kappa+1}^k \ee^{s \bar{\delta} \gamma^l j} \EE \Bigg[ \exp\bigg( -s \sum_{i=1}^j Z_i \bigg) \, \bigg| \, \bigcap_{j=1}^{\kappa} \{A_j \geqslant 1\}\Bigg].
\end{align}
Consider first the case $j = \kappa+1$ in (\ref{eqn:cond+markov}). We can bound the conditional expectation of $\ee^{-s Z_{\kappa+1}}$ in (\ref{eqn:cond+markov}) by using the linear program $\mathsf{LP}^{\mathrm{prim}}$ in (\ref{eqn:lp+prim}). This corresponds to the case where $i = j = \kappa+1$. In general for $i$ ranging from $\kappa$ to $1$, the conditioning events $\{A_{j^{\prime}} \geqslant 1\}$ for $j^{\prime} \leqslant \kappa$ impose further conditions on the $Z_i$s; \eg, $\{A_1 \geqslant 1\}$ is equivalent to $\{Z_1 \geqslant 1\}$. The same reasoning goes for a generic $j > \kappa+1$. Thus instead of using the objective function in (\ref{eqn:lp+prim}), we need to consider 
\begin{equation}
\label{eqn:cond+opt+Z}
\min_{\jmath, \{p_j\}} - \sum_{j=\jmath}^{\ttd^l} \ee^{-s j} p_j, 
\end{equation}
where $\jmath$ is the additional parameter taking values in $\{0,1,\dots,\ttd^l\}$. Furthermore, the constraints imposed on  (\ref{eqn:cond+opt+Z}) are $- \sum_{j \geqslant \jmath} p_j j + \gamma^l \leqslant 0$, $\sum_{j \geqslant \jmath} p_j -1 \leqslant 0$ and $- p_j \leqslant 0$ for all $j$. 

We claim that the optimal value of (\ref{eqn:cond+opt+Z}) is lower bounded by the optimal value of the linear program $\mathsf{LP}^{\mathrm{prim}}$. We show this by showing that any optimal solution to (\ref{eqn:cond+opt+Z}) is also a feasible solution to the linear program $\mathsf{LP}^{\mathrm{prim}}$. More precisely, let $\hat{\jmath}$ and $\{\hat{p}_j\}$ be an optimal solution to (\ref{eqn:cond+opt+Z}) and let $\{q_j\}$ be such that for all $j \leqslant \hat{\jmath}-1$, $q_j = 0$ and for $j \geqslant \hat{\jmath}$, $q_j = \hat{p}_j$. Then $\{q_j\}$ attains the optimum for (\ref{eqn:cond+opt+Z}) while it satisfies all constraints in (\ref{eqn:lp+prim}). Therefore we can use $\mathsf{LP}^{\mathrm{prim}}$ to bound $\PP( K \leqslant k \, | \, K > \kappa )$ as 
\begin{equation}
\label{eqn:K+bound+in+s}
\PP( K \leqslant k \, | \, K > \kappa ) \leqslant \sum_{j=\kappa+1}^k \Big( \ee^{s \bar{\delta} \gamma^l} \big( 1 - \gamma^l / \ttd^l + \ee^{-s \ttd^l} \gamma^l / \ttd^l \big) \Big)^j.
\end{equation}
Lemma~\ref{lem:gs} and Lemma~\ref{lem:exp+g+bound} at the end of the proof of this theorem allow us to write
\begin{align}
\PP( K \leqslant k \, | \, K > \kappa ) &\leqslant \sum_{j=\kappa+1}^k \ee^{-j \delta^2 \gamma^l / (2 \ttd^l)} \nonumber \\
&= c_{\kappa} \big( 1 - \ee^{- (k - \kappa) \delta^2 \gamma^l / (2 \ttd^l)} \big),
\end{align}
where
\begin{equation}
\label{eqn:ck+K}
c_{\kappa} \coloneqq \frac{\exp(- (\kappa+1) \delta^2 \gamma^l / (2 \ttd^l))}{1 - \exp(- \delta^2 \gamma^l / (2 \ttd^l))}.
\end{equation} 
Notice now that the numerator on the right-hand side of (\ref{eqn:ck+K}) decreases exponentially in $\kappa$ so that for the given parameters $\gamma$, $l$, $\ttd$ and any fixed $\delta \in (0,1)$ such that $\bar{\delta} \gamma^l > 1$, we can select a $\kappa = \kappa(\cc,\lambda,\rho,\delta) \in \NN$ such that $c_{\kappa} \in (0,1)$. This completes the proof of Theorem~\ref{thm:supercritical}.
\end{proof}

\begin{lemma}
\label{lem:gs}
Let $\gamma = \gamma(\cc,\lambda,\rho)$ and $l = l(\cc,\lambda,\rho)$ be specified by Lemma~\ref{lem:bhatta+error} and let $\delta \in (0,1)$ be arbitrary and fixed. Define the function $g: (0,\infty) \to (0,\infty)$ by
\begin{equation*}
g(s) = \ee^{s \bar{\delta} \gamma^l} \big( 1 - \gamma^l/\ttd^l + \ee^{-s \ttd^l} \gamma^l / \ttd^l \big).
\end{equation*}
Then for all $s >0$, it holds that
\begin{equation*}
g(s) \geqslant \bar{\delta}^{-\bar{\delta} \gamma^l / \ttd^l} \bigg( \frac{\ttd^l - \gamma^l}{\ttd^l - \bar{\delta} \gamma^l} \bigg)^{1 - \bar{\delta} \gamma^l / \ttd^l}.
\end{equation*}
\end{lemma}
\begin{proof}
By definition of $g$ we know that $g(0) = 1$ and $\lim_{s \to \infty} g(s) = \infty$. Besides, the first-order derivative of $g$ can be computed as
\begin{equation*}
g^{\prime}(s) = \ee^{s \bar{\delta} \gamma^l} \Bigg( \bar{\delta} \gamma^l \bigg( 1 - \frac{\gamma^l}{\ttd^l} + \frac{\gamma^l}{\ttd^l} \ee^{-s \ttd^l} \bigg) - \gamma^l \ee^{-s \ttd^l} \Bigg),
\end{equation*}
so that $g^{\prime}(0) = \bar{\delta} \gamma^l (1 - \gamma^l / \ttd^l + \gamma^l / \ttd^l) - \gamma^l = \bar{\delta} \gamma^l - \gamma^l = - \delta \gamma^l < 0$. Therefore $g$ has at least one stationary point in $(0,\infty)$. Setting the derivative $g^{\prime}(s)$ to be $0$, we can obtain that
\begin{equation}
\label{eqn:g+prime}
\bigg( \frac{\bar{\delta} \gamma^{2l}}{\ttd^l} - \gamma^l \bigg) \ee^{-s \ttd^l} + \bar{\delta} \gamma^l \bigg( 1 - \frac{\gamma^l}{\ttd^l} \bigg) = 0,
\end{equation}
which has a unique solution in $(0,\infty)$. So $g$ has a unique minimum at $s^{\star} \in (0,\infty)$ and solving (\ref{eqn:g+prime}) for $s^{\star}$ we get
\begin{equation*}
\ee^{-s^{\star} \ttd^l} = \frac{\bar{\delta} (\ttd^l - \gamma^l)}{\ttd^l - \bar{\delta} \gamma^l}.
\end{equation*}
Consequently we can evaluate $g(s^{\star})$ by computing
\begin{equation*}
\ee^{s^{\star} \bar{\delta} \gamma^l} = \big( \ee^{s^{\star} \ttd^l} \big)^{\bar{\delta} \gamma^l / \ttd^l} = \bigg( \frac{\ttd^l - \bar{\delta} \gamma^l}{\bar{\delta} (\ttd^l - \gamma^l)} \bigg)^{\bar{\delta} \gamma^l / \ttd^l},
\end{equation*}
and
\begin{align*}
1 - \frac{\gamma^l}{\ttd^l} + \frac{\gamma^l}{\ttd^l} \ee^{-s^{\star} \ttd^l} &= 1 - \frac{\gamma^l}{\ttd^l} + \frac{\gamma^l}{\ttd^l} \frac{\bar{\delta} (\ttd^l - \gamma^l)}{\ttd^l - \bar{\delta} \gamma^l} \\
&= \frac{\ttd^l - \gamma^l}{\ttd^l - \bar{\delta} \gamma^l},
\end{align*}
so that $g(s^{\star})$ is equal to
\begin{equation*}
\frac{\ttd^l - \gamma^l}{\ttd^l - \bar{\delta} \gamma^l} \bigg( \frac{\ttd^l - \bar{\delta} \gamma^l}{\bar{\delta} (\ttd^l - \gamma^l)} \bigg)^{\bar{\delta} \gamma^l / \ttd^l} = \bar{\delta}^{-\bar{\delta} \gamma^l / \ttd^l} \bigg( \frac{\ttd^l - \gamma^l}{\ttd^l - \bar{\delta} \gamma^l} \bigg)^{1 - \bar{\delta} \gamma^l / \ttd^l},
\end{equation*}
completing the proof.
\end{proof}

Lemma~\ref{lem:gs} allows us to conclude from (\ref{eqn:s+bound}) and (\ref{eqn:K+bound+in+s}) that
\begin{align}
\label{eqn:star+Ak}
\PP \big( A_k \leqslant (\bar{\delta} \gamma^l - 1) k \big) &\leqslant \big( g(s^{\star}) \big)^k; \\
\label{eqn:star+K}
\PP(\kappa < K \leqslant k) &\leqslant \sum_{j=\kappa+1}^k \big( g(s^{\star}) \big)^j.
\end{align}
Here we remark that the $\delta$ in these two inequalities depends on $\cc$ and $(\lambda,\rho)$, since we require that $\bar{\delta} \gamma^l > 1$. In what follows we will show that $g(s^{\star})$ can be bounded from above by an exponentially decreasing function, which can then be used to get exponential bounds on the left-hand sides of (\ref{eqn:star+Ak}) and (\ref{eqn:star+K}). 

\begin{lemma}
\label{lem:exp+g+bound}
Let $\gamma = \gamma(\cc,\lambda,\rho)$ and $l = l(\cc,\lambda,\rho)$ be specified by Lemma~\ref{lem:bhatta+error} and let $\delta \in (0,1)$ be arbitrary and fixed. Then
\begin{equation}
\label{eqn:expo+gs}
\bar{\delta}^{-\bar{\delta} \gamma^l / \ttd^l} \bigg( \frac{\ttd^l - \gamma^l}{\ttd^l - \bar{\delta} \gamma^l} \bigg)^{1 - \bar{\delta} \gamma^l / \ttd^l} \leqslant \ee^{-\delta^2 \gamma^l / (2 \ttd^l)}.
\end{equation}
\begin{proof}
By monotonicity of the logarithm function $x \mapsto \ln(x)$, it is enough to show that
\begin{equation}
f(\delta) \coloneqq \frac{\gamma^l}{2 \ttd^l} \delta^2 - \frac{\gamma^l}{\ttd^l} \bar{\delta} \ln(\bar{\delta}) + \bigg( 1 - \frac{\gamma^l}{\ttd^l} \bar{\delta} \bigg) \ln\bigg( \frac{\ttd^l - \gamma^l}{\ttd^l - \gamma^l \bar{\delta}} \bigg) \leqslant 0.
\end{equation}
Since $f(0) = 0$, the claim in (\ref{eqn:expo+gs}) will follow if we can show that $f$ is non-increasing on $[0,1)$. Indeed, the first-order derivative of $f$ can be computed as
\begin{align*}
f^{\prime}(\delta) &= \frac{\gamma^l}{\ttd^l} \delta + \frac{\gamma^l}{\ttd^l} \big( 1 + \ln(\bar{\delta}) \big) + \frac{\gamma^l}{\ttd^l} \ln \bigg( \frac{\ttd^l - \gamma^l}{\ttd^l - \gamma^l \bar{\delta}} \bigg) - \frac{\gamma^l}{\ttd^l}  \\
&= \frac{\gamma^l}{\ttd^l} \Bigg( \delta + \ln\bigg( \frac{(\ttd^l - \gamma^l)(1-\delta)}{\ttd^l - \gamma^l(1-\delta)} \bigg) \Bigg)  \\
&= \frac{\gamma^l}{\ttd^l} \Bigg( \delta + \ln\bigg( 1 - \frac{\ttd^l \delta}{\ttd^l - \gamma^l(1-\delta)} \bigg) \Bigg).
\end{align*}
Since $\ln(1-x) \leqslant -x$ for all $x < 1$ and $\ttd^l \delta / (\ttd^l - \gamma^l (1-\delta)) < 1$, it follows that
\begin{align*}
\delta + \ln\bigg( 1 - \frac{\ttd^l \delta}{\ttd^l - \gamma^l(1-\delta)} \bigg) &\leqslant \delta - \frac{\ttd^l \delta}{\ttd^l - \gamma^l(1-\delta)}  \\
&\leqslant \delta - \frac{\ttd^l \delta}{\ttd^l} = \delta - \delta,
\end{align*}
from which we can conclude that $f^{\prime}(\delta) \leqslant 0$ for all $\delta \in (0,1)$, and this completes the proof.
\end{proof}
\end{lemma}


\subsection{Appendix to Section \ref{subsec:lbd+bit+error+prob}}
\label{appdix:bsc+sum}

\begin{proof}[Proof of Lemma \ref{lem:bsc+density}]
The inequality in (\ref{eqn:bsc+binomial}) holds for the case $k = 0$. Thus let us show that the inequality holds for $1 \leqslant k \leqslant (l+1)/2$. To this end, let us first define the quantity 
\begin{equation*}
D(k) = B\big( k; k+(l-1)/2 \big) - B\big( 0; (l-1)/2 \big), 
\end{equation*}
so that $D(k)$ is the difference between the right-hand side of (\ref{eqn:bsc+binomial}) and the left-hand side of (\ref{eqn:bsc+binomial}). Then we have
\begin{align}
D(k) &= \sum_{i=k}^{k+(l-1)/2} \binom{l}{i} p^{l-i} {(1-p)}^i - \sum_{i=0}^{(l-1)/2} \binom{l}{i} p^{l-i} {(1-p)}^i \nonumber \\
&= \sum_{i=(l+1)/2}^{k+(l-1)/2} \binom{l}{i} p^{l-i} {(1-p)}^i - \sum_{i=0}^{k-1} \binom{l}{i} p^{l-i} {(1-p)}^i \nonumber \\
\label{eqn:bsc+Dk}
&= B\big( (l+1)/2; k+(l-1)/2 \big) - B( 0; k-1).
\end{align}
With $\bar{p} \coloneqq 1 - p$, we notice that
\begin{align}
B\big( (l+1)/2; k+(l-1)/2 \big) &= \sum_{i = (l+1)/2 - k}^{(l-1)/2} \binom{l}{i} p^{l-i} {\bar{p}}^i \ee^{(l-2i) \ln(\bar{p}/p)} \nonumber \\
\label{eqn:bsc+channel+symmetry}
&\geqslant \sum_{i = (l+1)/2 - k}^{(l-1)/2} \binom{l}{i} p^{l-i} {\bar{p}}^i \\
&= B\big( (l+1)/2 - k; (l-1)/2 \big), \nonumber
\end{align}
where the inequality in (\ref{eqn:bsc+channel+symmetry}) is due to the fact that $(l-2i)\ln(\bar{p}/p)$ is positive as $p \in {[0,1/2]}$ and $l-2i \geqslant 1$. Therefore plugging this lower bound into (\ref{eqn:bsc+Dk}), we can obtain that
\begin{equation}
\label{eqn:bsc+Dk+more}
D(k) \geqslant B\big( (l+1)/2 - k; (l-1)/2 \big) - B( 0; k-1 ).
\end{equation}
Next we will verify that the sequence $\{B_i\}_{0 \leqslant i \leqslant l}$, the $i$-th element of which is defined by 
\begin{equation*}
B_i = \binom{l}{i} p^{l-i} (1-p)^i, 
\end{equation*}
is nondecreasing for $i \leqslant (l-1)/2$. Then the claim that $D(k)$ is nonnegative for all $0 \leqslant k \leqslant (l+1)/2$ follows, since $B( (l+1)/2 - k; (l-1)/2 )$ in (\ref{eqn:bsc+Dk+more}) is the sum over the last $k$ elements of the sequence $\{B_i\}_{0 \leqslant i \leqslant (l-1)/2}$, while $B( 0; k-1 )$ in (\ref{eqn:bsc+Dk+more}) is the sum over the first $k$ elements of the sequence $\{B_i\}_{0 \leqslant i \leqslant (l-1)/2}$. Now to verify that $\{B_i\}$ is nondecreasing, we can compute 
\begin{equation*}
\frac{B_i}{B_{i-1}} = \frac{(l-i+1)(1-p)}{ip} = 1 + \frac{(l+1)(1-p) - i}{ip}.
\end{equation*}
So the sequence $\{B_i\}$ is nondecreasing if $i \leqslant \lfloor (1-p)(l+1) \rfloor$. Since $p \in [0,1/2]$, we conclude that 
\begin{equation*}
\lfloor (1-p)(l+1) \rfloor \geqslant (1-p)(l+1) - 1 \geqslant (l+1)/2 - 1 = (l-1)/2, 
\end{equation*}
and this verifies that $\{B_i\}$ is nondecreasing when $i \leqslant (l-1)/2$, which completes the proof of Lemma~\ref{lem:bsc+density}.
\end{proof}


\bibliographystyle{ieeetr}
\bibliography{reference}


\end{document}